\documentclass[11pt,a4paper]{article}

\usepackage[utf8]{inputenc}
\usepackage[english]{babel}
\usepackage{siunitx}
\usepackage{subcaption}
\usepackage{float}
\usepackage{fancyhdr}
 \usepackage{chngcntr}
 \usepackage{cancel}
\usepackage{hyperref}
\usepackage{authblk}

\usepackage{quiver}
\usepackage{ytableau}

\usepackage{stackengine}

\usepackage{hyperref}
\usepackage{color}
\usepackage{amsthm}
\usepackage{mathtools}

\usepackage[normalem]{ulem} 
\usepackage{graphicx}
\usepackage{amssymb,latexsym,cite}
\usepackage{amsmath}
\usepackage{amsfonts}
\usepackage{mathrsfs}
\usepackage{bbm}
\usepackage{bm}
\usepackage[T1]{fontenc}

\usepackage[matrix,arrow,color]{xy}

\usepackage{enumitem}

\def\slasha#1{\setbox0=\hbox{$#1$}#1\hskip-\wd0\hbox to\wd0{\hss\sl/\/\hss}}

\def\periodb#1{\setbox0=\hbox{$#1$}#1\hskip-\wd0\hbox to\wd0{-}}

\newcommand{\ident}{\mathbbm{1}}   			
\newcommand{\id}{\mathrm{id}}   			
\newcommand{\ii}{\mathrm{i}}   			
\newcommand{\jj}{\mathsf{j}}   			
\newcommand{\e}{\mathrm{e}}   			

\newcommand{\CA}{\mathcal{A}}    			

\newcommand{\CB}{\mathcal{B}}

\newcommand{\CD}{\mathcal{D}}

\newcommand{\CF}{\mathcal{F}}

\newcommand{\CN}{\mathcal{N}}
\newcommand{\CCN}{\mathscr{N}}
\newcommand{\CO}{\mathcal{O}}

\newcommand{\CP}{\mathcal{P}}

\newcommand{\CQ}{\mathcal{Q}}

\newcommand{\CR}{\mathcal{R}}

\newcommand{\CCV}{\mathscr{V}}

\newcommand{\CZ}{\mathcal{Z}}

\newcommand{\CCW}{\mathscr{W}}
\newcommand{\CE}{\mathcal{E}}
\newcommand{\CCE}{\mathscr{E}}

\newcommand{\frg}{\mathfrak{g}}				

\newcommand{\frM}{\mathfrak{M}}

\newcommand{\ulfour}{\underline{\sf 4}}
\newcommand{\ulthree}{\underline{\sf 3}}

\DeclareMathOperator*{\Timesbig}{\scalerel*{\times}{\displaystyle\sum}}
\usepackage{scalerel}

\def\tv{{\textrm{\tiny $V$}}}
\def\tv1{{\textrm{\tiny $V[1]$}}}

\newcommand{\mbf}[1]{{\boldsymbol {#1} }}

\newcommand{\rk}{\mathrm{rk}}

\makeatletter
\newcommand*{\doublerightarrow}[2]{\mathrel{
  \settowidth{\@tempdima}{$\scriptstyle#1$}
  \settowidth{\@tempdimb}{$\scriptstyle#2$}
  \ifdim\@tempdimb>\@tempdima \@tempdima=\@tempdimb\fi
  \mathop{\vcenter{
    \offinterlineskip\ialign{\hbox to\dimexpr\@tempdima+1em{##}\cr
    \rightarrowfill\cr\noalign{\kern.1ex}
    \rightarrowfill\cr}}}\limits^{\!#1}_{\!#2}}}
\newcommand*{\triplerightarrow}[1]{\mathrel{
  \settowidth{\@tempdima}{$\scriptstyle#1$}
  \mathop{\vcenter{
    \offinterlineskip\ialign{\hbox to\dimexpr\@tempdima+1em{##}\cr
    \rightarrowfill\cr\noalign{\kern.5ex}
    \rightarrowfill\cr\noalign{\kern.5ex}
    \rightarrowfill\cr}}}\limits^{\!#1}}}
\makeatother

\newcommand{\FR}{\mathbbm{R}}     			
\newcommand{\FC}{\mathbbm{C}}     			
\newcommand{\NN}{\mathbbm{N}}     			
\newcommand{\RZ}{\mathbbm{Z}}     			
\newcommand{\FS}{\mathbbm{S}}     			
\newcommand{\PP}{{\mathbbm{P}}}    			

\newcommand{\dd}{\mathrm{d}}     			
\newcommand{\embd}{{\lhook\joinrel\longrightarrow}}     		
\newcommand{\diag}{{\mathrm{diag}}}     		

\newcommand{\ccdot}{{\,\cdot\,}}
 
\newcommand{\Ab}{{\sGamma_{\mathsf{ab}}}}    			
\newcommand{\Abw}{{\widehat\sGamma_{\mathsf{ab}}}}     			

\newcommand{\sU}{\mathsf{U}}     			
\newcommand{\sK}{\mathsf{K}}     			

\newcommand{\sA}{\mathsf{A}}
\newcommand{\sG}{\mathsf{G}}
\newcommand{\sT}{\mathsf{T}}

\newcommand{\sHom}{\mathsf{Hom}}

\newcommand{\sInd}{\textsf{Ind}}

\newcommand{\sH}{\mathsf{H}}
\newcommand{\sSU}{\mathsf{SU}}
\newcommand{\sSL}{\mathsf{SL}}
\newcommand{\sGL}{\mathsf{GL}}

\newcommand{\sE}{\mathsf{E}}

\newcommand{\sD}{\mathsf{D}}
\newcommand{\sC}{\mathsf{C}}

\newcommand{\sGamma}{{\mathsf{\Gamma}}}
\newcommand{\sSigma}{{\mathsf{\Sigma}}}
\newcommand{\sUps}{{\mathsf{\Upsilon}}}

\newcommand{\sSO}{\mathsf{SO}}

\newcommand{\sEnd}{\mathsf{End}}

\newcommand{\comment}[1]{}     				
\def\tyng(#1){\hbox{\tiny$\yng(#1)$}}			
\def\tyoung(#1){\hbox{\tiny$\young(#1)$}}			

\newcommand{\beq}{\begin{eqnarray}}
\newcommand{\eeq}{\end{eqnarray}}

\newcommand{\sft}{{\sf t}}

\newcommand{\Hilb}{{\sf Hilb}}
\newcommand{\Ob}{{\sf Ob}}

\newcommand{\sfa}{\mathsf{a}}

\newcommand{\sfs}{\mathsf{s}}

\newcommand{\sfi}{{\mathsf{i}}}
\newcommand{\sfj}{{\mathsf{j}}}

\newcommand{\sfR}{\mathsf{R}}

\newcommand{\sfS}{\mathsf{S}}

\definecolor{outrageousorange}{rgb}{1.0, 0.43, 0.29}

\newenvironment{myitemize}{\begin{itemize}[itemsep=-0.05cm, leftmargin=*, topsep=0.1cm]}{\end{itemize}}

\newcommand{\qu}{\texttt{q}}

\newcommand{\Qu}{\texttt{Q}}

\newcommand{\ttF}{\texttt{F}}
\newcommand{\ttI}{\texttt{I}}

\newcommand{\Tr}{\mathrm{Tr}}

\newcommand{\rmB}{\mathrm{B}}

\theoremstyle{definition}
\newtheorem{theorem}[equation]{Theorem}

\newtheorem{proposition}[equation]{Proposition}

\newtheorem{remark}[equation]{Remark}
\newtheorem{example}[equation]{Example}

\newtheorem{notation}[equation]{Notation}

\newcommand{\midwedge}{\text{\Large$\wedge$}}

\def\beq{\begin{equation}}
\def\bee{\begin{equation}}
\def\eeq{\end{equation}}
\def\bea{\begin{eqnarray}}
\def\eea{\end{eqnarray}}
\def\ba{\begin{align}}
\def\ea{\end{align}}

\numberwithin{equation}{section}
\catcode`@=12

\newcommand{\ttQ}{{\sf Q}}

\newcommand{\tta}{\texttt{a}}  
\newcommand{\Quot}{{\sf Quot}}
\newcommand{\ch}{\mathrm{ch}}

\newcommand{\tts}{\texttt{s}}

\newcommand{\ttl}{\texttt{l}}

\usepackage{tikz}
\newcounter{x}
\newcounter{y}
\newcounter{z}

\newcommand\xaxis{210}
\newcommand\yaxis{-30}
\newcommand\zaxis{90}

\newcommand\topside[3]{
  \fill[fill=white, draw=black,shift={(\xaxis:#1)},shift={(\yaxis:#2)},
  shift={(\zaxis:#3)}] (0,0) -- (30:1) -- (0,1) --(150:1)--(0,0);
}

\newcommand\leftside[3]{
  \fill[fill=white, draw=black,shift={(\xaxis:#1)},shift={(\yaxis:#2)},
  shift={(\zaxis:#3)}] (0,0) -- (0,-1) -- (210:1) --(150:1)--(0,0);
}

\newcommand\rightside[3]{
  \fill[fill=white, draw=black,shift={(\xaxis:#1)},shift={(\yaxis:#2)},
  shift={(\zaxis:#3)}] (0,0) -- (30:1) -- (-30:1) --(0,-1)--(0,0);
}

\newcommand\cube[3]{
  \topside{#1}{#2}{#3} \leftside{#1}{#2}{#3} \rightside{#1}{#2}{#3}
}

\newcommand\topsideg[3]{
  \fill[fill=black!10, draw=black,shift={(\xaxis:#1)},shift={(\yaxis:#2)},
  shift={(\zaxis:#3)}] (0,0) -- (30:1) -- (0,1) --(150:1)--(0,0);
}

\newcommand\leftsideg[3]{
  \fill[fill=black!10, draw=black,shift={(\xaxis:#1)},shift={(\yaxis:#2)},
  shift={(\zaxis:#3)}] (0,0) -- (0,-1) -- (210:1) --(150:1)--(0,0);
}

\newcommand\rightsideg[3]{
  \fill[fill=black!10, draw=black,shift={(\xaxis:#1)},shift={(\yaxis:#2)},
  shift={(\zaxis:#3)}] (0,0) -- (30:1) -- (-30:1) --(0,-1)--(0,0);
}

\newcommand\cubeg[3]{
  \topsideg{#1}{#2}{#3} \leftsideg{#1}{#2}{#3} \rightsideg{#1}{#2}{#3}
}

\newcommand\planepartition[1]{
 \setcounter{x}{-1}
  \foreach \a in {#1} {
    \addtocounter{x}{1}
    \setcounter{y}{-1}
    \foreach \b in \a {
      \addtocounter{y}{1}
      \setcounter{z}{-1}
      \foreach \c in {1,...,\b} {
        \addtocounter{z}{1}
        \cube{\value{x}}{\value{y}}{\value{z}}
      }
    }
  }
}

\newcommand\planepartitiong[1]{
 \setcounter{x}{-1}
  \foreach \a in {#1} {
    \addtocounter{x}{1}
    \setcounter{y}{-1}
    \foreach \b in \a {
      \addtocounter{y}{1}
      \setcounter{z}{-1}
      \foreach \c in {1,...,\b} {
        \addtocounter{z}{1}
        \cubeg{\value{x}}{\value{y}}{\value{z}}
      }
    }
  }
}

\begin{document}

\title{\bf Tetrahedron Instantons on Orbifolds}

\author{Richard J. Szabo\thanks{R.J.Szabo@hw.ac.uk} \ }
\author{ Michelangelo Tirelli\thanks{mt2001@hw.ac.uk}}

\affil{\textit{\normalsize Department of Mathematics,
Heriot–Watt University}\\ \vspace{-1mm}
\textit{\normalsize Colin Maclaurin Building, Riccarton, Edinburgh EH14 4AS, UK}\\ \vspace{1mm}
\textit{\normalsize Maxwell Institute for Mathematical Sciences, Edinburgh, UK}} \date{}
\maketitle

\vspace{1cm}

\begin{abstract}
\noindent
Given a homomorphism $\tau$ from a suitable finite group $\sGamma$ to $\sSU(4)$ with image $\sGamma^\tau$, we construct a cohomological gauge theory on a noncommutative resolution of  the quotient singularity $\FC^4/\sGamma^\tau$ whose BRST fixed points are $\sGamma$-invariant tetrahedron instantons on a generally non-effective orbifold.
The partition function computes the expectation values of complex codimension one defect operators in rank $r$ cohomological Donaldson--Thomas theory on a flat gerbe over the quotient stack $[\FC^4/\,\sGamma^\tau]$.
We describe the generalized ADHM parametrization of the tetrahedron instanton moduli space, and evaluate the orbifold partition functions through virtual torus localization. If $\sGamma$ is an abelian group the partition function is expressed as a  combinatorial series over arrays of $\sGamma$-coloured plane partitions, while if $\sGamma$ is non-abelian the partition function localizes onto a sum over torus-invariant connected components of the moduli space labelled by lower-dimensional partitions. When $\sGamma=\RZ_n$ is a finite abelian subgroup of $\sSL(2,\FC)$, we exhibit the reduction of Donaldson--Thomas theory on the toric Calabi--Yau four-orbifold $\FC^2/\,\sGamma\times\FC^2$ to the cohomological field theory of tetrahedron instantons, from which we express the partition function  as a closed infinite product formula. We also use the crepant resolution correpondence to derive a closed formula for the partition function on any polyhedral singularity.
\end{abstract}

\newpage
 
{\baselineskip=14pt
\tableofcontents
}

\newpage

\section{Introduction}

\subsubsection*{Background}

Tetrahedron instantons~\cite{Pomoni:2021hkn,Pomoni:2023nlf,Fasola:2023ypx,Kimura:2023bxy} are particular solutions of generalized instanton equations in eight dimensions. They are defined by BRST fixed point equations for a generalized cohomological gauge theory on a singular stratification of spacetime which glues together different quantum field theories through real codimension two supersymmetric defects; the gluing is mediated by bifundamental matter fields on the codimension four junctions formed by intersections of the strata. These generalize the spiked instantons
introduced by Nekrasov \cite{Nekrasov:2015wsu} as extensions of instanton configurations from four dimensions to include the most general supersymmetric local and surface defects. They can be regarded as an intermediary step between instanton solutions in six and eight dimensions, thereby linking six- and eight-dimensional cohomological gauge theories. The premise is that one can recover them from eight-dimensional field configurations through certain specialisations of the moduli, analogously to how the six-dimensional theories are obtained from eight dimensions. 

Similarly to spiked instantons \cite{Nekrasov:2016qym,Nekrasov:2016gud}, tetrahedron instantons find their physical realization in type~IIB string theory as bound states of D1-branes probing configurations of intersecting stacks of D-branes which wrap smooth strata of a singular three-fold inside a local Calabi--Yau four-fold $M$, while preserving a suitable number of supersymmetries. In this paper we focus mostly (but not exclusively) on the case $M=\FC^4$, where intersecting D7-branes span the four complex codimension one coordinate hyperplanes in $\FC^4$, with an appropriate constant Neveu--Schwarz $B$-field turned on. These have a description as solutions to noncommutative instanton equations in the presence of the most general complex codimension one supersymmetric defects (see~\cite{Szabo:2022zyn} for a review of spiked and tetrahedron instantons in noncommutative field theory). Given a single stack of D7-branes, we regard its bound state with the D1-branes as a noncommutative instanton of the gauge theory. The remaining D7-branes with different spatial orientations then generate defects in its worldvolume theory. The moduli space of tetrahedron instantons is isomorphic to a Grothendieck Quot scheme which parametrizes quotients of a torsion sheaf on the possibly singular three-fold formed by the union of the hyperplanes $\FC^3$ in the Calabi--Yau four-fold $\FC^4$ \cite{Henni:2017,Cazzaniga:2020xru,Pomoni:2021hkn,Fasola:2023ypx}.

Generally, BPS state counting in six and eight dimensions is related to Donaldson--Thomas theory which enumerates virtual invariants of moduli spaces of coherent sheaves; more generally, Donaldson--Thomas invariants count objects in Calabi--Yau categories, where the relevant category in the former case is the derived category of coherent sheaves. From the perspective of cohomological gauge theory, the moduli space of $\sU(r)$ instantons on a toric background $M$ is compactified by deforming the BPS equations to instanton equations in noncommutative field theory and by introducing an $\Omega$-deformation of $M$. This enables evaluation of the instanton partition function exactly through virtual toric localization by reducing the path integral of the cohomological gauge theory to an equivariant integral over the instanton moduli space. It localizes onto isolated torus fixed points of the moduli space which are in one-to-one correspondence with higher-dimensional partitions \cite{Iqbal:2003ds,Cirafici:2008sn,m4,Kanno:2020ybd,Pomoni:2021hkn,Szabo:2022zyn}. An important role in this computation is played by the definition of suitable virtual fundamental classes through an obstruction theory defined by integration over antighost fields.

In eight dimensions, the Donaldson--Thomas invariants have been studied from the perspective of generalized instanton counting and the related BPS state counting of D-branes in \cite{m4,m4c,Bonelli:2020gku,Kimura:2022zsm,Szabo:2023ixw,Piazzalunga:2023qik, Nekrasov:2023nai,Kimura:2023bxy,Galakhov:2023vic,Franco:2023tly,Bao:2024ygr,Kimura:2024xpr}. The compactification of the instanton moduli space in this case results in a maximal holonomy group $\sSU(4)$ and sets the cohomological gauge theory  on a toric Calabi--Yau four-fold.
The virtual cycles for the Donaldson--Thomas invariants are constructed in gauge theory by Cao and Leung~\cite{Cao:2014bca}, in derived differential geometry by Borisov and Joyce \cite{Borisov:2015vha}, as well as in algebraic geometry by Oh and Thomas~\cite{Oh:2020rnj}. These cycles depend on a choice of local orientations of the moduli space, requiring a selection of signs that enter into the computation of the partition function. The choice is unique up to overall orientation; it was conjectured by Nekrasov and Piazzalunga~\cite{m4c} for instanton counting on $\FC^4$, and subsequently proven by Kool and Rennemo~\cite{KRinprep} for the Donaldson--Thomas theory of $\FC^4$. Related mathematical developments of Donaldson--Thomas invariants on Calabi--Yau four-folds are found in e.g.~\cite{Cao:2017swr,Cao:2019tvv,Bojko:2020rfg,KiemPark20,Park:2021hnu, Monavari:2022rtf,Cao:2023gvn,Cao:2023lon,Liu23,Bae:2024bpx,Cao:2024sln}.
An adaptation of the proof of~\cite{KRinprep} is presented by Fasola and Monavari for tetrahedron instantons in \cite{Fasola:2023ypx}, where the instanton partition function computes expectation values of codimension one defect operators in the Donaldson--Thomas theory of $\FC^3$.

Our current understanding of instanton counting in six and eight dimensions, as well as its relation to Donaldson--Thomas theory, is limited to abelian configurations. In both dimensionalities the matrix equations that result from noncommutative $\sU(r)$ instanton equations~\cite{Cirafici:2008sn,Szabo:2023ixw} contain more degrees of freedom (and equations) than what appear in the generalized ADHM equations from the D-brane picture or in the noncommutative Quot scheme construction, unless one restricts to solutions in the maximal torus $\sU(1)^r\subset\sU(r)$ in which case stability implies that the extra operators vanish. Thus our higher rank computations are limited to \emph{Coulomb branch invariants}, for which the non-abelian $\sU(r)$ gauge symmetry is broken to the abelian subgroup $\sU(1)^r$, mirroring the geometric property that the framed Quot schemes which are well-defined in these dimensions parametrize only split vector bundles (see~\cite[Corollary~1.6]{Cazzaniga:2020xru}). 

The geometric moduli problem associated to genuine non-abelian instanton counting is not currently understood, nor how to compute invariants as the standard equivariant localization techniques no longer apply. When $M$ is a Calabi--Yau three-fold, the Coulomb branch invariants of~\cite{Cirafici:2008sn} are interpreted by~\cite{Stoppa:2012sf} as a degenerate central charge limit of higher rank Donaldson--Thomas invariants for pure D0--D6 bound states, which enumerate rank $r$ torsion free sheaves on $M$ that are locally free in codimension three. Higher rank Donaldson--Thomas invariants for $\sU(r)$ gauge theory on any toric three-fold $M$ are constructed from M-theory considerations by~\cite{DelZotto:2021gzy}.

\subsubsection*{This Paper}

As a first extension of the original model of~\cite{Pomoni:2021hkn} beyond flat space, in this paper we provide a detailed and exhaustive analysis of tetrahedron instantons on local Calabi--Yau orbifolds of $\FC^4$ (when they exist). We extend the computations for spiked instantons on orbifolds in \cite{Nekrasov:2016qym,Nekrasov:2016ydq} to evaluate partition functions for tetrahedron instantons defined on orbifolds $\FC^4/\,\sGamma$, where $\sGamma$ is a suitable finite group whose action on $\FC^4$ is defined by a homomorphism $\tau:\sGamma\longrightarrow \sSU(4)$ to the holonomy group $\sSU(4)$. The choice of a general homomorphic image $\sGamma^\tau$ rather than a subgroup embedding of $\sGamma$ in $\sSU(4)$ allows for more freedom in a description of broader classes of stable ground states, and technically it enables the application of the virtual localization formula, even when $\sGamma$ is non-abelian. 

When the kernel $\sK^\tau\subset\sGamma$ of $\tau$ is non-trivial, the group $\sGamma$ acts \emph{non-effectively} on $\FC^4$, i.e. it contains a non-trivial subgroup which acts trivially on $\FC^4$. Nevertheless, the subgroup $\sK^\tau$ can still act non-trivially on the field content of the cohomological gauge theory. This sets the field theory on a $\sK^\tau$-gerbe over the quotient stack $[\FC^4/\,\sGamma^\tau]$ and is equivalent to a twist of the theory on a disjoint union of several copies of $[\FC^4/\,\sGamma^\tau]$. We interpret the corresponding enumerative invariants of the quotient singularity $\FC^4/\,\sGamma^\tau$ as the orbifold Donaldson--Thomas invariants `twisted' by a $\sK^\tau$-gerbe. The gerbe may be viewed as a flat $B$-field and the theory enumerates $\sK^\tau$-projectively $\sGamma^\tau$-equivariant coherent sheaves on $\FC^4$, which correspond to boundary states of D-branes supporting twisted Chan--Paton gauge bundles.

Our computations produce the instanton partition function of the cohomological gauge theory on a noncommutative resolution of the quotient singularity $\FC^4/\,\sGamma^\tau$, described by a certain noncommutative algebra $\mathsf{A}$. The algebra $\mathsf{A}$ is the path algebra of a generalization of the bounded McKay quiver determined by the representation theory data of $\sGamma$ together with the homomorphism $\tau$, whose relations provide  a generalized ADHM parametrization of the orbifold noncommutative tetrahedron instanton equations. The gauge theory is then defined by projecting onto the $\sGamma$-invariant field configurations on $\FC^4$, whose instanton moduli space is identified as a quiver variety associated to the generalized McKay quiver, or equivalently as the moduli space of stable framed representations for the bounded derived category of the McKay quiver. This bridges the cohomological gauge theories for orbifold instantons in six and eight dimensions, considered for the case of toric Calabi--Yau orbifolds in~\cite{Cirafici:2010bd} and~\cite{Szabo:2023ixw} respectively.

When $\sGamma$ is abelian, the image $\sGamma^\tau$ of $\tau$ commutes with the maximal torus $\sT_{\vec \epsilon}\,$ of the holonomy group $\sSU(4)$. Consequently, through toric localization, the equivariant partition function localizes onto isolated fixed points of the $\sT_{\vec \epsilon}\,$-action which are also $\sGamma$-invariant. The orbifold partition functions in this case describe the twisted orbifold Donaldson--Thomas theory of $\FC^3/\,\sGamma^\tau$ in the presence of general codimension one defects which are invariant under the maximal toric symmetry of the $\Omega$-deformation.

The case where $\sGamma$ is non-abelian presents some technical complications, as $\sGamma^\tau$ does not commute with $\sT_{\vec \epsilon}\,$ and it is necessary to work with the centralizer of $\sGamma^\tau$ in $\sT_{\vec \epsilon}\,$ in order to apply torus localization.  The gauge theory is then equivariant with respect to a smaller torus, and torus localization only reduces the partition function to a sum over contributions from the connected components of the moduli space of torus-invariant tetrahedron instantons, which generally admit continuous deformations, i.e. the torus fixed points are no longer isolated. We demonstrate that the partition function is still well-defined in these instances by proving that these components are compact in their natural complex analytic topology inherited from the ADHM parametrization, and we describe how to compute it.
The orbifold partition functions in these cases again describe the twisted orbifold Donaldson--Thomas theory of $\FC^3/\, \sGamma^\tau$, with or without a single codimension one defect and with reduced toric symmetry.

In both abelian and non-abelian cases, in addition to the generalizations to twisted orbifold Donaldson--Thomas invariants, another novelty of our approach that it is general enough to deal with orbifolds by arbitrary finite subgroups $\sGamma^\tau\subset\sU(3)$, and hence it computes the (twisted) Donaldson--Thomas theory of general local K\"ahler three-orbifolds.

\subsubsection*{Outline and Summary of Results}

In the following sections we shall begin with a review and extension of the pertinent cohomological gauge theories in six dimensions, which are then naturally extended to the field theories whose BPS states are tetrahedron instantons in eight dimensions. The structure of the remainder of this paper and its main results are summarised as follows:

\begin{myitemize}

\item In \underline{Section~\ref{sec:inst_3d}} we review the construction of a six-dimensional cohomological gauge theory for the holonomy group $\sU(3)$, following~\cite{Cirafici:2008sn} (see also~\cite{Cirafici:2012qc}). We study the generalized instanton equations and we evaluate the equivariant instanton partition function from the tangent-obstruction deformation complex of the instanton moduli space. It is expressed as a combinatorial expansion in plane partitions which can be summed to a closed form in terms of the MacMahon function.

\item In \underline{Section~\ref{sec:C3orbifold}} we analyse instanton configurations on orbifolds $\FC^3/\,\sGamma$, where $\sGamma$ is a finite group acting on $\FC^3$ by a homomorphism $\tau:\sGamma \longrightarrow \sU(3)$ to the holonomy group $\sU(3)$, vastly generalizing the treatment for toric Calabi--Yau three-orbifolds considered in~\cite{Cirafici:2010bd} (see also~\cite{Cirafici:2011cd,Cirafici:2012qc}). We describe the instanton moduli space as a quiver variety through an ADHM-type parametrization. In the case where $\sGamma$ is an abelian group, we evaluate the equivariant instanton partition functions explicitly as combinatorial series over $\sGamma$-coloured plane partitions.
 
\item In \underline{Section~\ref{sec:tetra_insta}} we study tetrahedron instantons. We construct their ADHM equations in analogy with the six-dimensional case. We evaluate the instanton partition function from both a quiver matrix model for the ADHM data and from the tangent-obstruction deformation complex of the instanton moduli space. Using the ADHM matrix model we further recover the partition function for tetrahedron instantons from the equivariant partition function for instantons on $\FC^4$, considered in~\cite{Szabo:2023ixw}, after a suitable specialisation of variables. Using this relation we derive a closed formula for the tetrahedron instanton partition function in terms of the MacMahon function which agrees with the generating functions computed by \cite{Pomoni:2023nlf,Fasola:2023ypx}. 

\item In \underline{Section~\ref{sec:Tetra_orb}} we generalize our discussion to tetrahedron instantons on orbifolds $\FC^4/\,\sGamma$ with a homomorphism $\tau:\sGamma\longrightarrow\sSU(4)$.  For orbifolds of the type $\FC^2/\,\sGamma\times \FC^2$, where $\sGamma=\RZ_n$ is a finite abelian subgroup of $\sSL(2,\FC)$, we extend the relation between the equivariant partition functions for tetrahedron instantons and instantons in eight dimensions, and hence derive closed formulas for the orbifold tetrahedron instanton partition functions in terms of MacMahon functions based on the results from \cite{Szabo:2023ixw} for instantons on local toric Calabi--Yau four-orbifolds. When $\sGamma$ is a finite abelian subgroup of $\sSL(3,\FC)$, we recover the instanton partition functions for local toric Calabi--Yau three-orbifolds $\FC^3/\,\sGamma$ with $\sU(3)$ holonomy. 

Finally, we consider the case of a finite non-abelian orbifold group $\sGamma$, with a generally non-faithful representation in $\sSU(4)$. We discuss in detail the two admissible classes of $\sGamma$-actions which permit the application of virtual localization techniques, and show that the torus-invariant connected components of the moduli space are parametrized respectively by linear partitions and integer points. We compute, for each case, the equivariant orbifold partition functions for tetrahedron instantons. We explain how to explicitly unravel the formulas for Kleinian singularities in $\FC^4$ using geometric crepant resolution techniques, and we derive a closed formula in terms of MacMahon functions for any polyhedral singularity.

\item In \underline{Section~\ref{sec:discussion}} we recapitulate our findings and comment on the physical and mathematical relevance of our results.

\item Two appendices at the end of the paper contain some technical results complementing the analysis of the main text. In \underline{Appendix~\ref{app:B}} we summarize the classification of the finite subgroups of $\sSU(3)$, which play a prominent role throughout this paper. In \underline{Appendix~\ref{app:compact}} we prove that, for the smaller tori $\sT'\subset \sT_{\vec \epsilon}\,$ which act on our theories, the \smash{$\sT'$-fixed} components of the moduli space for orbifold tetrahedron instantons are compact in the natural complex analytic topology inherited from the ADHM parametrization. 
\end{myitemize}

\subsubsection*{Acknowledgements}

We thank Michele Cirafici, Thomas Grimm, Martijn Kool, Sergej Monavari, Erik Plauschinn and Nicol\`o Piazzalunga for helpful discussions and correspondence. This article is based upon work from COST Actions CaLISTA CA21109 and THEORY-CHALLENGES CA22113 supported by COST (European Cooperation in Science and Technology). The work of {\sc M.T.} is supported by an EPSRC Doctoral Training Partnership grant. 

\section{Donaldson--Thomas Theory on K\"ahler Three-Folds}\label{sec:inst_3d}

In this section we review the computation of Donaldson--Thomas invariants of a K\"ahler three-fold from the perspective of instanton counting in a six-dimensional cohomological gauge theory. This sets the stage and notations for all subsequent computations in this paper.

\subsection{$\sU(3)$-Instanton Equations}
\label{subsec:U3instantons}

Let $(M_3,\omega)$ be a K\"ahler three-fold. We define a cohomological gauge theory on $M_3$ through a topological twist of the maximally supersymmetric $\CN=2$ Yang--Mills theory in six dimensions. It can be obtained by dimensional reduction from ten-dimensional $\CN = 1$ supersymmetric Yang--Mills theory on $M_3$ with gauge group $\sU(r)$ and holonomy group $\sU(3)$. The bosonic field content is valued in the adjoint representation of $\sU(r)$ and consists of a $\sU(r)$ gauge connection $\CA$ with curvature two-form \smash{$\CF =\nabla_{\!\CA}^2= \dd \CA+\CA\wedge \CA$}, which we assume has vanishing first Chern class, as well as a $(3, 0)$-form $\varphi$ and a complex Higgs field $\Phi$. We denote the associated covariant derivatives with a subscript ${}_\CA$.

The path integral of the gauge theory localizes onto solutions of BRST fixed point equations known as generalized instanton equations. They are given by~\cite{Baulieu:1997jx,Iqbal:2003ds,Cirafici:2008sn}
\begin{align}\label{eq:inst_3d}
\begin{split}
\CF^{2,0} + \bar{\partial}_\CA^{\,\dag}\,\varphi &= 0 \ , \\[4pt]
\omega\wedge\omega\wedge \CF^{1,1} + \varphi\wedge\bar{\varphi} &= 0 \ , \\[4pt]
\nabla_{\!\CA}\,\Phi&=0 \ .
\end{split}
\end{align}
Here $\CF=\CF^{2,0}+\CF^{1,1}+\CF^{0,2}$ is the decomposition of the field strength in the basis of $(1,0)$- and $(0,1)$-forms with respect to the underlying complex structure of $M_3$. 

When $M_3$ is a Calabi--Yau three-fold,  the holonomy group is reduced to $\sfS\sU(3)\subset\sU(3)$ and uniqueness of the holomorphic three-form of the $\sSU(3)$-structure implies $\varphi=0$ in \eqref{eq:inst_3d}. Then the first two instanton equations reduce to the Donaldson--Uhlenbeck--Yau equations which describe stable holomorphic vector bundles on $M_3$ with finite characteristic classes.

The finite action solutions of \eqref{eq:inst_3d} are labelled by the third Chern class 
\begin{align}
k=\frac{1}{48\pi^3}\,\int_{M_3}\,\Tr_{\mathfrak{u}(r)} \,\CF\wedge \CF\wedge \CF \ ,
\end{align}
which is a topological invariant called the \emph{instanton number}, as well as K\"ahler charges determined by the second Chern class which we suppress. For each charge $k\in\RZ_{\geq0}$ we define the instanton moduli space $\frM_{r,k}$. They form the connected components of the stratification of the moduli space
\begin{align}
\frM_{r}=\bigsqcup_{k\geq0}\,\frM_{r,k}
\end{align}
of solutions $\CA$ to the $\sU(r)$ instanton equations \eqref{eq:inst_3d} modulo gauge transformations. The moduli space has a global colour symmetry under ${\sf P}\sU(r) = \sU(r)/\sU(1)$, where $\sU(1)$ is the center of $\sU(r)$.

\subsection{ADHM Data}
\label{subsec:ADHMC3}

The BPS equations \eqref{eq:inst_3d} on the affine K\"ahler three-fold $M_3=\FC^3$ describe the low-energy interactions of $k$ D0-branes inside $r$ D6-branes
in type IIA string theory in the limit where the D6-branes are heavy. From the perspective of the theory on the D0-branes, bound states corresponding to supersymmetric vacua are solutions to certain quadratic matrix equations, generalizing the celebrated ADHM equations~\cite{Atiyah:1978ri}, deformed by a Fayet--Iliopoulos coupling $\zeta\in\FR_{>0}$ related to a suitable large non-zero constant background $B$-field~\cite{Boundstates}. They arise as F-term and D-term equations. The Neveu--Schwarz $B$-field induces a noncommutative deformation of the gauge theory on the D6-branes obtained by Berezin--Toeplitz quantization of the constant Poisson structure $\theta=\zeta\,\omega^{-1}$~\cite{Cirafici:2012qc}.

\subsubsection*{Generalized ADHM Equations}

Let $V$ and $W$ be Hermitean vector spaces of complex dimensions $k$ and $r$ respectively; from the perspective of the D0-branes, $V$ is the Chan--Paton space while $W$ is a flavour representation. Then the ADHM equations are
\begin{align}\label{eq:ADHM3d}\begin{split}
\mu^{\FC}_{ab}&:= \, [ B_{a}, B_{b} ]- \tfrac{1}{2}\,\varepsilon_{abc}\,[B_c^\dagger,   Y ] = 0  \ , \\[4pt]
\mu^{\FR}&:=\sum_{a\,\in\,\ulthree}\, [B_{a}\,B_{a}^\dagger] +[ Y^\dagger \, ,  Y]+ I\,I^\dagger= \zeta\,  \ident_{V}  \ , \\[4pt]
\sigma&:=I^\dagger\,  Y=0 \ ,
\end{split}
\end{align}
where $B_a, Y\in \sEnd_\FC(V)$ for
\begin{align} 
a \ \in \ \ulthree:=\{1,2,3\} \qquad \mbox{and} \qquad (a,b) \ \in \ \ulthree^\perp:=\big\{(1,2)\,,\,(1,3)\,,\,(2,3)\big\} \ ,
\end{align}
while $I\in\sHom_\FC(W,V)$. Here $\varepsilon_{abc}$ is the Levi--Civita symbol in three dimensions with $\varepsilon_{123}=+1$, and throughout implicit summation over repeated indices is assumed unless otherwise explicitly indicated.

The ADHM equations~\eqref{eq:ADHM3d} are invariant under  the natural action by unitary automorphisms $g\in\sU(V)\simeq\sU(k)$ of the vector space $V$ given by
\begin{align} \label{eq:ADHMUVaction}
 g\cdot (B_{a}, Y,I)_{a\,\in\,\ulthree} = (g\,B_a\, g^{-1}\,, \, g\,  Y\, g^{-1}\,, \, g\,I)_{a\,\in\,\ulthree}  \ .
 \end{align}
The instanton moduli space $\frM_{r,k}$ is then equivalently described as the quotient by this $\sU(V)$-action of the subvariety of the affine space of ADHM data cut out by the equations \eqref{eq:ADHM3d}.
 There is additionally a natural action on the moduli space by unitary automorphisms $h\in\sU(W)\simeq\sU(r)$ of the vector space $W$ given by framing rotations
 \begin{align}
 h\cdot (B_a, Y,I)_{a\,\in\,\ulthree}  = (B_a\,,\, Y\,,\,I\,h^{-1})_{a\,\in\,\ulthree}  \ .
 \end{align}

\subsubsection*{Stability and Quot Schemes}

By standard arguments the second equation of \eqref{eq:ADHM3d} (the D-term relation) is equivalent to the following \emph{stability condition}: there is no proper subspace $S\subset V$ such that $B_a(S) \subset S$ for all $a\in\ulthree\,$, $Y^\dagger(S) \subset S$ and $\mathrm{im}(I)\subset S$. 

We write
\begin{align}
\|T\|_{\textrm{\tiny F}}^2 := \Tr_{U_2}\big(T^\dag\,T\big) = \Tr_{U_1}\big(T\,T^\dag\big)
\end{align}
for the Frobenius norm of a linear map $T\in\sHom_\FC(U_1,U_2)$ between Hermitian vector spaces $U_1$ and $U_2$. Then
\begin{align}
\sum_{(a,b)\,\in\,\ulthree^\perp} \, \big\|\mu_{ab}^\FC \big\|^2_{\textrm{\tiny F}} \, = \, \frac12\,\sum_{(a,b)\,\in\,\ulthree^\perp}\,  \big\| [B_a,B_b]\big\|_{\textrm{\tiny F}}^2 + \frac12\, \sum_{a\,\in\,\ulthree}\,  \big\| [B_a,Y^\dagger]\big\|_{\textrm{\tiny F}}^2 \ .
\end{align}
This vanishes by the first equation of \eqref{eq:ADHM3d}, which implies
\begin{align}\label{eq:commutingB}
[B_a,B_b]=0 \qquad \mbox{and} \qquad [B_a,Y^\dagger]=0  \ ,
\end{align}
{for} all $ a,b\in\ulthree\,$. 

Using the relations \eqref{eq:commutingB} and the third equation of \eqref{eq:ADHM3d}, the stability condition is thus equivalent to
\begin{align} \label{eq:stabilityC3}
V=\FC[B_1,B_2,B_3]\,I(W) \ .
\end{align}
This implies, by the first equation of \eqref{eq:commutingB} and the third equation of \eqref{eq:ADHM3d}, that $Y^\dagger=0$.
If we denote \smash{${\mu}^\FC:=(\mu_{ab}^\FC)_{(a,b)\,\in\,\ulthree^\perp}$}, then  the instanton moduli space $\frM_{r,k}$ is equivalently expressed as the noncommutative Quot scheme
\begin{align}
\frM_{r,k} \ \simeq \ {\mu}^{\FC\,-1}(0)^{\rm stable}\, \big/ \,\sGL(V) \ ,
\end{align}
where the superscript ${}^{\rm stable}$ indicates the stable solutions of the first equation of \eqref{eq:ADHM3d} with $ Y=0$ (the F-term relations), and $g\in\sGL(V)\simeq\sGL(k,\FC)$ acts on the ADHM data as in \eqref{eq:ADHMUVaction}.

It now follows from~\cite{Cazzaniga:2020xru} that the  instanton moduli space $\frM_{r,k}$ is isomorphic to the Quot scheme $\Quot_r^k(\FC^3)$ of zero-dimensional quotients of the free sheaf $\CO_{\FC^3}^{\oplus r}$ on $\FC^3$ with length $k$,
\begin{align}\label{eq:quot3d}
\frM_{r,k} \, \simeq \, \Quot_r^k(\FC^3) \ ,
\end{align}
which parametrizes framed torsion free sheaves $\CE$ on complex projective space $\PP^3$ of rank $r$ and $\ch_3(\CE)=k$. When $r=1$ the quotients are structure sheaves of closed zero-dimensional subschemes of $\FC^3$, and in this case the Quot scheme is the Hilbert scheme $\Hilb^k(\FC^3)$ of $k$ points on $\FC^3$.

\subsection{Tangent-Obstruction Theory}
\label{subsec:obC3}

The local geometry of the instanton moduli space $\frM_{r,k}$ is described
by the instanton deformation complex~\cite{Baulieu:1997jx}
\begin{align}\label{eq:defcomplex}
\midwedge^0\,T^*M_3\otimes\frg \xrightarrow{ \ C \ } \begin{matrix} \midwedge^{0,1}\,T^*M_3\otimes\frg \\[1ex] \oplus \\[1ex] \midwedge^{0,3}\,T^*M_3\otimes\frg \end{matrix} \xrightarrow{ \ D_\CA \ } \midwedge^{0,2}\,T^*M_3\otimes\frg \ ,
\end{align}
whose differentials are defined by linearized complex gauge transformations $C$ and the linearization $D_\CA$ of the first equation in \eqref{eq:inst_3d} respectively. 

We assume that the degree zero cohomology of the complex \eqref{eq:defcomplex} vanishes, i.e.~$\ker(C)=0$, which amounts to restricting to irreducible connections $\CA$ with only trivial automorphisms. The first cohomology $\ker(D_\CA)/\mathrm{im}(C)$ of the complex \eqref{eq:defcomplex} describes the tangent bundle $T\frM_{r,k}\longrightarrow \frM_{r,k}$ over a fixed holomorphic connection $\CA$. The second cohomology $\mathrm{coker}(D_\CA)$ defines the obstruction bundle $\Ob_{r,k}\longrightarrow\frM_{r,k}$ whose fibres are spanned by the zero modes of the antighost fields.

The \emph{virtual tangent bundle} $T^{\rm vir}\frM_{r,k}$ is the two-term elliptic complex
\begin{align} \label{eq:TvirMdef}
T^{\rm vir}\frM_{r,k} := \big[T\frM_{r,k} \xrightarrow{ \ \CD_\CA \ } \Ob_{r,k} \big] \ ,
\end{align}
where the fibrewise Kuranishi map $\CD_\CA$ is the linearization of the first two equations in \eqref{eq:inst_3d} composed with the projector onto the subspace orthogonal to the tangent space to the gauge orbit of $\CA$. Accordingly we define the complex virtual dimension of $\frM_{r,k}$ as
\begin{align} \label{eq:vdimC3}
\textrm{vdim} \, \frM_{ r,  k} := {\rm rk}(T\frM_{r,k}) -{\rm rk}(\Ob_{r,k}) = \dim\ker\CD_\CA - \dim \mathrm{coker}\,\CD_\CA
 \ ,
\end{align}
and the Euler class of its virtual tangent bundle as
\begin{align} \label{eq:Eulervirtual}
e (T^{\rm vir}\frM_{ r, k}):= \frac{e\big(T\frM_{r,k}\big)}{e\big(\Ob_{r,k}\big)} \ .
\end{align}

The complex \eqref{eq:TvirMdef} defines a virtual fundamental class  $[\frM_{r,k}]^{\rm vir}$; roughly speaking, it can be thought of as the Poincar\'e dual of the Euler class $e\big(\Ob_{r,k}\big)$ of the obstruction bundle.
The Atiyah--Singer index theorem computes its virtual dimension \eqref{eq:vdimC3} as the Euler character of the deformation complex \eqref{eq:defcomplex}. 
When $M_3=\FC^3$, the virtual dimension can also be computed from the ADHM parametrization by subtracting the number of equations and gauge symmetries from the total number of ADHM variables $(B_a,I, Y)_{a\,\in\,\ulthree}$, which vanishes:
\begin{align}
\textrm{vdim}\,\frM_{ r,  k} = (3\,k^2+r\,k+k^2)-(3\,k^2+r\,k) - k^2 =0 \ .
\end{align}

The ADHM parametrization of the instanton deformation complex is described by introducing two complex vector bundles over $\frM_{ r,  k}$ whose fibres over a gauge orbit $[\CA]$ are respectively the complex vector spaces $V$ and $W$ introduced in Section~\ref{subsec:ADHMC3}: the tautological  rank $k$ vector bundle
\begin{align}
\CCV = {\mu}^{\FC\,-1}(0)^{\rm stable} \, \times_{\sGL(V)} \, V \ ,
\end{align}
and the trivial rank $r$ Chan--Paton framing bundle
\begin{align}
\CCW = \frM_{r,k}\times W \ .
\end{align}

Then the tangent-obstruction theory is equivalently described by the cochain complex of vector bundles
\begin{align}\label{eq:complexbun_C3}
    \sEnd(\CCV)\xrightarrow{ \ \dd_1 \ }\begin{matrix}
        \sHom(\CCV, \CCV \otimes  Q_3 ) \\[1ex] \oplus\\[1ex] \sHom(\CCW,\CCV) \\[1ex] \oplus\\[1ex] \sHom(\CCV, \CCV \otimes\midwedge^3\,  Q_3 )
    \end{matrix}
    \xrightarrow{ \ \dd_2\ }\begin{matrix} \sHom(\CCV,\CCV\otimes\midwedge^2\,  Q_3  )
    \\[1ex] \oplus\\[1ex] \sHom(\CCV, \CCW \otimes\midwedge^3\,  Q_3 )
    \end{matrix} 
    \ , 
\end{align}
where the differentials $\dd_1$ and $\dd_2$ act fibrewise as an infinitesimal \smash{$\sGL(V)$} gauge transformation and the linearization of the two complex ADHM equations of \eqref{eq:ADHM3d}, respectively, while the three-dimensional Hermitian vector space $Q_3$ is the fundamental representation of the $\sU(3)$ holonomy group. The stability condition implies that the degree zero cohomology is trivial: $\ker(\dd_1)=0$.

\subsection{Instanton Partition Function}\label{sec:pf_C3}

Since the virtual dimension is zero, the instanton partition function of the six-dimensional cohomological gauge theory is given by
\begin{align}\label{eq:generic_pf}
Z^{r,k}_{\FC^3}= \int_{[\frM_{r,k}]^{\rm vir}}\, 1 \ .
\end{align}
The integral \eqref{eq:generic_pf} is understood as the $\sT$-equivariant volume of the moduli space $\frM_{r,k}$, evaluated via the virtual localization formula with respect to the action of some torus group $\sT$ \cite{Graber}. The $\sT$-action on the moduli space induces $\sT$-equivariant structures on the vector bundles $\CCV$ and $\CCW$. 

\subsubsection*{$\boldsymbol\Omega$-Background}

The natural choice for $\sT$ is associated to defining the gauge theory on Nekrasov's $\Omega$-background~\cite{Nekrasov:2002qd,Nekrasov:2003rj}. The global symmetry group of the six-dimensional cohomological field theory is
\begin{align}
\sG={\sf P}\sU(r) \times \sU(3)  \ ,
\end{align}
where ${\sf P}\sU(r)$ is the group of non-trivially acting framing rotations, and the holonomy group $\sU(3)$ acts in the fundamental representation $Q_3$ on $B=(B_a)_{a\,\in\,\ulthree}$, trivially on $I$, and in the determinant representation $\midwedge^3Q_3$ on $Y$. 

After conjugating $\sG$ to its maximal torus, the symmetry group acting on the theory is
\begin{align}\label{eq:maxtori_abelian}
\sT=\sT_{\vec  \tta} \times \sT_{\vec \epsilon}  \ ,
\end{align}
where $\sT_{\vec \tta}$ and $\sT_{\vec \epsilon}$ are (complex) maximal tori of ${\sf P}\sU(r)$  and $\sU(3)$ with coordinates  $\vec \tta=(\tta_1,\dots,\tta_r)$ (the vacuum expectation values of the complex Higgs field $\Phi$ parametrizing the positions of the $r$ D6-branes) and $\vec \epsilon=(\epsilon_1,\epsilon_2,\epsilon_3)$ (the parameters of the $\Omega$-deformation), respectively. The Coulomb moduli are equivalence classes, identified under simultaneous shifts $\sfa_l\longmapsto\sfa_l+{\sf c}$ by any ${\sf c}\in\FC$ for~$l=1,\dots,r$.

By the virtual localization formula~\cite{Graber}, the full equivariant instanton partition function is given as a function of the equivariant parameters $(\vec \tta, \vec \epsilon\,)$ by a sum over $\sT$-fixed points
\begin{align} \label{eq:ZC3def}
Z_{\FC^3}^{r}( \qu;\vec \tta,\vec\epsilon\,) =\sum_{k=0}^\infty\, \qu^k\, Z_{\FC^3}^{ r,k}(\vec \tta,\vec\epsilon\,)= \sum_{k=0}^\infty\, \qu^k \ \sum_{\vec\pi\,\in\,\frM_{r,k}^{\sT}} \,\frac{1}{e_{\sT}\big(T^{\textrm{vir}}_{\vec \pi}\frM_{r,k}\big)} \ ,
\end{align}
where $\qu$ is the Boltzmann weight parameter for instantons, $\frM_{r,k}^{\sT}$ is the set of $\sT$-fixed points of the instanton moduli space, and $e_\sT$ denotes the $\sT$-equivariant Euler class.

\subsubsection*{Fixed Points and Plane Partitions}

The $\sT$-fixed points of the moduli space $\frM_{r,k}$ are all isolated and in one-to-one correspondence with arrays $\vec \pi=(\pi_1,\dots,\pi_r)$, where each $\pi_l$ for $l=1,\dots, r$ is a plane partition \cite{Cirafici:2008sn}. A plane partition is an ordered sequence $\pi=(\pi_{i,j})_{i,j\geq 1}$  of non-negative integers
$\pi_{i,j}\in\RZ_{\geq 0}$ decreasing along both directions: 
\begin{align}
\pi_{i,j}\geq\pi_{i+1,j} \qquad \text{and}\qquad \pi_{i,j}\geq\pi_{i,j+1} \ .
\end{align}
We may view $\pi$ as a three-dimensional Young diagram in \smash{$\RZ_{\geq0}^3$}, obtained by piling $\pi_{i,j}$ boxes over \smash{$(i,j)\in\RZ_{\geq0}^2$}.
The size of $\pi$ is the total number of boxes and is denoted $|\pi|:=\sum_{i,j\geq1}\,\pi_{i,j}$.
The size $|\vec \pi|$ of $\vec \pi$ is defined to be the sum of the sizes of its components $\pi_l$. Then \smash{$\vec\pi\in\frM_{r,k}^{\sT}$} partitions the instanton number $k$:
\begin{align}
|\vec\pi| = \sum_{l=1}^{r}\, | \pi_l|=k \ .
\end{align}

To explicitly compute the Euler classes in \eqref{eq:ZC3def}, we use the ADHM parametrization of the instanton deformation complex. The fibre of the complex of vector bundles \eqref{eq:complexbun_C3} over the fixed point \smash{$\vec\pi\in\frM_{r,k}^{\sT}$} reads
\begin{align}\label{eq:complex_C3}
    \sEnd_\FC(V_{\vec\pi})\xrightarrow{ \ \dd_1 \ }\begin{matrix}
        \sHom_\FC(V_{\vec\pi}, V_{\vec\pi} \otimes  Q_3 ) \\[1ex] \oplus\\[1ex] \sHom_\FC(W_{\vec\pi},V_{\vec\pi}) \\[1ex] \oplus\\[1ex] \sHom_\FC(V_{\vec\pi}, V_{\vec\pi} \otimes\midwedge^3\,  Q_3 )
    \end{matrix}
    \xrightarrow{ \ \dd_2\ }\begin{matrix} \sHom_\FC(V_{ \vec\pi},V_{\vec\pi}\otimes\midwedge^2\,  Q_3  )
    \\[1ex] \oplus\\[1ex] \sHom_\FC(V_{\vec\pi}, W_{\vec\pi} \otimes\midwedge^3\,  Q_3 )
    \end{matrix} 
    \ ,
\end{align}
where here the vector space $Q_3\simeq\FC^3$ is regarded as the three-dimensional fundamental $\sT_{\vec \epsilon}\,$-module with weight decomposition
\begin{align}
Q_3=t_1^{-1}+t_2^{-1}+t_3^{-1}
\end{align}
in the representation ring of $\sT_{\vec\epsilon}\,$, where $t_a=\e^{\,\ii\,\epsilon_a}$.

The equivariant character of the virtual tangent bundle is computed from the index of the complex \eqref{eq:complex_C3} and is given by
\begin{align}\begin{split}\label{eq:complex_3d}
\ch_{\sT}\big(T_{\vec\pi}^{\rm vir} \frM_{ r,k}\big)
&= W_{\vec \pi}^*\otimes V_{\vec\pi} - \frac{V^*_{\vec\pi}\otimes W_{\vec\pi}}{ t_1\,t_2\,t_3} + V_{\vec\pi}^*\otimes V_{\vec\pi} \ \frac{(1-t_1)\,(1-t_2)\,(1-t_3)}{t_1\,t_2\,t_3} \ . \end{split}
\end{align}
Seen as modules in the representation ring of $\sT$, it follows from the stability condition \eqref{eq:stabilityC3} that, after a gauge transformation, the vector spaces $V$ and $W$ decompose at the fixed point \smash{$\vec\pi\in\frM_{r,k}^\sT$} with respect to the $\sT$-action as
\begin{align}\label{eq:decom}
V_{\vec \pi}=\sum_{l=1}^{r}\,e_l \ \sum_{\vec p\,\in\pi_l}\,t_1^{p_1-1}\,t_2^{p_2-1}\,t_3^{p_3-1} \qquad \mbox{and} \qquad  
W_{\vec \pi}=\sum_{l=1}^{r}\,e_l \ ,
\end{align}
where $e_l=\e^{\,\ii\, \tta_l}$. The dual involution acts on the weights as $t_a^*=t_a^{-1}$ and $e_l^*=e_l^{-1}$.
We can then extract the Euler classes from the top-form part of the character \eqref{eq:complex_3d} through the operation
\begin{align}\label{eq:top}
\widehat{\texttt e}\Big[\sum_I \, n_I\,\e^{\,{\sf w}_I}\Big] = \prod_{{\sf w}_I\neq0} \, {\sf w}_I^{n_I} \ .
\end{align}

\subsubsection*{Equivariant Generating Function}

The full equivariant instanton partition function is given by the combinatorial formula
\begin{align}\begin{split}\label{eq:pf_3c}
Z_{\FC^3}^{r}(\qu;\vec \tta,\vec\epsilon\,)&=\sum_{\vec \pi\in\frM_r^\sT}\,{{\texttt q}^{|\vec \pi|}} \ \widehat{\texttt e}\big[-\ch_{\sT}(T_{\vec\pi}^{\rm vir} {\frM}_{ r,k}) \big] \\[4pt] 
&=\sum_{\vec \pi\in\frM_r^\sT}\,\qu^{|\vec \pi|} \ \prod_{l=1}^{r} \ \prod_{\vec p_l\in{\pi}_{l}}^{\neq0} \, \frac{P_{r}(-\tta_{l}-\vec p_{l}\cdot\vec\epsilon\,|\epsilon_{123}-\vec \tta)}{P_{r}(\tta_{l}+\vec p_{l}\cdot\vec\epsilon\,|\vec \tta)}\\ & \hspace{3cm}
 \times \prod_{ l'=1}^{r} \ \prod_{\vec p_{l'}^{\,\prime}\in{\pi}_{l'}}^{\neq0} \,R (\tta_l-\tta_{l'}+(\vec p_l-\vec p^{\,\prime}_{l'})\cdot\vec\epsilon\,|\vec \epsilon\,) \ , \end{split}
\end{align}
where $\vec p\cdot\vec\epsilon:=\sum_{a\,\in\,\ulthree}\,p_a\,\epsilon_a$.

In \eqref{eq:pf_3c} we introduced the polynomial and rational functions
\begin{align}\label{eq:CPCRdef}
P_r(x|\vec w) = \prod_{l=1}^r \, (x-w_l) \qquad \mbox{and} \qquad R (x |\vec \epsilon\,)=\frac{x\,(x-\epsilon_{12})\,(x-\epsilon_{23})\,(x-\epsilon_{13})}{(x-\epsilon_1)\,(x-\epsilon_{2})\,(x-\epsilon_3)\,(x-\epsilon_{123})} \ ,
\end{align}
along with the shorthand notation
\begin{align}
\epsilon_{ab\cdots}=\epsilon_a+\epsilon_b+\cdots \ .
\end{align}
The superscripts ${}^{\neq0}$ on the products designate the omission of terms with zero numerator or denominator according to the top-form operation \eqref{eq:top}.

The complicated combinatorial series \eqref{eq:pf_3c} can be summed to a simple closed formula~\cite{Awata:2009dd,Szaboconj,proofconj}.

\begin{theorem} \label{thm:C3inst}
The generating function $Z^r_{\FC^3}(\qu;\vec \tta,\vec \epsilon\,)$ for the rank $r$ Donaldson--Thomas invariants of $\FC^3$ with $\sU(3)$ holonomy is independent of the Coulomb moduli $\vec \tta$ and can be expressed as
\begin{align}
Z_{\FC^3}^r(\qu;\vec \epsilon\,)= M\big((-1)^r\,\qu\big)^{-r\,\frac{\epsilon_{12}\,\epsilon_{23}\,\epsilon_{13}}{\epsilon_1\,\epsilon_2\,\epsilon_3}} \ ,
\end{align}
where $M(q):=M(1,q)$ is the generating function which counts plane partitions, and
\begin{align}
M(x,q)=\prod_{n=1}^\infty\,\frac1{(1-x\,q^n)^{n}}
\end{align}
is the MacMahon function.
\end{theorem}

\section{Cohomological Gauge Theory on Local K\"ahler Three-Orbifolds}\label{sec:C3orbifold}

In this section we turn to the study of noncommutative instantons on orbifolds of $\FC^3$, i.e. on \emph{local} three-orbifolds. We consider holomorphic actions on $\FC^3$ by finite orbifold groups which preserve the $\sU(3)$ holonomy, and hence the K\"ahler form $\omega$. The orbifold cohomological gauge theory is constructed by allowing the fields to be equivariant and gauging the orbifold group action, followed by projection to invariant states described by equivariant decomposition of the generalized instanton equations \eqref{eq:inst_3d}; this can be thought of as a field theory on the corresponding orbifold resolution of the quotient singularity. The construction is motivated by considerations of D-branes on orbifolds~\cite{Cirafici:2010bd}, and in particular it naturally incorporates `twisted sectors' corresponding to conjugacy classes of the orbifold group. The orbifold BRST fixed point equations are naturally realised in noncommutative gauge theory, along the lines of~\cite{Cirafici:2010bd,Cirafici:2012qc,Szabo:2022zyn}. Here we describe the vacuum states via equivariant decomposition of the corresponding ADHM parametrization.

We start by reviewing the construction of the generalized McKay quiver $\ttQ^\sGamma$ for a finite subgroup $\sGamma$ of $\sSU(3)$. To these quivers we associate ADHM-type equations which parametrize the moduli space of instantons on $\FC^3/\,\sGamma$, viewed as the moduli space of $\sGamma$-equivariant instantons on $\FC^3$, as a quiver variety,  that is, as the moduli space of stable framed representations of the bounded McKay quiver. Then we analyze the most general admissible orbifolds which allow for the definition of a torus-equivariant gauge theory, in both cases where $\sGamma$ is an abelian and a non-abelian finite group represented in $\sU(3)$; these considerations lead to more general classes of orbifold theories based on non-effectively acting groups $\sGamma$. Although the orbifold singularity is generally supersymmetric only when $\sGamma$ embeds in $\sSU(3)\subset\sU(3)$, the orbifold instanton locus of the cohomological gauge theory is always stable and has a realisation in terms of states of D-branes.

See Appendix~\ref{app:B} for our notational conventions for finite groups, as well as for the classification of the finite subgroups of $\sSU(3)$ which we use extensively throughout this paper. The McKay quivers $\ttQ^\sGamma$ for finite subgroups $\sGamma\subset\sSL(3,\FC)$ are described in~\cite{ButinPerets}, while for small finite subgroups $\sGamma\subset\sGL(3,\FC)$ they are detailed in~\cite{Hu:2012bz}.

\subsection{Quiver Varieties}\label{sec:quiver_variety}

Let $\sGamma$ be a finite subgroup of $\sSL(3,\FC)$ which acts on $\FC^3$ by the fundamental representation $Q_3$. 

\subsubsection*{McKay Quivers}

The McKay quiver associated to $\sGamma$ is denoted $\ttQ^\sGamma=\big(\ttQ^\sGamma_0,\ttQ^\sGamma_1\big)$, where $\ttQ^\sGamma_0$ and $\ttQ^\sGamma_1$ denote the sets of vertices and edges,  respectively, and it is constructed in the following way. As a set, \smash{$\ttQ^\sGamma_0\simeq\widehat\sGamma$} is the set of irreducible representations of $\sGamma$, which corresponds bijectively to the set of conjugacy classes of $\sGamma$. We write $\lambda_\sfi\in\widehat\sGamma$ for the irreducible representation labelled by $\sfi\in\ttQ^\sGamma_0$; the trivial one-dimensional representation is denoted $\lambda_{0}$. The number of oriented edges (arrows) from a vertex $\sfi$ to  a vertex $\sfi'$ is determined by the adjacency matrix \smash{$A=(a_{\sfi\sfi'} )_{\sfi,\sfi'\in\ttQ^\sGamma_0}$} of tensor product multiplicities $a_{\sfi\sfi'}=\dim\sHom_\sGamma(\lambda_\ii,Q_3\otimes\lambda_{\sfi'})  \in \RZ_{\geq0}$ in the decomposition of $\sGamma$-modules
\begin{align}\label{eq:decomp_Q}
Q_3\otimes\lambda_\sfi=\bigoplus_{\sfi'\in\ttQ^\sGamma_0}\,a_{\sfi\sfi'}\,\lambda_{\sfi'}\ .
\end{align}

If $e\in\ttQ^\sGamma_1$ is an edge determined by \eqref{eq:decomp_Q}, the source vertex of $e$ is denoted by $\sfs(e)$ and its target vertex by $\sft(e)$; this defines source and target maps \smash{$\ttQ_1^\sGamma \doublerightarrow{ \ \sfs \ }{ \ \sft \ } \ttQ_0^\sGamma$}. The quiver $\ttQ^\sGamma$ contains no loop edges $e$, i.e. $\sfs(e)=\sft(e)$, if and only if the trivial
representation $\lambda_0$ does not appear in the decomposition of $Q_3$ into irreducible $\sGamma$-modules.

\begin{example}\label{ex:quiver3d}
Consider the non-abelian group $\sGamma=\sC_3(1,0)=(\RZ_3\times\RZ_3)\rtimes\RZ_3$ of type~C, where the action of the groups $\RZ_3\times\RZ_3=\{c_{i,j}\}_{i,j\in\{0,1,2\}}$ and $\RZ_3=\{1,C,C^2\}$ on $\FC^3$ is given by the $\sSU(3)$ matrices
\begin{align}
c_{i,j}={\small \begin{pmatrix}
\xi_3^i&&\\
&\xi_3^{-j}&\\
&&\xi_3^{-i+j}
\end{pmatrix} } \normalsize \qquad \mbox{and} \qquad C= {\small \begin{pmatrix}
0&1&0\\
0&0&1\\
1&0&0
\end{pmatrix} } \normalsize\ ,
\end{align}
with $\xi_3=\e^{\,2\pi\,\ii/3}$ a primitive third root of unity. As shown in \cite{Hu:2012bz},  the group $\sC_3(1,0)$ has two three-dimensional irreducible representations, $\lambda_{10}=Q_3$ and $\lambda_{20}$, and nine one-dimensional irreducible representations, $\lambda_{00i}$, $\lambda_{12i}$ and $\lambda_{21i}$ with $i\in\{0,1,2\}$, where $\lambda_{000}=\lambda_0$. 

The tensor product decompositions with the fundamental representation $Q_3$ give
\begin{align}
\begin{split}
Q_3\otimes\lambda_{10}=3\,\lambda_{20} \qquad ,& \qquad Q_3\otimes\lambda_{20}=\bigoplus_{i=0}^2\,\big(\lambda_{00i}\oplus\lambda_{12i}\oplus\lambda_{21i}\big) \ , \\[4pt] Q_3\otimes\lambda_{00i}=\,&Q_3\otimes\lambda_{12i} =Q_3\otimes\lambda_{21i}=\lambda_{10} \ .
\end{split}
\end{align}
The generalized McKay quiver $\ttQ^{\sC_3(1,0)}$ constructed from this representation theory data is 
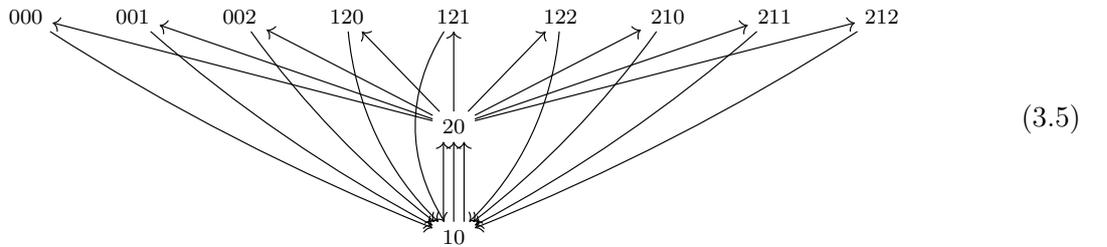
\begin{equation}
\scriptsize
\begin{tikzcd}
	000 & 001 & 002 & 120 & 121 & 122 & 210 & 211 & 212 \\
	\\
	&&&& 20 \\
	\\
	&&&& 10
	\arrow[from=3-5, to=1-5]
	\arrow[from=5-5, to=3-5]
	\arrow[shift left=2, from=5-5, to=3-5]
	\arrow[shift right=2, from=5-5, to=3-5]
	\arrow[curve={height=-6pt}, from=1-8, to=5-5]
	\arrow[curve={height=-6pt}, from=1-9, to=5-5]
	\arrow[curve={height=-6pt}, from=1-7, to=5-5]
	\arrow[curve={height=-12pt}, from=1-6, to=5-5]
	\arrow[curve={height=18pt}, from=1-5, to=5-5]
	\arrow[curve={height=12pt}, from=1-4, to=5-5]
	\arrow[curve={height=6pt}, from=1-3, to=5-5]
	\arrow[curve={height=6pt}, from=1-2, to=5-5]
	\arrow[curve={height=6pt}, from=1-1, to=5-5]
	\arrow[from=3-5, to=1-1]
	\arrow[from=3-5, to=1-2]
	\arrow[from=3-5, to=1-3]
	\arrow[from=3-5, to=1-4]
	\arrow[from=3-5, to=1-6]
	\arrow[from=3-5, to=1-7]
	\arrow[from=3-5, to=1-8]
	\arrow[from=3-5, to=1-9]
\end{tikzcd}
\normalsize
\end{equation}
\end{example}

\subsubsection*{Enhanced Framed Quiver Representations}

The McKay quiver $\ttQ^\sGamma=\big(\ttQ_0^\sGamma,\ttQ_1^\sGamma\big)$ serves as a powerful combinatorial device for describing the $\sGamma$-equivariant decomposition of the ADHM equations \eqref{eq:ADHM3d}, which we view as BRST fixed point equations for a cohomological field theory on the orbifold crepant resolution 
\begin{align} \label{eq:orbresolutionSU3}
\pi_{\rm orb}:\big[\FC^3\,\big/\,\sGamma\big] \longrightarrow \FC^3\,\big/\,\sGamma
\end{align}
of the quotient singularity $\FC^3/\,\sGamma$. A field on the quotient stack $[\FC^3/\,\sGamma]$ is the same thing as a $\sGamma$-equivariant field on $\FC^3$; for example, we may present the quotient stack as the action groupoid \smash{$\sGamma\times\FC^3 \doublerightarrow{\ \ }{ \ } \FC^3$} and the morphism $\pi_{\rm orb}$ as the quotient map to the orbit space. The nodes $\sfi\in\ttQ_0^\sGamma$ specify the basis of \emph{fractional instantons} which are stuck at the orbifold singularity. The McKay quiver will also aid in  describing the corresponding moduli space of solutions. 

To implement the orbifold projection, we regard the Hermitian vector spaces $V$ and $W$ as $\sGamma$-modules and decompose them into irreducible representations of the orbifold group as
\begin{align} \label{eq:VWdecompSU3}
V=\bigoplus_{\sfi\in\ttQ_0^\sGamma} \, V_\sfi\otimes\lambda_\sfi^* \qquad \mbox{and} \qquad W=\bigoplus_{\sfi\in\ttQ_0^\sGamma} \, W_\sfi\otimes\lambda_\sfi^* \ .
\end{align}
The multiplicity spaces $V_\sfi=\sHom_\sGamma(\lambda_\sfi,V)$ and $W_\sfi=\sHom_\sGamma(\lambda_\sfi,W)$ are Hermitean vector spaces  of complex dimensions $k_\sfi$ and $r_\sfi$, respectively, which carry a trivial $\sGamma$-action; the dimensions $k_\sfi$ are called fractional instanton charges. 
We assemble the dimensions into vectors \smash{$\vec k,\vec r\in \RZ^{\ttQ_0^\sGamma}_{\geq 0}$}. 
The dimensions $k=\dim V$ and $r=\dim W$ correspondingly decompose into sums
\begin{align}
k = |\vec k\,| := \sum_{\sfi\in\ttQ_0^\sGamma} \, d_\sfi\,k_\sfi \qquad \mbox{and} \qquad r =|\vec r\,| := \sum_{\sfi\in\ttQ_0^\sGamma} \, d_\sfi\,r_\sfi \ ,
\end{align}
where $d_\sfi$ is the dimension of the irreducible representation $\lambda_\sfi$. The special case of $n$ freely moving instantons corresponds to taking $k_\sfi=n\,d_\sfi$, with total charge $k=n\,\#\sGamma$; when $n=1$ this is called a \emph{regular instanton}, as it lives in the regular representation $\FC[\sGamma]$ of the orbifold group $\sGamma$.

Next we regard the ADHM variables as $\sGamma$-equivariant maps
\begin{align}
(B,I, Y) \ \in \ \sHom_\sGamma(V,V\otimes Q_3) \ \oplus \ \sHom_\sGamma(W,V) \ \oplus \ \sHom_\sGamma(V, V\otimes \midwedge^3Q_3) \ .
\end{align}
From \eqref{eq:decomp_Q} and Schur's Lemma it follows that $B$ decomposes into linear maps associated to each edge of the McKay quiver:
\begin{align} \label{eq:Bdecomp}
B=\bigoplus_{e\in\ttQ_1^\sGamma} \, B_e \qquad \mbox{with} \quad B_e: V_{\sfs(e)}\longrightarrow V_{\sft(e)} \ .
\end{align}
Thus the ADHM datum $B$ defines a linear representation of the McKay quiver $\ttQ^\sGamma$ with dimension vector~$\vec k$. 

Similarly, $I$ decomposes into linear maps associated to each vertex:
\begin{align} \label{eq:Idecomp}
I = \bigoplus_{\sfi\in\ttQ_0^\sGamma} \, I_\sfi \qquad \mbox{with} \quad I_\sfi:W_\sfi\longrightarrow V_\sfi \ .
\end{align}
Thus $I$ defines a framing of the representation of the McKay quiver $\ttQ^\sGamma$ with dimension vector $\vec r$.

Finally, since $\sGamma\subset\sSL(3,\FC)$, it follows that the determinant representation $\midwedge^3Q_3\simeq\lambda_0$ is trivial as a $\sGamma$-module and hence
\begin{align} \label{eq:Ydecomp}
Y = \bigoplus_{\sfi\in\ttQ_0^\sGamma} \, Y_\sfi \qquad \mbox{with} \quad Y_\sfi\in \sEnd_\FC(V_\sfi) \ .
\end{align}
We may depict the decomposition of $Y$ by the addition of a single edge loop at each vertex. We call this an \emph{enhancement} of the framed quiver representation of $\ttQ^\sGamma$.

\subsubsection*{Orbifold ADHM Equations}

The set of maps $(B_{e},I_\sfi, Y_\sfi)_{e\in\ttQ_1^\sGamma,\sfi\in\ttQ_0^\sGamma}$ satisfy ADHM-type equations which are derived by decomposing the equations \eqref{eq:ADHM3d} as $\sGamma$-equivariant maps
\begin{align}
\big(\mu^\FC,\mu^\FR,\sigma\big) \ \in \ \sHom_\sGamma(V,V\otimes\midwedge^2Q_3) \ \oplus \  \sEnd_\sGamma(V) \ \oplus \ \sHom_{\sGamma}(V,W\otimes\midwedge^3Q_3) \ .
\end{align}

Since $\midwedge^3Q_3\simeq\lambda_0$, the second and third equations have isotypical components which live at the vertices $\sfi\in\ttQ_0^\sGamma$. Writing their equivariant decompositions
\begin{align}
\mu^\FR=\bigoplus_{\sfi\in\ttQ_0^\sGamma}\,\mu_\sfi^\FR \qquad \mbox{and} \qquad \sigma=\bigoplus_{\sfi\in\ttQ_0^\sGamma}\,\sigma_\sfi \ ,
\end{align}
with $\mu_\sfi^\FR\in\sEnd_\FC(V_\sfi)$ and $\sigma_\sfi\in\sHom_\FC(V_\sfi,W_\sfi)$, these equations read explicitly as
\begin{align}\label{eq:ADHMrho1}\begin{split}
\mu_\sfi^{\FR}&:=\sum_{e\,\in\, \sft^{-1}(\sfi)}\,B_{e}\,B_{e}^\dagger - \sum_{e\,\in\, \sfs^{-1}(\sfi)}\,B^\dagger_{e}\,B_{e} +\big[ Y_\sfi^\dagger \, ,  Y_\sfi\big]+ I_\sfi\,I_\sfi^\dagger= \zeta_\sfi\,  \ident_{V_\sfi}  \ , \\[4pt]
\sigma_\sfi&:=I^\dagger_\sfi\,  Y_\sfi=0 \ , \end{split}
\end{align}
for all $\sfi\in\ttQ^\sGamma_0$, where the Fayet--Iliopoulos parameters $\zeta_\sfi\in\FR_{>0}$ are determined by the decomposition of the Neveu--Schwarz $B$-field into twisted NS--NS sectors of type~IIA string theory on $\FC^3/\,\sGamma$. 

The isotypical decomposition of the equations $\mu^\FC\in\sHom_\sGamma(V,V\otimes\midwedge^2Q_3)$ is more complicated. We start by rewriting it in a basis independent form as
\begin{align}\label{eq:ADHMrho2}
\mu^\FC&=B\wedge B-\langle B^\dagger, Y\rangle_{Q_3} + \langle Y, B^\dagger\rangle_{Q_3} = 0  \ ,
\end{align}
where $\langle\,\cdot\,,\,\cdot\,\rangle_{Q_3}$ is the Hermitian inner product on $Q_3$. The equations $\mu_{ab}^\FC=0$ from \eqref{eq:ADHM3d} follow by expanding $B$ in the canonical basis of $Q_3=\FC^3$.

The tensor product decomposition \eqref{eq:decomp_Q} together with triviality of the determinant representation imply
\begin{align}\label{eq:decomp_wedge2Q}
\midwedge^2Q_3\otimes\lambda_\sfi\simeq Q_3^*\otimes\lambda_\sfi =\bigoplus_{\sfi'\in\ttQ^\sGamma_0}\,a_{\sfi'\sfi}\,\lambda_{\sfi'}\ .
\end{align}
Hence the multiplicities of linear maps from $V_\sfi$ to $V_{\sfi'}$ given by the isotypical decomposition of the equations $\mu^\FC$ is equal to the number of oriented edges connecting the vertex $\sfi'$ to the vertex $\sfi$; that is, the number of arrows $\sfi'\longrightarrow \sfi$ in the \emph{opposite} direction. In particular, the isotypical components of $\mu^\FC$ can be labelled by the edges $e\in\ttQ_1^\sGamma$.

Writing the equivariant decomposition $\mu^\FC=0$ as
\begin{align}\label{eq:muC_decom}
\mu^\FC=\bigoplus_{e\in\ttQ_1^\sGamma}\,\mu_e^\FC \ ,
\end{align}
the ADHM equations $\mu_e^\FC=0$ can be inferred from unravelling \eqref{eq:ADHMrho2} in a basis tailored to the particular $\sGamma$-action on $\FC^3$, by multiplying matrices in the equivariant decompositions \eqref{eq:Bdecomp} and \eqref{eq:Ydecomp}. In concrete examples, the equations are always independent of all choices made for a particular quiver \smash{$\ttQ^\sGamma$}. 

\begin{example}\label{ex:gammaSU(2)}
Let $\sGamma$ be a finite subgroup of $\sSU(2)$ acting in the fundamental representation $Q_2$ on an affine plane $\FC^2\subset\FC^3$ and trivially on the affine line $\FC=\FC^3\setminus\FC^2$. Then $\FC^3/\,\sGamma\simeq\FC^2/\,\sGamma\times\FC$. Since the representation $Q_2\simeq Q_2^*$ is self-dual, the adjacency matrix $A$ of the McKay quiver $\ttQ^\sGamma$ is symmetric, i.e. $a_{\sfi\sfi'}=a_{\sfi'\sfi}$. Thus \smash{$\ttQ^\sGamma=\overline{\mathsf{Dynk}_\sGamma}$} is the \emph{double} of a quiver $\mathsf{Dynk}_\sGamma$, i.e. the quiver with the same set of nodes $\ttQ_0^\sGamma = \mathsf{Dynk}_{\sGamma\,0}$ and with arrow set $\ttQ_1^\sGamma = \mathsf{Dynk}_{\sGamma\,1}\sqcup\mathsf{Dynk}_{\sGamma\,1}^{\textrm{op}}$, 
where the opposite quiver $\mathsf{Dynk}_\sGamma^{\textrm{op}}$ is obtained from $\mathsf{Dynk}_\sGamma$ by reversing the orientation of the edges. By the classical McKay correspondence~\cite{McKay}, the quiver $\mathsf{Dynk}_\sGamma$ is associated to any choice of orientation of an affine Dynkin diagram of type ADE~\cite{Nakajima:1994nid,Nakajima:1998}, with an additional edge loop at each vertex. In this case we label the vertices of the McKay quiver as $\ttQ_0^\sGamma=\{0,1,\dots,r_\sGamma\}$, where $0$ indicates the trivial representation and $r_\sGamma$ is the rank of the corresponding simply-laced Lie algebra $\frg_\sGamma$.

To each arrow $e$ of the extended Dynkin diagram underlying $\mathsf{Dynk}_\sGamma$ we associate two linear maps $B_e:V_{\sfs(e)}\longrightarrow V_{\sft(e)}$ and $\bar B_e: V_{\sft(e)}\longrightarrow V_{\sfs(e)}$. To each vertex $\sfi$ of $\mathsf{Dynk}_\sGamma$ we associate three maps $L_\sfi, I_\sfi, Y_\sfi\in\sEnd_\FC(V_\sfi)$. Then the ADHM equations \eqref{eq:ADHMrho1} and \eqref{eq:ADHMrho2} can be expressed as
\begin{align}\label{eq:ADHMsu2}
\begin{split}
\mu_{\sfi}^{\FC}&=\sum_{e\,\in\, \sfs^{-1}(\sfi)} \, \bar B_{e}\,B_{e}- \sum_{e\,\in\, \sft^{-1}(\sfi)}\, B_{e}\,\bar B_{e} +
\big[L_{\sfi}^\dagger , Y_\sfi \big]= 0  \ ,
\\[4pt]
\mu_{e}^{\FC}&=  L_{\sft(e)}\,B_{e}-  B_{e}\, L_{\sfs(e)} -\bar B_e^\dagger \,Y_{\sfs(e)} +Y_{\sft(e)}\,\bar B_e^\dagger = 0 \ ,
\\[4pt]
\bar\mu_{e}^{\FC}&= L_{\sfs(e)}\,\bar B_{e}-  \bar B_{e}\, L_{\sft(e)} +B_e^\dagger\, Y_{\sfs(e)} -Y_{\sft(e)}\, B_e^\dagger= 0  \ , \\[4pt]
\mu_\sfi^{\FR}&=\sum_{e\,\in\, \sft^{-1}(\sfi)}\,\big(B_{e}\,B_{e}^\dagger-\bar B_{e}^\dagger\, \bar B_{e}\big) - \sum_{e\,\in\, \sfs^{-1}(\sfi)}\,\big( B^\dagger_{e}\,B_{e} -\bar B_{e}\,\bar B_{e}^\dagger\big)\\
&\hspace{1cm} +\big[ Y_\sfi^\dagger \, ,  Y_\sfi\big]+\big[L_\sfi^\dagger \, , L_\sfi\big]+ I_\sfi\,I_\sfi^\dagger= \zeta_\sfi \, \ident_{V_\sfi}  \ ,  \\[4pt]
\sigma_\sfi&=I^\dagger_\sfi\,  Y_\sfi=0 \ . \end{split}
\end{align}
This construction is independent of the choice of orientation of the Dynkin diagram.
\end{example}

\subsubsection*{Moduli Spaces of Orbifold Instantons}

The action of $\sGamma$ on the decompositions \eqref{eq:VWdecompSU3} is defined by group homomorphisms $\gamma_{\scriptscriptstyle V}:\sGamma\longrightarrow\sU(k)$ and $\gamma_{\scriptscriptstyle W}:\sGamma\longrightarrow\sU(r)$ with 
\begin{align} \label{eq:VWGammaaction}
\gamma_{\scriptscriptstyle V}(g)\big(v^\sfi\otimes \ell_\sfi\big) = v^\sfi\otimes\big(\lambda_\sfi^*(g)(\ell_\sfi)\big) \qquad \mbox{and} \qquad \gamma_{\scriptscriptstyle W}(g)\big(w^\sfi\otimes \ell_\sfi\big) = w^\sfi\otimes\big(\lambda_\sfi^*(g)(\ell_\sfi)\big) \ ,
\end{align}
for all $g\in\sGamma$, $\ell_\sfi\in\lambda_\sfi^*$, $v^\sfi\in V_\sfi$ and $w^\sfi\in W_\sfi$, where $\lambda_\sfi^*(g)\in\sU(d_\sfi)$.
These break the $\sU(k)$ and $\sU(r)$ symmetries to the subgroups 
\begin{align}
\sU\big(\vec k\,\big) := \Timesbig_{\sfi\in\ttQ^\sGamma_0}\,\sU(k_\sfi) \qquad \mbox{and} \qquad \sU(\vec r\,):=\Timesbig_{\sfi\in\ttQ^\sGamma_0}\,\sU(r_\sfi)
\end{align}
commuting with the respective $\sGamma$-actions in \eqref{eq:VWGammaaction}. In the type IIA picture, the isotypical components of \eqref{eq:VWdecompSU3} specify fractional D0-branes and D6-branes, respectively, whose bound states can be identified geometrically with $\sGamma$-equivariant sheaves on $\FC^3$.

The action of a unitary automorphism
\begin{align}
g = (g_\sfi)_{\sfi\in\ttQ_0^\sGamma} \ \in \  \sU\big(\vec k\,)
\end{align}
on the orbifold ADHM data, given by 
\begin{align}\label{eq:U(k)action}
 g\cdot \big(B_{e}\,,\, I_\sfi\,,\,Y_\sfi\big)_{\sfi\in\ttQ_0^\sGamma\,,\, e\in\ttQ_1^{\sGamma}} = \big(g_{\sft(e)}\,B_{e}\,g_{\sfs(e)}^{-1}\,, \, g_\sfi \, I_\sfi \,,\, \, g_\sfi \, Y_\sfi\,g_\sfi^{-1}\big)_{\sfi\in\ttQ_0^\sGamma\,,\, e\in\ttQ_1^{\sGamma}} \ ,
\end{align}
leaves the ADHM equations \eqref{eq:ADHMrho1} and \eqref{eq:ADHMrho2} invariant.
Let
\begin{align}
\vec \mu :=\big(\mu^\FC\,,\,\mu_\sfi^{\FR}\,,\,\sigma_\sfi\big)_{\sfi\in\ttQ^\sGamma_0} \qquad \mbox{and} \qquad \vec \zeta:=\big( 0\,,\,\zeta_\sfi\,,\,0\big)_{\sfi\in\ttQ^\sGamma_0} \ .
\end{align}
For each pair of dimension vectors $\vec k,\vec r\in\RZ_{\geq0}^{\ttQ_0^\sGamma}$, we define the \textit{quiver variety} as the quotient
\begin{align}\label{eq:Quiver_variety}
\frM_{\vec r,\vec k}=\vec \mu\,^{-1}\big(\vec \zeta\, \big) \,\big/\, \sU\big(\vec k\,\big) \ .
\end{align}

The quiver varieties form the connected components of the stratification of the moduli space
\begin{align}
\frM_{r,k}^\sGamma = \bigsqcup_{|\vec r\,|=r \,,\, |\vec k\,|=k} \, \frM_{\vec r,\vec k}
\end{align}
of charge $k$ noncommutative $\sU(r)$ instantons on the Calabi--Yau orbifold $\FC^3/\,\sGamma$. It can be regarded as a moduli space of modules over a corresponding path algebra of the quiver $\ttQ^\sGamma$, which is Morita equivalent to the skew group algebra $\FC[Q_3]\rtimes\sGamma$. We view this algebra as a noncommutative crepant resolution of the quotient singularity $\FC^3/\,\sGamma$, and identify \smash{$\frM_{r,k}^\sGamma$} as the moduli space for the noncommutative Donaldson--Thomas theory of $\FC^3/\,\sGamma$.

\begin{remark}[{\bf Framing Symmetry}] \label{rem:framingrot}
The quiver variety \eqref{eq:Quiver_variety} is invariant under the framing rotations
\begin{align}
h = (h_\sfi)_{\sfi\in\ttQ_0^\sGamma} \ \in \  \sU(\vec r\,)
\end{align}
which act on the orbifold ADHM data as
\begin{align}
h\cdot \big(B_{e}\,,\, I_\sfi\,,\,Y_\sfi\big)_{\sfi\in\ttQ_0^\sGamma\,,\, e\in\ttQ_1^{\sGamma}} = \big(B_{e}\,, \, I_\sfi\,h_\sfi^{-1} \,,\, \, Y_\sfi\big)_{\sfi\in\ttQ_0^\sGamma\,,\, e\in\ttQ_1^{\sGamma}}    \ .
\end{align}
The maximal torus of the global colour group $\sU(\vec r\,)$ is
\begin{align}
\sT_{\vec \tta} = \Timesbig_{\sfi\in\ttQ_0^\sGamma} \, \sT_{\vec \tta_\sfi} \ ,
\end{align}
where $\sT_{\vec \tta_\sfi}$ is the maximal torus of $\sU(r_\sfi)$.
\end{remark}

\begin{remark}[{\bf Trivial Orbifold}]
The McKay quiver associated to the action of the trivial group $\sGamma=\ident$ on $\FC^3$ is the three-loop quiver $\mathsf{L}_3$:
\begin{align}
\begin{tikzcd}
	 {\bullet}\arrow[in=150,out=120,loop,swap]\arrow[out=60,in=30,loop,swap,] \arrow[out=280,in=250,loop,swap,]
\end{tikzcd}  
\end{align}
Its enhanced framed ADHM representation is
\begin{equation}\label{eq:ADHMquiverC3}
{\small
\begin{tikzcd}
	{\boxed{W}} && {\boxed{V}}\arrow["B_1"',out=150,in=120,loop,swap,]\arrow["B_2"',out=60,in=30,loop,swap,] \arrow["B_3"',out=330,in=300,loop,swap,]\arrow[dashed,"Y"',out=240,in=210,loop,swap,]
	\arrow["I", from=1-1, to=1-3] 
\end{tikzcd}  }\normalsize
\end{equation}

In this case the equations~\eqref{eq:ADHMrho1} and \eqref{eq:ADHMrho2} reduce to the ADHM equations~\eqref{eq:ADHM3d} of Section \ref{subsec:ADHMC3},  and \smash{$\frM^{\ident}_{ r,k}$}  is  the moduli space $\frM_{r,k}$ of rank $r$ noncommutative $k$-instantons on $\FC^3$. Note that \eqref{eq:ADHMquiverC3} is identical to the framed ADHM quiver representation for $\sSU(4)$-instantons on $\FC^4$~\cite{Szabo:2023ixw}. This is not a coincidence; it will be explained through the introduction of tetrahedron instantons in Section~\ref{sec:tetra_insta}.
\end{remark}

\subsubsection*{Stability and Quot Schemes}

A set of maps \eqref{eq:Bdecomp} and \eqref{eq:Idecomp} is said to be \emph{stable} if there are no proper $\sGamma$-submodules
\begin{align}
S = \bigoplus_{\sfi\in\ttQ_0^\sGamma} \, S_\sfi\otimes\rho_\sfi^*
\end{align}
of $V$ such that $B_e(S_{\sfs(e)})\subset S_{\sft(e)}$, $Y^\dagger_\sfi(S_\sfi)\subset S_\sfi$ and ${\rm im}(I_\sfi)\subset S_\sfi$, for all $\sfi\in\ttQ_0^\sGamma$ and $e\in\ttQ_1^\sGamma$. The stability condition is equivalent to the condition that the actions of the operators $B_e$ and \smash{$Y^\dagger_\sfi$} for $e\in\ttQ_1^\sGamma$ and $\sfi\in \ttQ_0^\sGamma$  on \smash{$I^{\sfi'}(W_{\sfi'})$} generate the  subspaces $V_{\sfi''}$.
Similarly to the proof of~\cite{Nekrasov:2015wsu}, we can show that the D-term equations $\mu_\sfi^{\FR}=\zeta_\sfi\,\ident_{V_\sfi}$ in \eqref{eq:ADHMrho1} for generic Fayet--Iliopoulos parameters $\zeta_\sfi>0$ can be traded for the stability condition. 

Let $\varPi_{\sfi}$ be the orthogonal projection of $V_\sfi$ to the orthogonal complement $S_\sfi^\perp$ of the invariant subspace $S_\sfi\subset V_\sfi$, for each $\sfi\in\ttQ_0^\sGamma$. Then $\varPi_{\sfi}\,I_\sfi=0$, $\varPi_{\sft(e)}\,B_e\,\varPi_{\sfs(e)}=\varPi_{\sft(e)}\,B_e$ and  $\varPi_\sfi\,Y_\sfi^\dagger\,\varPi_\sfi=\varPi_\sfi\,Y^\dagger_\sfi$, so
\begin{align}
\begin{split}
0 \ & \leq \ \sum_{\sfi\in\ttQ_0^\sGamma}\,\zeta_\sfi \dim S_\sfi^\perp \ = \ \sum_{\sfi\in\ttQ_0^\sGamma} \, \Tr_{V_\sfi}\big(\varPi_{\sfi}\,\mu^{\FR}_\sfi\big)  \\[4pt]
&= \ \sum_{e\in\ttQ_1^\sGamma}\Tr_{V_{\sft(e)}}\big(\varPi_{\sft(e)}\,B_e\,B_e^{ \dagger}-B_e\, \varPi_{\sfs(e)}\, B_{e}^{ \dagger}\big) + \sum_{\sfi\in\ttQ_0^\sGamma}\Tr_{V_\sfi} \big(\varPi_\sfi \,Y_\sfi^\dagger \,Y_\sfi -\varPi_\sfi \,Y_\sfi\, Y_\sfi^\dagger \big)\\[4pt]
&= \ -\sum_{e\in\ttQ_1^\sGamma} \  \big\|(\ident_{V_{\sft(e)}}-\varPi_{\sft(e)})\,B_e\,\varPi_{\sfs(e)}\big\|_{\textrm{\tiny F}}^2 - \sum_{\sfi\in\ttQ_0^\sGamma} \  \big\|(\ident_{V_{\sfi}}-\varPi_\sfi)\,Y^\dagger_\sfi\,\varPi_\sfi\big\|_{\textrm{\tiny F}}^2 \ \leq \ 0 \ .
\end{split}
\end{align}
This implies that $S_\sfi=V_\sfi$ for all $\sfi\in\ttQ_0^\sGamma$. 

The equations $\mu^\FC=0$ from \eqref{eq:ADHMrho2} arise  as the complex moment map equations $\mu^\FC_{ab}=0$ from \eqref{eq:ADHM3d} for the $\sGamma$-equivariant decomposition \eqref{eq:muC_decom}. Since the equations $\mu_{ab}^\FC=0$ are equivalent to the commuting relations $[B_a,B_b]=0$ and \smash{$[B_a,Y^\dagger]=0$}, we can replace \eqref{eq:ADHMrho2} with the equations
\begin{align}\label{eq:com_ADHM_nonab}
B\wedge B=0  \qquad \mbox{and} \qquad \langle B^\dagger,   Y\rangle_{Q_3} = \langle Y, B^\dagger\rangle_{Q_3}  \ .
\end{align}
In particular, the second equation in \eqref{eq:com_ADHM_nonab} implies
\begin{align}
B^\dagger_e\,Y_{\sft(e)} = Y_{\sfs(e)}\,B^\dagger_e \ ,
\end{align}
for all $e\in\ttQ_1^\sGamma$. Then the relations \eqref{eq:com_ADHM_nonab} and the equations $\sigma_\sfi=0$ from \eqref{eq:ADHMrho1} enable us to restate the stability condition as the condition that the actions of the operators $B_e$ for  $e\in \ttQ_1^\sGamma$  on \smash{$I^{\sfi}(W_{\sfi})$} generate the  subspaces $V_{\sfj}$. As a consequence, \smash{$Y^\dagger_\sfi=0 $} for all $\sfi\in\ttQ_0^\sGamma$.

The quiver variety \eqref{eq:Quiver_variety} may now be equivalently described as the noncommutative $\sGamma$-Quot scheme
\begin{align}
\frM_{\vec r,\vec k} \ \simeq \ \mu^\FC{}^{-1}(0)^{\rm stable} \, \big/ \, \sG_{\vec k} \ ,
\end{align}
where ${}^{\rm stable}$ designates the stable solutions of \eqref{eq:ADHMrho1} with $Y=0$, and
\begin{align}\label{eq:sGvec}
\sG_{\vec k}:=\Timesbig_{\ii\in\ttQ_0^\sGamma}\,\sGL(k_\ii,\FC)
\end{align}
is the complex gauge group of the $\sGamma$-module $V$, acting on the orbifold ADHM data as in \eqref{eq:U(k)action}. In this holomorphic description, the orbifold instanton moduli space \smash{$\frM_{r,k}^\sGamma$} parametrizes zero-dimension{-}al quotients of \smash{$\CO_{[\FC^3/\,\sGamma]}^{\oplus r}$} with length $k$. When $r=1$ these correspond to properly supported substacks of the orbifold resolution $[\FC^3/\,\sGamma]$, which may be regarded as zero-dimensional $\sGamma$-invariant closed subschemes of $\FC^3$.

\subsection{Non-Effective Orbifolds}\label{sec:orb_3d_inst}

The global symmetries of the cohomological gauge theory which are used to define equivariant instanton partition functions severely restrict the allowed $\sGamma$-actions. In order to preserve the holonomy, $\sGamma$ must be a subgroup of $\sU(3)$, whereas to preserve the maximal torus $\sT = \sT_{\vec \tta}\times\sT_{\vec\epsilon}$ it must commute with the action of the maximal torus $\sT_{\vec \epsilon}\subset\sU(3)$. These conditions force $\sGamma$ to be an abelian diagonally embedded subgroup of $\sU(3)$, and if $\sGamma\subset\sSU(3)$ it is of the form $\sGamma = \RZ_{n_1}\times\RZ_{n_2}$ with order $n=n_1\,n_2$. The orbifold instanton partition functions in this case have been thoroughly analysed in \cite{Cirafici:2010bd}. 

However, we can relax the condition that $\sGamma$ is an embedded subgroup of the holonomy group and consider generic finite abelian groups $\sGamma$ by defining the action of $\sGamma$ on $\FC^3$ via a homomorphism $\tau: \sGamma\longrightarrow \sU(3)$ whose image lies in the maximal torus $\sT_{\vec \epsilon}\,$; this provides a (not necessarily faithful) representation of $\sGamma$ in  the holonomy group.
Even more generally, we can still define an equivariant gauge theory for any finite group $\sGamma$ as long as the theory has a torus action commuting with the $\sGamma$-action, in order to enable the application of the virtual localization formula. The quotient stacks obtained from these more general quotients of $\FC^3$ yield `twisted' orbifold resolutions, in a sense which we momentarily explain.

Let 
\begin{align}
\tau:\sGamma\longrightarrow\sU(3)
\end{align}
be a homomorphism from a finite group $\sGamma$ to the holonomy group. Although $\sGamma$ is not necessarily a finite subgroup of $\sU(3)$, the image $\tau(\sGamma)$ is. To identify this subgroup, we note that the kernel
\begin{align}
\sK^\tau := \ker(\tau)
\end{align}
is a normal subgroup of $\sGamma$, and the First Isomorphism Theorem for groups implies
\begin{align} \label{eq:firstisothm}
\tau(\sGamma) \ \simeq \ \sGamma\,\big/\,\sK^\tau \ .
\end{align}
We write $\sGamma^\tau\,\embd\,\sU(3)$ for the embedding of $\sGamma/\,\sK^\tau$ in the holonomy group by the isomorphism \eqref{eq:firstisothm}. 

We can collect the finite groups introduced so far into a short exact sequence
\begin{align}\label{eq:Gammaexactseq}
1\longrightarrow \sK^\tau\longrightarrow \sGamma\longrightarrow \sGamma^\tau\longrightarrow 1 \ .
\end{align}
Whereas $\sGamma^\tau$ acts effectively on $\FC^3$, because it is represented faithfully in $\sU(3)$, its extension $\sGamma$ acts non-effectively, because the subgroup $\sK^\tau$ acts trivially on $\FC^3$ by construction. The extension \eqref{eq:Gammaexactseq} means that $\sGamma$ acts on $\FC^3$ by first projecting to $\sGamma^\tau$, and the quotient stack $[\FC^3/\,\sGamma]$ is a principal $\rmB\sK^\tau$-bundle, i.e. a $\sK^\tau$-gerbe, over $[\FC^3/\,\sGamma^\tau]$; the classifying stack
\begin{align}
\rmB\sK^\tau=[1\,/\,\sK^\tau]
\end{align}
may be presented as the delooping groupoid \smash{$\sK^\tau \doublerightarrow{\ \ }{ \ } 1$}. Since $\tau(\sGamma)\subset\sU(3)$, it is a K\"ahler $\sK^\tau$-gerbe.

To implement the quotient by the trivially acting kernel in \eqref{eq:firstisothm}, we take the semi-direct product of the groupoids presenting $[\FC^3/\,\sGamma]$ and $\rmB\sK^\tau$~\cite{Mathoverflow}. Since the $\sK^\tau$-action on $\FC^3$ is trivial, this is just the direct product $[\FC^3/\,\sGamma]\times\rmB\sK^\tau$, which has the same orbit space as the quotient stack $[\FC^3/\,\sGamma^\tau]$. Hence we regard it as a `twisted' orbifold resolution 
\begin{align}\label{eq:twistedstack}
\pi_{\rm orb}^{\tt tw}:\big[\FC^3\,\big/\,\sGamma\big]\,\times\,\rmB\sK^\tau \ \longrightarrow \ \FC^3\,\big/\,\sGamma^\tau
\end{align}
of the quotient singularity $\FC^3/\,\sGamma^\tau$.

The cohomological gauge theory on the non-effective orbifold $[\FC^3/\,\sGamma]$ is \emph{not} the same as the theory on $[\FC^3/\,\sGamma^\tau]$, even though the kernel $\sK^\tau\subset\sGamma$ acts trivially on $\FC^3$: gauging a non-effective group action is not equivalent to gauging an effective group action~\cite{Pantev:2005rh}. This will become evident in our ensuing constructions of quiver varieties below, as well as in explicit computations of orbifold instanton partition functions. If $\textsf{Z}(\sK^\tau)$ denotes the centre of $\sK^\tau$, then the $2$-group $\rmB\textsf{Z}(\sK^\tau)$ acts on $\rmB\sK^\tau$ and on $[\FC^3/\,\sGamma^\tau]$.
In the modern language of generalized global symmetries~\cite{Gaiotto:2014kfa}, in addition to $\sGamma^\tau$-equivariance, the fields of the non-effective orbifold theory are equivariant under a one-form symmetry corresponding to the action of $\rmB\textsf{Z}(\sK^\tau)$ by translations along the fibres of the $\sK^\tau$-gerbe $[\FC^3/\,\sGamma]$ over~$[\FC^3/\,\sGamma^\tau]$.

We interpret this theory as the orbifold Donaldson--Thomas theory of $[\FC^3/\,\sGamma^\tau]$ twisted by a $\sK^\tau$-gerbe, which computes the ordinary (untwisted) Donaldson--Thomas invariants of $\FC^3/\,\sGamma^\tau$ if and only if $\tau$ is a monomorphism. This is supported by the general statement~\cite{Pantev:2005wj} that a sheaf on a gerbe is the same thing as a  twisted sheaf on the underlying base. Whence a $\sGamma$-equivariant coherent sheaf on $\FC^3$ is a $\sK^\tau$-projectively $\sGamma^\tau$-equivariant coherent sheaf on $\FC^3$. Similarly to~\cite{GholampourTseng}, where the Donaldson--Thomas invariants of gerbes over projective Calabi--Yau orbifolds are studied in this setting, we shall find that the cohomological gauge theory on the $\sK^\tau$-gerbe $[\FC^3/\,\sGamma]$ is equivalent to a suitable twist of the cohomological gauge theory on a disjoint union of $\#\sK^\tau$ copies of the base~$[\FC^3/\,\sGamma^\tau]$.

This picture is in agreement with the structure of boundary states of D-branes in non-effective orbifolds, which is discussed in~\cite{Pantev:2005rh,Pantev:2005wj}. Although the subgroup $\sK^\tau$ acts trivially on $\FC^3$, in general $\sK^\tau$ can act non-trivially on the Chan--Paton bundles, as in \eqref{eq:VWGammaaction}. This is consistent as long as one distinguishes boundary states in each of the twisted sectors corresponding to conjugacy classes in \smash{$\widehat{\sK}{}^\tau$}.  The combination of a trivial $\sK^\tau$-action on $\FC^3$ and a non-trivial $\sK^\tau$-action on the Chan--Paton bundles means that the worldvolume theories of D-branes on the non-effective orbifold support twisted gauge bundles, as in the more familiar cases of D-branes in flat non-trivial $B$-field backgrounds~\cite{Freed:1999vc}.

\begin{remark}[{\bf Banded Gerbes}]
If $\sGamma$ is a central extension of $\sGamma^\tau$, then $\textsf{Z}(\sK^\tau)=\sK^\tau$ and $[\FC^3/\,\sGamma]\longrightarrow[\FC^3/\,\sGamma^\tau]$ is a \emph{banded} $\sK^\tau$-gerbe. It is trivial if and only if the principal $\sGamma^\tau$-bundle $\FC^3\longrightarrow[\FC^3/\,\sGamma^\tau]$ has a lift to a principal $\sGamma$-bundle on $[\FC^3/\,\sGamma^\tau]$~\cite{Pantev:2005wj}.
\end{remark}

\begin{notation} \label{not:3dorbcases}
We write $\Ab$ for a general finite abelian group. It takes the form
\begin{align}
\Ab=\Timesbig_{i=1}^p\,\RZ_{n_i} \ ,
\end{align}
with order $n=n_1\cdots n_p$. The set $\Abw$ is also an abelian group, isomorphic to $\Ab$, under the tensor product of irreducible representations.

If we use a homomorphism $\tau$ to represent the action of a finite non-abelian group $\sGamma$ on $\FC^3$, then there are two possible classes of groups which are of the form
\begin{align}
\sGamma_m = \mathsf{\Upsilon}_m \times \Ab \ , 
\end{align}
 where $\sUps_m$ is a finite non-abelian  subgroup of $\sSU(m)$ for $m=2,3$. We call the corresponding twisted orbifold resolution of the quotient singularity $\FC^3/\,\sGamma_m$ an \emph{$\sSU(m)\,\times\,$abelian orbifold}. 
In contrast to the abelian orbifolds based on $\Ab$ alone, neither of these orbifold actions commute with the maximal torus $\sT_{\vec \epsilon}\,$. 

In general, the maximal torus is broken to the centralizer $\sC^{\tau}$ of the image of $\sGamma$ in $\sT_{\vec \epsilon}$ under the homomorphism $\tau$. The maximal torus of the equivariant gauge theory thus becomes 
\begin{align}
\sT^{\tau}= \sT_{\vec \tta}\times \sC^{\tau} \ . 
\end{align}
\end{notation}

We proceed to study each of the three cases of Notation~\ref{not:3dorbcases} in turn.

\subsubsection*{Abelian Orbifolds}

Since the irreducible representations of an abelian group are all one-dimensional, the most general representation of $\Ab$ in $\sU(3)$ which commutes with the maximal torus $\sT_{\vec \epsilon}\,$ is through a homomorphism $\tau_{\vec s} :\, \Ab\longrightarrow \sU(3)$ specified by a triple of weights $\vec s=(s_1,s_2,s_3)$. It is defined by
\begin{align} \label{eq:tauabelian}
\tau_{\vec s\,}(\Ab)=
\rho_{s_1}(\Ab)\times\rho_{s_2}(\Ab)\times\rho_{s_3}(\Ab) \ \subset \ \sU(1)^{\times 3}
  \ \subset \ \sU(3) \ ,
\end{align}
where $\rho_s: \,\Ab\longrightarrow \sU(1)$ is the unitary irreducible representation of $\Ab$ with weight $s$. This defines the action of $\Ab$ on $\FC^3$ as the three-dimensional $\Ab$-module
\begin{align}
Q_3^{\vec s} = \rho_{s_1}\oplus\rho_{s_2}\oplus\rho_{s_3} \ .
\end{align}
The kernel of $\tau_{\vec s}$ is the subgroup of $\Ab$ given by
\begin{align}
\sK^{\vec s} := \ker(\tau_{\vec s}) = \ker(\rho_{s_1})\,\cap\,\ker(\rho_{s_2})\,\cap\,\ker(\rho_{s_3}) 
\end{align}

The McKay quivers $\ttQ^{\tau_{\vec s\,}(\Ab)}$ are built similarly to the construction in Section \ref{sec:quiver_variety}.  To each irreducible representation of $\Ab$ we associate a vertex $s\in\widehat{\sGamma}_{\textsf{ab}}$. Using 
\begin{align}
\rho_s\otimes\rho_{s'}\,\simeq\,\rho_{s+s'} \qquad \mbox{and} \qquad \rho_s^*\,\simeq\,\rho_{-s} \ , 
\end{align}
the number $a_{ss'}$ of arrows connecting vertex $s$ to vertex $s'$ is determined by the tensor product decomposition of $\Ab$-modules
\begin{align}
Q_3^{\vec s}\otimes\rho_s= \rho_{s_1+s}\oplus\rho_{s_2+s}\oplus\rho_{s_3+s}
\end{align}
to be
\begin{align} \label{eq:edgeAb}
a_{ss'} = \delta_{s',s+s_1} + \delta_{s',s+s_2} + \delta_{s',s+s_3} \ .
\end{align}

\begin{example}\label{ex:AbelianGamma}
Let 
$
\Ab=\RZ_2\times\RZ_2\times\RZ_2
$
be represented in $\sSO(3)\subset\sU(3)$  as 
\begin{align}
\tau_{\vec s\,}(\RZ_2^{\times 3})=\rho_{(1,0,0)}(\RZ_2^{\times 3})\times\rho_{(0,1,0)}(\RZ_2^{\times 3})\times\rho_{(1,1,0)}(\RZ_2^{\times 3}) \ ,
\end{align}
with
\begin{align}
\rho_{(l_1,l_2,l_3)}\big(\xi_2^{n_1},\xi_2^{n_2},\xi_2^{n_3}\big)=\e^{\,\pi\,\sfi\,(n_1\,l_1+n_2\,l_2+n_3\,l_3)} \ , 
\end{align}
where $\xi_2$ is the generator of $\RZ_2$ and $l_i,n_i\in\{0,1\}$. The kernel of $\tau_{\vec s}$ is
\begin{align}
\sK^{\vec s} = \RZ_2 \ ,
\end{align}
generated by $(1,1,\xi_2)\in\RZ_2^{\times 3}$. It follows from \eqref{eq:firstisothm} that the image of $\tau_{\vec s}$ 
\begin{align}
\big(\RZ_2^{\times 3}\big)^{\vec s\,} \ \simeq \ \RZ_2\times\RZ_2
\end{align}
is the Klein four-group in $\sSO(3)$, generated by $(\xi_2,1,1)$ and $(1,\xi_2,1)$.

The generalized McKay quiver $\ttQ^{\tau_{\vec s}\,(\RZ_2^{\times 3})}$ is the disconnected quiver
\begin{equation}
\scriptsize
\begin{tikzcd}
	{(0,1,0)} && {(1,0,0)} & {(1,0,1)} && {(1,1,1)} \\
	& {(0,0,0)} &&& {(0,0,1)} \\
	& {(1,1,0)} &&& {(0,1,1)}
	\arrow[shift left=1, from=1-1, to=1-3]
	\arrow[shift left=1, from=1-3, to=1-1]
	\arrow[shift right=1, curve={height=-6pt}, from=1-3, to=2-2]
	\arrow[shift right=1, curve={height=-6pt}, from=2-2, to=1-3]
	\arrow[curve={height=6pt}, from=1-1, to=2-2]
	\arrow[shift left=1, curve={height=6pt}, from=2-2, to=1-1]
	\arrow[shift left=1, from=1-3, to=3-2]
	\arrow[shift right=2, curve={height=6pt}, from=3-2, to=1-3]
	\arrow[shift left=2, curve={height=-6pt}, from=3-2, to=1-1]
	\arrow[shift right=1, from=1-1, to=3-2]
	\arrow[shift left=1, from=2-2, to=3-2]
	\arrow[shift left=1, from=3-2, to=2-2]
	\arrow[shift left=1, from=1-4, to=1-6]
	\arrow[shift left=1, from=1-6, to=1-4]
	\arrow[shift left=1, from=3-5, to=2-5]
	\arrow[shift left=1, from=2-5, to=3-5]
	\arrow[shift right=1, curve={height=-6pt}, from=2-5, to=1-6]
	\arrow[shift right=1, curve={height=-6pt}, from=1-6, to=2-5]
	\arrow[shift left=1, curve={height=6pt}, from=2-5, to=1-4]
	\arrow[shift left=1, curve={height=6pt}, from=1-4, to=2-5]
	\arrow[shift left=1, from=1-6, to=3-5]
	\arrow[shift right=2, curve={height=6pt}, from=3-5, to=1-6]
	\arrow[shift right=1, from=1-4, to=3-5]
	\arrow[shift left=2, curve={height=-6pt}, from=3-5, to=1-4]
\end{tikzcd}
\normalsize
\end{equation}
Thus 
\begin{align}
\ttQ^{\tau_{\vec s}\,(\RZ_2^{\times 3})} = \ttQ^{\RZ_2^{\times 2}} \, \sqcup \, \ttQ^{\RZ_2^{\times 2}} \ , 
\end{align}
where $\ttQ^{\RZ_2^{\times 2}}$ is the McKay quiver  for the toric Calabi--Yau three-orbifold \smash{$\FC^3/\RZ_2\times\RZ_2$} considered in~e.g.\cite[Section~6]{Cirafici:2010bd}. This represents the twisted orbifold resolution
\begin{align}
\pi_{\rm orb}^{\tt tw}:\big[\FC^3\,\big/\,\RZ_2\times\RZ_2\times\RZ_2\big]\,\times\,\rmB\,\RZ_2 \ \longrightarrow \ \FC^3\,\big/\,\RZ_2\times\RZ_2 \ .
\end{align}
\end{example}

Since $\tau_{\vec s\,}(\Ab)$ commutes with $\sT_{\vec\epsilon}=\sU(1)^{\times 3}$, the maximal torus of the equivariant gauge theory is unbroken and is again
\begin{align}
\sT = \sT_{\vec \tta} \times \sT_{\vec \epsilon} \ .
\end{align}
We study this gauge theory in detail in Section~\ref{sec:orb_3d_pf}.

\subsubsection*{$\boldsymbol{\sSU(2)\,\times}\,$Abelian Orbifolds}

Let $\sGamma_2= \sUps_2\times \sGamma_{\textsf{ab}}$, where $\sUps_2$ is a finite non-abelian  subgroup of $\sSU(2)$ acting on $\FC^2$ in the fundamental representation $Q_2$. Let $\sGamma_2$
 act on $\FC^3$ via the homomorphism $\tau_{\vec s} :\, \sGamma_2\longrightarrow \sU(3)$ defined by
\begin{align}\label{eq:A_action}
\tau_{\vec s\,}(\sGamma_2)= \big(
\sUps_2\times \rho_{s_1}(\sGamma_{\textsf{ab}})\big) \, \times \, \rho_{s_2}(\sGamma_{\textsf{ab}}) \ \subset \ \sU(2) \times \sU(1) \ \subset \ \sU(3) \ , 
\end{align}
where $\vec s=(s_1,s_2)$ and as before $\rho_s: \,\sGamma_{\textsf{ab}}\longrightarrow \sU(1)$ is the unitary irreducible  representation of $\Ab$ with weight $s$. This defines the action of $\sGamma_2$ on $\FC^3$ as the three-dimensional $\sGamma_2$-module
\begin{align}
Q_3^{\vec s} = (Q_2\otimes\rho_{s_1})\,\oplus\,(\lambda_0\otimes\rho_{s_2}) \ ,
\end{align}
where $\lambda_0$ is the trivial one-dimensional representation of $\sUps_2$. The kernel of $\tau_{\vec s}$ is the normal subgroup
\begin{align}
\sK^{\vec s} := \ker(\tau_{\vec s}) = \big\{(g,\xi)\in\sUps_2\times\Ab \ \big| \ g=\rho_{-s_1}(\xi)\,\ident_2\in\sUps_2 \ , \ \xi\in\ker(\rho_{s_2})\big\}
\end{align}
of $\sGamma_2$.

When $\vec s=(0,0)$, the kernel is $\sK^{(0,0)}=\ident_2\times\Ab$ and the McKay quiver $\ttQ^{\tau_{(0,0)}(\sGamma_2)}=\ttQ^{\sUps_2}$ is constructed in Example~\ref{ex:gammaSU(2)} as the double of an oriented affine Dynkin diagram \smash{$\mathsf{Dynk}_{\sUps_2}$} of type~ADE. 

When $\vec s\neq(0,0)$, the associated McKay quiver \smash{$\ttQ^{\tau_{\vec s\,}(\sGamma_2)}$} is formed from $\#\sGamma_{\textsf{ab}}$ copies of the vertices of the McKay quiver \smash{$\ttQ^{\sUps_2}$}, one for each irreducible representation $\rho_s$ of $\sGamma_{\textsf{ab}}$. An irreducible representation
\begin{align}
{ \CR}_{(\sfi,s)}=\lambda_\sfi\otimes\rho_s
\end{align}
of $\sGamma_2$ is labelled by a pair $(\sfi,s)$, where $\sfi\in\ttQ_0^{\sUps_2}$ labels an irreducible representation $\lambda_\sfi$ of $\sUps_2$ and \smash{$s\in\widehat\sGamma_{\textsf{ab}}$}. The number $a_{(\sfi,s)\,(\sfi',s')}$ of arrows from the vertex $(\sfi,s)$ to the vertex $(\sfi',s')$ in \smash{$\ttQ^{\tau_{\vec s\,}(\sGamma_2)}$} is determined by the tensor product decomposition of $\sGamma_2$-modules
\begin{align}
Q_3^{\vec s}\otimes\CR_{(\sfi,s)}=\bigoplus_{\sfi' \in \ttQ_0^{\sUps_2}} \, a^{\sUps_2}_{\sfi\sfi'} \ \CR_{(\sfi',s_1+s)} \ \oplus \ \CR_{(\sfi,s_2+s)} \ ,
\end{align} 
where \smash{$A_{\sUps_2}=\big(a^{\sUps_2}_{\sfi\sfi'}\big)$} is the adjacency matrix of the simply-laced extended Dynkin diagram corresponding to $\sUps_2$. Thus
\begin{align}
a_{(\sfi,s)\,(\sfi',s')} = a^{\sUps_2}_{\sfi\sfi'} \ \delta_{s',s+s_1} + \delta_{\sfi',\sfi} \ \delta_{s',s+s_2} \ .
\end{align}

\begin{example}\label{ex:nofixedpoint}
Let $\sUps_2={\mathbbm{S}}^*_3\subset\sSU(2)$ be the generalized quaternion group of order $12$; this is the binary extension of the symmetric group $\mathbbm{S}_3\subset\sSO(3)$ of degree three, which corresponds to the dihedral group $\mathbbm{D}_3 = \RZ_3\rtimes\RZ_2$ of the triangle in the ADE classification.
It has a pair of two-dimensional irreducible representations, $\lambda_2=Q_2$ and $\lambda_3$, and four one-dimensional irreducible representations, $\lambda_0$, $\lambda_1$, $\lambda_4$ and $\lambda_5$. 
Given an orientation for the affine Dynkin diagram of type $\sD_5$, the McKay quiver \smash{$\ttQ^{{\mathbbm{S}}^*_3}$} is
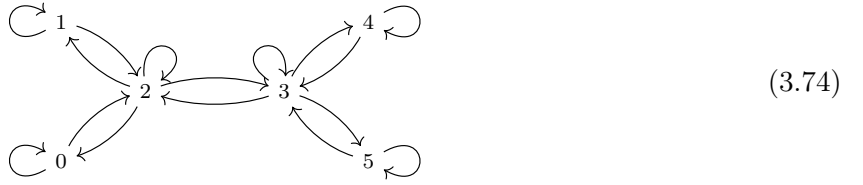
\begin{equation} \label{eq:D3*quiver}
\scriptsize
\begin{tikzcd}
	1 &&&& 4 \\
	& 2 && 3 \\
	0 &&&& 5
	\arrow[ curve={height=-6pt}, from=3-1, to=2-2]
	\arrow[ curve={height=-6pt}, from=2-2, to=3-1]
	\arrow[ curve={height=-6pt}, from=1-1, to=2-2]
	\arrow[ curve={height=-6pt}, from=2-2, to=1-1]
	\arrow[ curve={height=-6pt}, from=2-2, to=2-4]
	\arrow[curve={height=-6pt}, from=2-4, to=1-5]
	\arrow[curve={height=-6pt}, from=3-5, to=2-4]
	\arrow[curve={height=-6pt}, from=2-4, to=3-5]
	\arrow[curve={height=-6pt}, from=1-5, to=2-4]
	\arrow[curve={height=-6pt}, from=2-4, to=2-2]
	\arrow[out=210,in=150,loop,swap,from=3-1,to=3-1]
	\arrow[out=30,in=330,loop,swap,from=1-5,to=1-5]
		\arrow[out=30,in=330,loop,swap,from=3-5,to=3-5]
	\arrow[out=210,in=150,loop,swap,from=1-1,to=1-1]
		\arrow[out=95,in=35,loop,swap,from=2-2,to=2-2]
		\arrow[out=145,in=85,loop,swap,from=2-4,to=2-4]
\end{tikzcd}
\normalsize
\end{equation}

Let $\sGamma_2={\mathbbm{S}}^*_3\times \RZ_2$,  acting on $\FC^3$ as in \eqref{eq:A_action} with weights $s_1=1$ and $s_2=0$. The kernel of $\tau_{(1,0)}$ is
\begin{align}
\sK^{(1,0)} = \RZ_2\times\RZ_2 \ ,
\end{align}
embedded as the central subgroup \smash{$\{\pm\ident_2\}\times\RZ_2\subset {\mathbbm{S}}^*_3\times\RZ_2$}. It follows from \eqref{eq:firstisothm} that the image of $\tau_{(1,0)}$ is given by
\begin{align}
({\mathbbm{S}}^*_3\times\RZ_2)^{(1,0)} \ \simeq \ \mathbbm{S}_3 \ ,
\end{align}
where $\mathbbm{S}_3={\mathbbm{S}}^*_3\,/\,\RZ_2$ under the double covering
\begin{align} \label{eq:SU2SO3}
\begin{split}
\xymatrix{
{\mathbbm{S}}^*_3 \ar[d] \lhook\joinrel\!\ar[r] & \sSU(2) \ar[d] \\
\mathbbm{S}_3 \lhook\joinrel\!\ar[r] & \sSO(3)
}
\end{split}
\end{align}

The McKay quiver $\ttQ^{\tau_{(1,0)}({\mathbbm{S}}^*_3\times \RZ_2)}$ is 
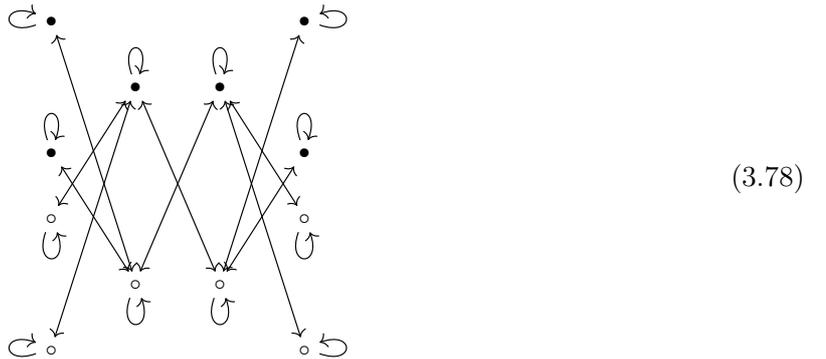
\begin{equation} \label{eq:D3Z2quiver}
{\scriptsize
\begin{tikzcd}
	{\bullet}\arrow[,out=200,in=160,loop,swap,] &&& {\bullet}\arrow[,out=340,in=20,loop,swap,]\\
	& {\bullet}\arrow[,out=110,in=70,loop,swap,] & {\bullet}\arrow[,out=110,in=70,loop,swap,] \\
	{\bullet}\arrow[,out=110,in=70,loop,swap,] &&& {\bullet}\arrow[,out=110,in=70,loop,swap,] \\
	{\circ}\arrow[,out=250,in=290,loop,swap,] &&& {\circ}\arrow[,out=250,in=290,loop,swap,] \\
	& {\circ}\arrow[,out=250,in=290,loop,swap,] & {\circ} \arrow[,out=250,in=290,loop,swap,]\\
	{\circ}\arrow[,out=200,in=160,loop,swap,] &&& {\circ}\arrow[,out=340,in=20,loop,swap,]
	\arrow[tail reversed, from=4-1, to=2-2]
	\arrow[tail reversed, from=2-2, to=5-3]
	\arrow[tail reversed, from=5-3, to=1-4]
	\arrow[tail reversed, from=5-3, to=3-4]
	\arrow[tail reversed, from=4-4, to=2-3]
	\arrow[tail reversed, from=6-4, to=2-3]
	\arrow[tail reversed, from=6-1, to=2-2]
	\arrow[tail reversed, from=5-2, to=2-3]
	\arrow[tail reversed, from=3-1, to=5-2]
	\arrow[tail reversed, from=1-1, to=5-2]
\end{tikzcd}
}
\normalsize
\end{equation}
where the empty (filled) vertices carry the irreducible representation $\rho_0$ ($\rho_1$) of $\RZ_2$. This represents the twisted orbifold resolution 
\begin{align}
\pi_{\rm orb}^{\tt tw}:\big[\FC^3\,\big/\,{\mathbbm{S}}^*_3\times\RZ_2\big] \, \times \, \rmB\,\RZ_2^{\times 2} \ \longrightarrow \ \FC^3\,\big/\,\mathbbm{S}_3
\end{align}
of the dihedral singularity $\FC^3\,/\,\mathbbm{S}_3$, whose standard noncommutative Donaldson--Thomas theory is studied in~\cite{Mozgovoy:2022gcb}.

Indeed, the quiver \eqref{eq:D3Z2quiver} is a four-cover lift of the McKay quiver $\ttQ^{\mathbbm{S}_3}$:
\begin{equation}
{\scriptsize
\begin{tikzcd}
	1 \\
	&& 2 \\
	0
	\arrow[curve={height=6pt}, from=3-1, to=2-3]
	\arrow[curve={height=6pt}, from=2-3, to=3-1]
	\arrow[curve={height=6pt}, from=2-3, to=1-1]
	\arrow[curve={height=6pt}, from=1-1, to=2-3]
	\arrow[curve={height=-6pt}, from=1-1, to=3-1]
	\arrow[curve={height=-6pt}, from=3-1, to=1-1]
	\arrow[out=20,in=80,loop,swap,from=2-3,to=2-3]
	\arrow[out=340,in=280,loop,swap,from=2-3,to=2-3]
\end{tikzcd}
}
\normalsize
\end{equation}
which can also be obtained from \eqref{eq:D3*quiver} by removing representations of $\sSU(2)$ which are not pullbacks of representations of $\sSO(3)$ by \eqref{eq:SU2SO3}.
\end{example}

The centralizer of $\tau_{\vec s\,}(\sGamma_2)$ is
\begin{align}
\sC^{{\vec s\,}} =\sU(1)_{\vec\epsilon}^{\times 2} \ \subset \ \sT_{\vec\epsilon}=\sU(1)^{\times 3} \ ,
\end{align}
consisting of diagonal matrices
\begin{align}
\begin{pmatrix}
t_1\,\ident_2 & \\
& t_2
\end{pmatrix} \ ,
\end{align} 
where $\vec\epsilon = (\epsilon_1,\epsilon_2)$ are the equivariant parameters and $t_a=\e^{\,\ii\,\epsilon_a}$. It follows that the maximal torus of the equivariant gauge theory is broken to
\begin{align}
\sT^{{\vec s\,}}= \sT_{\vec \tta}\times\sU(1)_{\vec\epsilon}^{\times 2} \ .
\end{align}
In the notation of Example~\ref{ex:gammaSU(2)}, the action of $\sU(1)^{\times 2}_{\vec\epsilon}$ on the ADHM data $(B,\bar B, L, I,Y)$ is 
\begin{align}
(B,\bar B,L,I,Y) \longmapsto \big(t_1^{-1}\, B \, , \, t_1^{-1}\,	\bar B \, , \,
t_2^{-1}\, L \,,\, I \,,\, t_1^{-2}\,t_2^{-1}\,  Y \big) \ .
\end{align}

\begin{remark}[{\bf $\boldsymbol{\sSU(3)}$-Holonomy}]
Restricting the holonomy to $\sSU(3)\subset\sU(3)$ imposes the constraint
\begin{align}
\rho_{\,2s_1}(\xi) \, \rho_{s_2} (\xi)=1 \ ,
\end{align}
for all $\xi\in\sGamma_{\textsf{ab}}$.
This implies $2s_1+s_2=0$. In this sector of the equivariant gauge theory the centralizer is broken to $\sU(1)_\epsilon$ with $\epsilon:=\epsilon_2=-2\,\epsilon_1$.
\end{remark}

\subsubsection*{$\boldsymbol{\sSU(3)\,\times\,}$Abelian Orbifolds}

Let $\sGamma_3=\sUps_3\times\sGamma_{\textsf{ab}}$, where $\sUps_3$ is a finite non-abelian subgroup of $\sSU(3)$ acting on $\FC^3$ in the fundamental representation $Q_3$. Let $\sGamma_3$ act on $\FC^3$ by the homomorphism $\tau_{\tilde s}:\sGamma_3\longrightarrow \sU(3)$ defined by
\begin{align}\label{Gamma_2_action}
\tau_{\tilde s}(\sGamma_3)\, = \,
{\sUps_3}\times \rho_{\tilde s}(\sGamma_{\textsf{ab}}) \ \subset \ \sU(3) \ .
\end{align}
This defines the action of $\sGamma_3$ on $\FC^3$  as the three-dimensional $\sGamma_3$-module
\begin{align}
Q_3^{\tilde s} = Q_3\otimes\rho_{\tilde s} \ .
\end{align}
The kernel of $\tau_{\tilde s}$ is the normal subgroup of $\sGamma_3$ given by
\begin{align}
\sK^{\tilde s} := \ker(\tau_{\tilde s}) = \big\{(g,\xi)\in\sUps_3\times\Ab \ \big| \ g=\rho_{-\tilde s}(\xi)\,\ident_3\in\sUps_3 \big\}
\end{align}

The McKay quiver $\ttQ^{\tau_{\tilde s}(\sGamma_3)}$ is constructed in a completely analogous way to the McKay quivers for the $\sSU(2)\,\times\,$abelian orbifolds above, starting from the general construction of the McKay quivers $\ttQ^{\sUps_3}$ for finite subgroups $\sUps_3\subset\sSU(3)$ discussed in Section~\ref{sec:quiver_variety}. Again there are $\#\Ab$ copies of the vertices of $\ttQ^{\sUps_3}$, labelled by \smash{$(\ii,s)\in \ttQ_0^{\sUps_3}\times\Abw$}. The number of arrows $a_{(\sfi,s)\,(\sfi',s')}$ from the vertex $(\sfi,s)$ to the vertex $(\sfi',s')$ in \smash{$\ttQ^{\tau_{\tilde s}(\sGamma_3)}$} is given by
\begin{align}
a_{(\sfi,s)\,(\sfi',s')} = a^{\sUps_3}_{\sfi\sfi'} \ \delta_{s',s+\tilde s} \ ,
\end{align}
where \smash{$A_{\sUps_3}=\big(a^{\sUps_3}_{\sfi\sfi'}\big)$} is the adjacency matrix of $\ttQ^{\sUps_3}$.

The centralizer of $\tau_{\tilde s}(\sGamma_3)$ is
\begin{align}
\sC^{\tilde s} =\sU(1)_\epsilon \ \subset  \ \sT_{\vec\epsilon} \ , 
\end{align} 
consisting of diagonal matrices $t\,\ident_3$ where $t=\e^{\,\ii\,\epsilon}$. 
The maximal torus of the equivariant gauge theory is thereby broken to 
\begin{align}
\sT^{\tilde s}=\sT_{\vec \tta}\times\sU(1)_\epsilon \ .
\end{align} 
In the notation of Section~\ref{sec:quiver_variety}, the torus $\sU(1)_\epsilon$ acts on the ADHM data $(B,I,Y)$ as
\begin{align}
(B,I,Y) \longmapsto \big(t^{-1}\, B \,,\, I \,,\,  t^{-3}\,  Y\big) \ .
\end{align}

Unlike the previous case of $\sSU(2)\,\times\,$abelian orbifolds, it is not possible to define a twisted Calabi--Yau quotient stack for the cohomological gauge theory. Indeed, if the holonomy group were reduced to $\sSU(3)$ then the centralizer would be $\RZ_3$, and the only $\RZ_3$-fixed point is $B=Y=0$.

\subsection{Abelian Orbifold Partition Functions}\label{sec:orb_3d_pf}

In the remainder of this section we  define and evaluate the instanton partition functions for the twisted quotient stacks $[\FC^3\,/\,\Ab]\times\rmB\sK^{\vec s}$, where $\Ab$ is a generic finite abelian group acting on $\FC^3$ by the homomorphism $\tau_{\vec s}:\Ab\longrightarrow\sU(3)$ defined in \eqref{eq:tauabelian}. Then the image $\tau_{\vec s\,}(\Ab)$ commutes with $\sT_{\vec \epsilon\,}$, and both the holonomy group $\sU(3)$ as well as the maximal torus $\sT=\sT_{\vec \tta}\times \sT_{\vec \epsilon}$ are preserved. This generalizes the cases where $\tau_{\vec s}\,$ is a monomorphism ($\sK^{\vec s}=\ident$) and $\Ab$ is a finite abelian subgroup of $\sSL(3,\FC)$, which were exhaustively discussed in \cite{Cirafici:2010bd}. 

We do not treat the non-abelian $\sSU(m)\,\times\,$abelian orbifolds in this section. They will appear in Section~\ref{sec:Tetra_orb} as special cases in our treatment of tetrahedron instantons on orbifolds.

\subsubsection*{Moduli Spaces}

We follow closely the construction of quiver varieties from Section~\ref{sec:quiver_variety}, except that now we carefully relax the Calabi--Yau condition of $\sSU(3)$ holonomy by using the isomorphisms of $\Ab$-modules
\begin{align} \label{eq:wedge3Qiso}
\midwedge^2 Q_3^{\vec s} \, \simeq \,\rho_{s_{12}}\oplus\rho_{s_{13}}\oplus\rho_{s_{23}} \qquad \mbox{and} \qquad \midwedge^3Q_3^{\vec s} \, \simeq \, \rho_{s_{123}} \ ,
\end{align}
where we introduced the shorthand notation
\begin{align}
s_{ab\cdots} = s_a+s_b+\cdots \ .
\end{align}
The vertices of the McKay quiver $\ttQ^{\tau_{\vec s\,}(\Ab)}$ are labelled by the weights $s\in\Abw$ of irreducible representations $\rho_s$ of $\Ab$, while the edge structure is given by \eqref{eq:edgeAb}. 

It follows that the isotypical components of the ADHM data $(B,I,Y)$ are linear maps
\begin{align}
B^s_a: V_{s}\longrightarrow V_{s+s_a} \ , \quad I^s: W_{s}\longrightarrow V_s \qquad \text{and} \qquad Y^s: V_{s}\longrightarrow V_{s+s_{123}} \ ,
\end{align}
for $a\in\ulthree\,$. 
In the quiver picture, the maps \smash{$(B_a^s, Y^s,I^s)_{s\in\Abw\,,\,a\,\in\,\ulthree}$} constitute the field content of the quiver variety \smash{$\frM_{\vec r, \vec k}$}. They are required to satisfy the  $\Ab$-equivariant version of the ADHM equations \eqref{eq:ADHM3d}, where the isotypical components of \smash{$(\mu^\FC,\mu^\FR,\sigma)$} are linear maps
\begin{align}
\mu_{ab}^{\FC s}:V_s\longrightarrow V_{s+s_{ab}} \ , \quad \mu^{\FR s}:V_s\longrightarrow V_s \qquad \text{and} \qquad \sigma^s:V_s\longrightarrow W_{s+s_{123}} \ ,
\end{align}
for $(a,b)\in\ulthree^\perp$. 

The component equations then read
\begin{align}\label{eq:complex_C3_orb}
\begin{split}
    \mu^{\FC s}_{ab} &= B_a^{s+s_b}\,B_b^s-B_b^{s+s_a}\,B_a^s - \tfrac{1}{2}\,\varepsilon_{abc}\, \big( B_c^{s+s_{ab}\dagger }\, Y^s- Y^{s-s_c}\, B_c^{s-s_c\dagger }\big)=0 \ , \\[4pt]
   \mu^{\FR s} &= \sum_{a\,\in\,\ulthree}\, \big( B_a^{s-s_a}\,B_a^{ s-s_a \dagger}-B_a^{s \dagger }\,B_a^{s}\big)+Y^{s\dag}\,Y^s - Y^{s-s_{123}}\,Y^{s-s_{123}\dag}+I^s\,I^{ s \dagger}=\zeta_s\,\ident_{V_s} \ , \\[4pt]
  \sigma^{s} &=  I^{s+s_{123}\dagger}\, Y^{s}=0 \ .
    \end{split}
\end{align}
These equations are invariant under the action of unitary automorphisms $g\in\sU(\vec k\,)$ given by
\begin{align}\label{eq:U(k)actionab}
 g\cdot \big(B_{a}^s\,,\, I^s\,,\,Y^s\big)_{s\in\Abw\,,\,a\,\in\,\ulthree} = \big(g_{s+s_a}\,B_{a}^s\,g_{s}^{-1}\,, \, g_s \, I^s \,,\, \, g_{s+s_{123}} \, Y^s\,g_s^{-1}\big)_{s\in\Abw\,,\,a\,\in\,\ulthree}    \ .
\end{align}
As in Section~\ref{sec:quiver_variety}, the D-term relation can be traded for a stability condition. 

The maximal torus $\sT_{\vec \tta}$ of the usual framing symmetry from Remark~\ref{rem:framingrot} gives rise to an equivariant decomposition of the Coulomb moduli \smash{$\vec \tta=(\tta_1,\dots,\tta_r)=(\vec \tta_s)_{s\in\Abw}$}, which associates $r_s=\dim W_s$ parameters $\vec \tta_s$  to the irreducible representation $\rho_s$. This defines a map
\begin{align}
\tts:\{1,\dots,r\}\longrightarrow\Abw \ , \quad l\longmapsto \tts(l) \ .
\end{align}

\subsubsection*{Fixed Points and Coloured Plane Partitions}

Since the actions of $\Ab$ and $\sT$ commute, the $\sT$-fixed points of the moduli space \smash{$\frM_{r,k}^\Ab$} are all isolated and are in one-to-one correspondence with arrays of plane partitions $\vec \pi=(\pi_1,\dots,\pi_r)$ of size $k$. Each plane partition $\pi_l$ is coloured according to the $\Ab$-colouring defined through the homomorphism $\tau_{\vec s}$ and the isomorphism \smash{$\Abw\simeq\Ab$} of finite abelian groups by
\begin{align}
\RZ_{\geq 0}^{\oplus 3}\longrightarrow \Abw \ , \quad \vec n=(n_1,n_2,n_3)\longmapsto \rho_{s_1}^{\otimes n_1}\otimes\rho_{s_2}^{\otimes n_2}\otimes\rho_{s_3}^{\otimes n_3} \ ,
\end{align}
where the box of $\pi_l$ situated at \smash{$\vec p\in\RZ_{>0}^3$} carries an irreducible representation of the orbifold group $\Ab$ given by
\begin{align}
\rho_{\,l;\,\vec p} := \rho_{\tts(l)}\otimes \rho_{s_1}^{\otimes(p_1-1)}\otimes\rho_{s_2}^{\otimes(p_2-1)}\otimes \rho_{s_3}^{\otimes(p_3-1)} \ .
\end{align}
When \smash{$\vec\pi\in\frM^\sT_{\vec r,\vec k}\,$}, the total number of boxes of colour $\rho_s$ in $\vec{\pi}$ for each $s\in\Abw$ is the fractional instanton number $|\vec{\pi}|_s=k_s$.

The instanton deformation complex for the quiver variety \smash{$\frM_{\vec r,\vec k}$} at a fixed point \smash{$\vec\pi\in\frM^\sT_{\vec r,\vec k}$} is the $\Ab$-equivariant version of the complex \eqref{eq:complex_C3} given by
\begin{align} \label{eq:complex_3d_equiv}
    \sEnd_\Ab(V_{\vec\pi})\xrightarrow{ \ \dd_1^\Ab \ }\begin{matrix}
        \sHom_\Ab(V_{\vec\pi}, V_{\vec\pi} \otimes  Q_3^{\vec s}\, ) \\[1ex] \oplus\\[1ex] \sHom_\Ab(W_{\vec\pi},V_{\vec\pi}) \\[1ex] \oplus\\[1ex] \sHom_\Ab(V_{\vec\pi}, V_{\vec\pi} \otimes\midwedge^3\,  Q_3^{\vec s}\, )
    \end{matrix}
    \xrightarrow{ \ \dd_2^\Ab \ }\begin{matrix} \sHom_\Ab(V_{ \vec\pi},V_{\vec\pi}\otimes\midwedge^2\,  Q_3^{\vec s}\,  )
    \\[1ex] \oplus\\[1ex] \sHom_\Ab(V_{\vec\pi}, W_{\vec\pi} \otimes\midwedge^3\,  Q_3^{\vec s}\, )
    \end{matrix}
    \ , 
\end{align}
where the map $\dd_1^\Ab$ is an infinitesimal \smash{$\sG_{\vec k}\,$} gauge transformation, and $\dd_2^\Ab$ is the linearization of the complex ADHM equations in \eqref{eq:complex_C3_orb}. 

The instanton partition function is obtained by computing the character of the complex \eqref{eq:complex_3d_equiv}.
Since $\tau_{\vec s\,}(\Ab)\subset\sT_{\vec\epsilon\,}$, this may be calculated by taking the $\Ab$-invariant part of the character \eqref{eq:complex_3d}, with the weight decomposition
\begin{align}
Q_3^{\vec s}=t_1^{-1}\,\rho_{s_1}+t_2^{-1}\,\rho_{s_2}+t_3^{-1}\,\rho_{s_3}
\end{align}
in the representation ring of $\sT_{\vec\epsilon}\times\Ab$. It reads as
\begin{align}\begin{split}\label{eq:complex_3d_orb}
\ch^\Ab_{\sT}\big(T_{\vec\pi}^{\rm vir} \frM_{ \vec r,\vec k}\big)
&= \Big[W_{\vec \pi}^*\otimes V_{\vec\pi} - \frac{V^*_{\vec\pi}\otimes W_{\vec\pi}}{ t_1\,t_2\,t_3} \ \rho_{s_{123}} \\
& \hspace{2cm} + V_{\vec\pi}^*\otimes V_{\vec\pi} \ \frac{(1-t_1\,\rho^*_{s_1})\,(1-t_2\,\rho^*_{s_2})\,(1-t_3\,\rho^*_{s_3})}{t_1\,t_2\,t_3} \ \rho_{s_{123}}\Big]^\Ab \ . \end{split}
\end{align}

The  decompositions of the vector spaces $V$ and $W$ at the fixed point \smash{$\vec\pi\in\frM_{\vec r,\vec k}^\sT$} are given by
\begin{align}\label{eq:decom_abelian}\begin{split}
V_{\vec \pi}=\sum_{l=1}^{r}\,e_l \ \sum_{\vec p\,\in\pi_l}\,t_1^{p_1-1}\,t_2^{p_2-1}\,t_3^{p_3-1}\otimes\rho^*_{\,l;\,\vec p} \qquad \mbox{and} \qquad  
W_{\vec \pi}=\sum_{l=1}^{r}\, e_l \otimes\rho^*_{ \tts(l)} \ ,\end{split}
\end{align}
as elements of the representation ring of the group $\sT\times\Ab$. 
The character \eqref{eq:complex_3d_orb} is evaluated by projecting onto the trivial representation $\rho_0$, leaving an element in the representation ring of $\sT$. 

\subsubsection*{Equivariant Generating Functions}

The full instanton partition function is defined as
\begin{align}
Z_{[\FC^3/\,\Ab]\times\rmB\sK^{\vec s\,}}^{\vec r}(\vec \qu\,;\vec \tta,\vec \epsilon\,)=\sum_{\vec k\in \RZ_{\geq 0}^{\#\Ab}} \, \vec\qu^{\,\vec k} \ Z_{[\FC^3/\,\Ab]\times\rmB\sK^{\vec s}\,}^{\vec r,\vec k}(\vec \tta,\vec \epsilon\,) \ ,
\end{align}
where \smash{$\vec\qu=(\qu_s)_{s\in\Abw}$} is a set of fugacities for the fractional instanton sectors \smash{$\vec k=(k_s)_{s\in\Abw}$} with
\begin{align}
\vec \qu^{\,\vec k}:=\prod_{s\in\Abw}\,\qu_{s}^{k_{s}} \ .
\end{align}

The equivariant partition function for the quiver variety $\frM_{\vec r,\vec k}$ is given by
\begin{align}\begin{split}\label{eq:Zk}
Z_{[\FC^3/\,\Ab]\times\rmB\sK^{\vec s}\,}^{\vec r,\vec k}(\vec \tta,\vec \epsilon\,):&\!=\sum_{\vec\pi\,\in\,\frM^\sT_{\vec r,\vec k}}\,\widehat{\texttt e}\big[-\ch^\Ab_{\sT}(T_{\vec\pi}^{\textrm{vir}}\frM_{\vec r,\vec k}) \big] \\[4pt]
& =  \sum_{\vec\pi\,\in\,\frM_{\vec r,\vec k}^\sT} \ \prod_{l=1}^r \  \prod_{\vec p_l\in{\pi}_{l}}^{\neq0} \, \frac{P_{r}\circ\delta^\Ab_0(-\tta_{l}-\vec p_{l}\cdot\vec\epsilon\,|\epsilon_{123}-\vec \tta)}{P_{r}\circ\delta^\Ab_0(\tta_{l}+\vec p_{l}\cdot\vec\epsilon\,|\vec \tta)}\\
& \hspace{3cm} \times \prod_{ l'=1}^{r} \ \prod_{\vec p_{l'}^{\,\prime}\in{\pi}_{l'}}^{\neq0} \,R \circ\delta^\Ab_0(\tta_l-\tta_{l'}+(\vec p_l-\vec p^{\,\prime}_{l'})\cdot\vec\epsilon\,|\vec \epsilon\,)\ ,
\end{split}
\end{align}
where the operation $\delta_0^\Ab$  acts on a combination of equivariant parameters $x$ as the identity if $x$ is associated to the trivial representation $\rho_0$ and returns $1$ otherwise; for example
\begin{align}\begin{split}
\delta^\Ab_0\big(\tta_{l}-\tta_{l'}+(\vec p_l-\vec p_{l'}^{\,\prime})\cdot\vec\epsilon\,\big)= \begin{cases}\tta_{l}-\tta_{l'}+(\vec p_l-\vec p_{l'}^{\,\prime})\cdot\vec\epsilon \quad \text{if} \ \ \rho_{\,l;\,\vec p}\otimes\rho^*_{\,l';\,\vec p^{\,\prime}}\simeq\rho_{0} \ , \\[4pt]
1 \quad \text{otherwise} \ .
\end{cases}
\end{split}
\end{align}

\begin{example}\label{ex:C3Z2Z2}
Let $\Ab=\RZ_2\times\RZ_2=\{\ident,g_1,g_2,g_3\}$ be the Klein four-group represented faithfully in $\sSO(3)\subset\sU(3)$ by the matrices
\begin{align}
    g_1={\small \begin{pmatrix}
    -1&&\\
    &-1&\\
    &&1
    \end{pmatrix}} \normalsize \qquad \mbox{and} \qquad  g_2={\small \begin{pmatrix}
    -1&&\\
    &1&\\
    &&-1
    \end{pmatrix} } \normalsize \ ,
\end{align}
together with $g_3=g_1\, g_2$. The four irreducible representations $\Abw=\{\rho_0,\rho_1,\rho_2,\rho_3\}$ have weights $s_0=(0,0,0)$, $s_1=(1,1,0)$, $s_2=(1,0,1)$ and $s_3=s_1+s_2=(0,1,1)$, respectively. The McKay quiver \smash{$\ttQ^{\RZ_2^{\times 2}}$} is displayed in Example~\ref{ex:AbelianGamma}.

Consider the $\sU(1)$ theory on $[\FC^3/\RZ_2\times\RZ_2]$ with dimension vector $\vec r=(1,0,0,0)$. Then $\tts(1)=0$ in \eqref{eq:decom_abelian}, and the equivariant instanton partition function in this case gives the generating function for rank one noncommutative Donaldson--Thomas invariants of the toric orbifold $\FC^3/\RZ_2\times\RZ_2$ with holonomy group $\sU(3)$. It evaluates to the closed formula~\cite[Proposition~5.12]{Szabo:2023ixw}
\begin{align}\label{eq:C3Z2Z2U3}
\begin{split}
Z^{\vec{r}=(1,0,0,0)}_{[\FC^3/\RZ_2\times\RZ_2]}(\vec \qu\,;\vec\epsilon\,)&=\frac{M(-\Qu)^{\frac{\epsilon_1\,\epsilon_2\,\epsilon_3-\epsilon_1^2\,\epsilon_2-\epsilon^2_1\,\epsilon_3-\epsilon^2_2\,\epsilon_3-\epsilon_1\,\epsilon_2^2-\epsilon_1\,\epsilon_3^2-\epsilon_2\,\epsilon^2_3}{\epsilon_1\,\epsilon_2\,\epsilon_3}}}{\widetilde M(\qu_1,-\Qu) \, \widetilde M(\qu_2,-\Qu) \, \widetilde M(\qu_3,-\Qu)\,\widetilde M(\qu_1\,\qu_2\,\qu_3,-\Qu)} \\[2pt]
& \hspace{2cm} \times \prod_{1\leq p<s\leq 3}\,\widetilde{M}(\qu_p\,\qu_s,-\Qu)^{\frac{\epsilon_{(ps)^-}-\epsilon_{ps}}{2\,\epsilon_{(ps)^-}}} \ ,\end{split}
\end{align}
where 
\begin{align}
\widetilde{M}(x,q)=M(x,q)\, M(x^{-1},q)
\end{align}
is the MacMahon tilde function, and we introduced the notation
\begin{align}
\Qu=\qu_0\,\qu_1\,\qu_2\, \qu_{3} \ ,
\end{align}
along with \smash{$(ps)^-=\ulthree\setminus\{p,s\}$}.
\end{example}

\begin{example}
Let $\tau_{\vec s}$ be the representation of the group $\Ab=\RZ_2\times\RZ_2\times\RZ_2$ in $\sSO(3)\subset\sU(3)$ from Example \ref{ex:AbelianGamma}, and consider the cohomological $\sU(2)$ gauge theory on $[\FC^3/\RZ_2\times\RZ_2\times\RZ_2]\times\mathrm{B}\,\RZ_2$. We focus on the contributions to the generating function from the array of $\RZ_2^{\times 3}$-coloured plane partitions
{\begin{align}
\vec\pi \ = \ ( \hspace{0.5cm}  ,  \hspace{0.5cm} )
\end{align}
\vspace{-1.95cm}
$$
\hspace{1cm}
\begin{tikzpicture}
\scalebox{0.2}{\planepartition{{1}}}
\end{tikzpicture} 
\hspace{-1.1cm}
\begin{tikzpicture}
\scalebox{0.2}{\planepartitiong{{1}}} 
\end{tikzpicture}
$$}

\vspace{-1cm}
\noindent
with $|\vec k\,|=2$ boxes and the following $\RZ_2^{\times 3}$-colourings:
\begin{align}\begin{split} \label{ex:example_Aborb}
\hspace{-0.2cm} Z^{\vec r,\vec k}_{[\FC^3/\RZ_2\times\RZ_2\times\RZ_2]\times\rmB\RZ_2}(\vec\tta,\vec \epsilon\,)=
\begin{cases} \, \displaystyle \frac{\tta ^2-\epsilon_{12}^2}{\tta ^2-\epsilon_{1}^2} \ , \quad  r_{(0,0,0)}=r_{(1,0,0)}=1 \ , \ k_{(0,0,0)}=k_{(1,0,0)}=1 \ , \\[4pt]  \, 1 \ , \quad r_{(0,0,0)}=r_{(0,0,1)}=1 \ , \ k_{(0,0,0)}=k_{(0,0,1)}=1 \ , \end{cases}
\end{split}
\end{align}
where $\tta :=\tta_1-\tta_2$.
In the first case the partition function \eqref{ex:example_Aborb} coincides with the contribution to the partition function \smash{$Z^{\vec r , \vec k}_{[\FC^3/\RZ_2\times\RZ_2]}(\vec\tta,\vec \epsilon\,)$} of Example~\ref{ex:C3Z2Z2} from the $\RZ_2\times\RZ_2$-colouring with $r_{(0,0,0)}=r_{(1,1,0)}=1$ and $k_{(0,0,0)}=k_{(1,1,0)}=1$, whereas in the second case there is no correspondence.

This example serves to demonstrate that the non-effective orbifold theory is not generally equivalent to the theory defined solely by the action of a finite subgroup of $\sU(3)$. On the other hand, there is an equivalence of partition functions 
\begin{align}
Z^{\vec r}_{[\FC^3/\RZ_2\times\RZ_2\times\RZ_2]\times\rmB\RZ_2}(\vec\qu\,;\vec \tta,\vec \epsilon\,) \, = \, Z^{\vec r}_{[\FC^3/\RZ_2\times\RZ_2]}(\vec\qu\,;\vec \tta,\vec \epsilon\,)
\end{align}
for all dimension vectors $\vec r$ whose non-zero entries are coloured by $\Ab$ in the same connected component of the McKay quiver. 
\end{example}

\section{Tetrahedron Instantons in Cohomological Field Theory}\label{sec:tetra_insta}

Tetrahedron instantons were introduced in~\cite{Pomoni:2021hkn} as a generalization of noncommutative instantons on $M_3=\FC^3$. Roughly speaking, they correspond to configurations of instantons on the codimension one coordinate hyperplanes $\mathbbm{C}^3$ inside the affine Calabi--Yau four-fold $M_4=\mathbbm{C}^4$. In this section we elaborate on the analysis of tetrahedron instantons and
their generalized ADHM parametrization.

\begin{notation}
The set of coordinate labels of the four complex lines $\FC_a\subset \FC^4$ is denoted by
\begin{align}
    a \ \in \ \underline{4}:=\{1,2,3,4\} \ .
\end{align}
There are four complex codimension one hyperplanes
\begin{align}
    \FC^3_A&=\Timesbig_{a\in A}\,\mathbbm{C}_a \ \subset \ \mathbbm{C}^4 \qquad \text{with} \quad A \ \in \ \ulfour^\perp :=\big\{(123)\,,\,(124)\,,\,(134)\,,\,(234)\big\} \ ,
\end{align}
and for any $A\in\ulfour^\perp$ we define $\bar A\in\ulfour$ to be its complement
\begin{align}
   \bar A = \ulfour \setminus A \ .
\end{align}
The lexicographically ordered sets $\ulfour$ and $\ulfour^\perp$ respectively label the vertices and faces of a tetrahedron. We will denote by $A_1\cap A_2=(a_1\,a_2)$ the unique pair of vertices $a_1,a_2\in\ulfour$ joined by the common edge of two distinct faces~$A_1,A_2\in\ulfour^\perp$; note that $\bar A_1\in A_2$ and $\bar A_2\in A_1$. 

We introduce the following vector spaces:
\begin{myitemize} 
\item $Q_4\,$: \  a four-dimensional Hermitian vector space which forms the fundamental representation of the group $\sSU(4)$. 
\item $Q_A\,$: \ a three-dimensional Hermitian vector space which forms the fundamental representation of the subgroup $\sU(3)_A\subset\sSU(4)$ acting on the hyperplane $\FC^3_A\subset\FC^4$ for $A\in\ulfour^\perp$. 
\item $Q_{A_1,A_2}\,$: \ a two-dimensional Hermitian vector space which forms the fundamental representation of the subgroup $\sU(2)_{A_1,A_2}\subset\sSU(4)$ acting on the intersections \smash{$\FC^2_{A_1,A_2}:=\FC^3_{A_1}\cap\FC^3_{A_2}\subset\FC^4$} for distinct \smash{$A_1,A_2\in\ulfour^\perp$}.
\item $Q_{(a)}\,$: \ the one-dimensional representation of $\sU(1)_{(a)}\subset\sSU(4)$ acting on the line $\FC_a\subset\FC^4$ for $a\in\ulfour\,$.
\end{myitemize}
\end{notation}

\subsection{$\sSU(4)$-Instanton Equations}\label{sec:gaugetheory_tetra}

Tetrahedron instantons are solutions of BRST fixed point equations for a cohomological gauge theory with two supercharges on a K\"ahler four-fold $(M_4,\omega)$ with $\sSU(4)$ holonomy, which can also be obtained through dimensional reduction of $\CN=1$ supersymmetric Yang--Mills theory with gauge group $\sU(r)$ in ten dimensions~\cite{Szabo:2023ixw}. 

Let $\varOmega$ be the nowhere-vanishing holomorphic four-form associated to
the $\sSU(4)$-structure. Define the involution 
\begin{align} \label{eq:starOmega}
\star_\varOmega:\midwedge^{0,2}\,\FC^4\longrightarrow\midwedge^{0,2}\,\FC^4 \ , \qquad \alpha\longmapsto\star_\varOmega\,\alpha:=\overline{\ast\,(\alpha\wedge\varOmega)} 
\end{align}
for \smash{$\alpha\in\midwedge^{0,2}\,\FC^4$}, where $\ast$ is the Hodge duality operator compatible with the K\"ahler form $\omega$. It gives an orthogonal decomposition of the space of $(0,2)$-forms 
\begin{align}
\midwedge^{0,2}\,\FC^4 = \midwedge^{0,2}_+\,\FC^4\, \oplus\, \midwedge^{0,2}_-\,\FC^4
\end{align}
into real $\pm1$-eigenspaces $\midwedge^{0,2}_\pm\,\FC^4$ of $\star_\varOmega$. This induces a decomposition of the $(0,2)$-form part of the field strength tensor
\smash{$\CF^{0,2} = \CF_+^{0,2}+\CF_-^{0,2}$} into eigencurvatures as
\begin{align}
\CF_\pm^{0,2} = \tfrac12\,\big(\CF^{0,2}\pm\star_\varOmega\, \CF^{0,2}\big) \qquad \mbox{with} \quad \star_\varOmega \CF_\pm^{0,2} = \pm\,\CF_\pm^{0,2} \ .
\end{align}

The BRST symmetry localizes the path integral of the gauge theory onto the space of solutions of the generalized instanton equations
\begin{align}\label{eq:gauge_tetra}
\begin{split}
\CF^{0,2}_-&=0 \ , \\[4pt]
\omega\wedge\omega\wedge\omega\wedge \CF^{1,1} &= 0 \ , \\[4pt] 
\nabla_{\!\CA}\,\Phi&=0 \ ,
\end{split}
\end{align}
where again we assume the first Chern class vanishes.

Tetrahedron instantons correspond to particular solutions of \eqref{eq:gauge_tetra} 
on a singular three-fold $M_\triangle$ embedded in a local Calabi--Yau four-fold $M_4$ as a stratification
\begin{align}
M_\triangle = \bigcup_{A\,\in\,\ulfour^\perp} \, M_A \ \subset \ M_4 \ .
\end{align}
For instance, $M_\triangle = \varpi^{-1}(0)$ may arise as the central fibre of a toric degeneration  $\varpi:M_4\longrightarrow\FC$ with gluing data along intersections of strata $M_A$ prescribed by a polyhedral complex which forms a tetrahedron~\cite{Fasola:2023ypx}. 
The $\sU(r)$ gauge connection $\CA$ and complex Higgs field $\Phi$ are constrained to assume the forms 
\begin{align} \label{eq:tetrasols}
\CA= \bigoplus_{A\,\in\,\ulfour^\perp}\,\CA_A \qquad \mbox{and} \qquad \Phi=\bigoplus_{A\,\in\,\ulfour^\perp}\,\Phi_A \ ,
\end{align}
where $\CA_A$ and $\Phi_A$ are supported on the smooth stratum $M_A\subset M_\triangle$ with values in the adjoint representation of $\sU(r_A)$ for $A\in\ulfour^\perp$. The restrictions of the  fields to the codimension two transverse intersections $M_{A_1,A_2}:=M_{A_1}\cap M_{A_2}\subset M_\triangle$ yield bifundamental multiplets of the product group $\sU(r_{A_1})\times\sU(r_{A_2})$.

The solution \eqref{eq:tetrasols} is labelled by an array of ranks
\begin{align}
\mbf{r}:=(r_A)_{A\,\in\,\ulfour^\perp}=(r_{123},r_{124},r_{134}, r_{234}) \ ,
\end{align}
which partitions the rank $r$ of the cohomological gauge theory:
\begin{align}
r = |\mbf r| := \sum_{A\,\in\,\ulfour^\perp} \, r_A = r_{123}+r_{124}+r_{134}+r_{234} \ .
\end{align}
It breaks to the $\sU(r)$ gauge symmetry to the subgroup
\begin{align}
\sU(\mbf{r}) := \Timesbig_{A\,\in\,\ulfour^\perp} \, \sU(r_A) \ .
\end{align}
The non-zero entries of the dimension vector $\mbf r$ also determine the unbroken holonomy group of the solution \eqref{eq:tetrasols} preserving the codimension one defects as
\begin{align} \label{eq:sHmbfr}
\sH_{\mbf r} = \bigcap_{r_A\neq 0} \, \sU(3)_A \ \subset \ \sSU(4) \ .
\end{align}

We can further group the solutions according to their instanton number (fourth Chern class)
\begin{align}
k = \frac1{384\pi^4} \, \int_{M_4} \, \Tr_{\mathfrak{u}(r)}\,\CF\wedge \CF\wedge \CF\wedge \CF \ ,
\end{align}
which is again a topological charge of the theory.
The moduli space of solutions of \eqref{eq:gauge_tetra} with charge $k\in\RZ_{\geq0}$ is called the \textit{moduli space of tetrahedron $k$-instantons of type $\mbf r$}, denoted $\frM_{\mbf{r}, k}$. The group ${\sf P}\sU(\mbf r)=\sU(\mbf r)/\sU(1)$ remains a global symmetry of the moduli space, where $\sU(1)$ is the diagonal subgroup of \smash{$\bigtimes_{r_A\neq0}\,\sU(1)$} corresponding to the common centre of the groups $\sU( r_A)$ for $A\in\ulfour^\perp$.

\subsection{ADHM Data}
\label{subsec:ADHMtetra}

Similarly to noncommutative instantons on $\FC^3$, tetrahedron instantons appear in the context of type IIB string theory~\cite{Pomoni:2021hkn}, where the ten-dimensional spacetime $\FR^{1,9}$ is identified with $\FR^{1,1}\times \FC^4$ through a choice of complex structure. When $M_4=\FC^4$, the singular divisor 
\begin{align}\label{eq:C4strat}
\FC_\triangle^3=\bigcup_{A\,\in\,\ulfour^\perp}\,\FC^3_A \ \subset \ \FC^4
\end{align} 
has strata corresponding to the codimension one coordinate hyperplanes $\FC_A^3\subset\FC^4$.

In this case tetrahedron instantons describe bound states of $k$ D1-branes along $\FR^{1,1}$ which probe intersecting stacks of $r$ D7-branes located in the four different spatial orientations labelled by $A\in\ulfour^\perp$, with $r_A$ D7$_A$-branes wrapping the  stratum $\FC_A^3\subset\FC_\triangle^3$. 
The worldvolume $\FR^{1,1}\times\FC_A^3$ of the $r_A$ D7$_A$-branes for fixed \smash{$A\in\ulfour^\perp$} supports $k$ D1-branes. It is intersected by $r_{A^\circ}$ D7$_{A^\circ}$-branes, labelled by $A^\circ\in\ulfour^\perp\setminus A$, in hyperplanes which produce defects in $\FC_A^3$ of codimension one or two. 

In the low-energy two-dimensional $\CN = (0,2)$ field theory on the D1-branes, the Higgs branch is described by generalized ADHM equations, deformed again by a Fayet--Iliopoulos parameter $\zeta\in\FR_{>0}$ related to an appropriate large constant Neveu--Schwarz $B$-field. This breaks supersymmetry, while the D1--D7-brane states decay to a supersymmetric string theory vacuum via tachyon condensation~\cite{Boundstates}.

\subsubsection*{Generalized ADHM Equations}

Let $V$ and $W_A$ be Hermitian vector spaces of dimensions $k$ and $r_A$, respectively; from the perspective of the D1-branes, $V$ is the Chan--Paton space and $W_A$ are flavour representations. Then the ADHM equations are~\cite{Pomoni:2021hkn}
\begin{align}\begin{split}\label{eq:ADHM_tetra}
    \mu_{ab}^\mathbbm{C}&:=[B_a,B_b]-\tfrac{1}{2}\,\varepsilon_{abcd}\,[B_c^\dagger,B_d^\dagger]= 0 \ , \\[4pt]
     \mu^\mathbbm{R} &:=\sum_{a\,\in\,\underline{4}}\,[B_a,B_a^\dagger]+\sum_{A\,\in\, \ulfour^\perp}\, I_A\, I_A^\dagger =\zeta\,\ident_V \ ,  \\[4pt]
   \sigma_A&:=\big(B_{\bar A}\,I_A\big)^\dag=0 \ , 
    \end{split}
\end{align}
where $a,b\in\ulfour$ and $A\in\ulfour^\perp$, while $B_a\in \sEnd_{\FC}(V)$ and $I_A\in \sHom_{\FC}(W_A,V)$. Here $\varepsilon_{abcd}$ is the Levi--Civita symbol in four dimensions with $\varepsilon_{1234}=+1$. Note that only three of the first set of equations are independent and we may restrict them to $(a,b)\in\ulthree^\perp$, or equivalently any other triple of distinct pairs of vertices from $\ulfour\,$.

If there is only a single non-zero rank $r_A=r$, for some \smash{$A\in\ulfour^\perp$}, then upon setting \emph{$Y:=B_{\bar A}^\dag$} and $I=I_A$ the equations \eqref{eq:ADHM_tetra} reduce to the ADHM equations \eqref{eq:ADHM3d} for instantons on $\FC_A^3$ with holonomy $\sU(3)_A$. On the other hand, by neglecting the last equation in \eqref{eq:ADHM_tetra} and combining the linear maps $I_A$ into a single map \smash{$I=\bigoplus_{A\,\in\,\ulfour^\perp}\,I_A\in\sHom_\FC(W,V)$} with \smash{$W:=\bigoplus_{A\,\in\,\ulfour^\perp}\,W_A$}, we recover the ADHM equations for the magnificent four model, i.e. the Donaldson--Thomas theory of the affine Calabi--Yau four-fold $\FC^4$ with $\sSU(4)$-holonomy~\cite{m4,m4c,Bonelli:2020gku,Szabo:2023ixw}. 

The ADHM data provide a framed linear representation of the four-loop quiver ${\sf L}_4$ with one vertex and four edge loops:
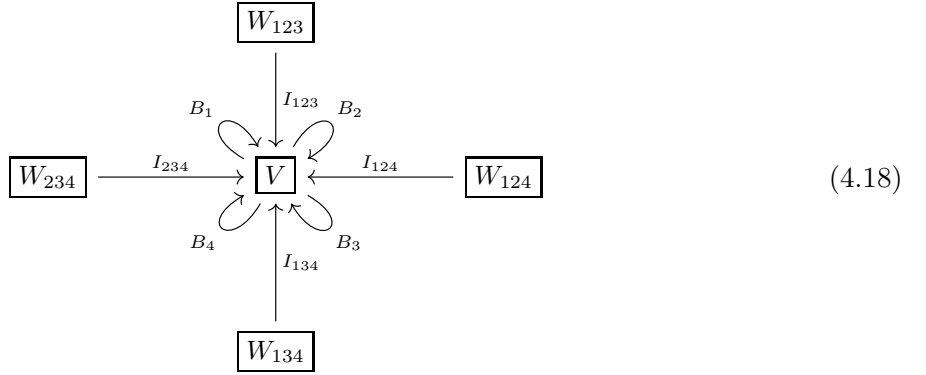
\begin{equation} \label{eq:tetraquiver}
{\small
\begin{tikzcd}
	&& {\boxed{W_{123}}} \\
	\\
	{\boxed{W_{234}}} && {\boxed{V}} \arrow["B_1"',out=150,in=120,loop,swap,]\arrow["B_2"',out=60,in=30,loop,swap,] \arrow["B_3"',out=330,in=300,loop,swap,]\arrow["B_4"',out=240,in=210,loop,swap,] && {\boxed{W_{124}}} \\
	&& {} \\
	&& {\boxed{W_{134}}}
	\arrow["{I_{234}}", from=3-1, to=3-3]
	\arrow["{I_{124}}"', from=3-5, to=3-3]
	\arrow["{I_{123}}"', from=1-3, to=3-3,swap]
	\arrow["{I_{134}}", from=5-3, to=3-3,swap]
\end{tikzcd} } \normalsize
\end{equation}
This generalizes the enhanced framed representation of ${\sf L}_3$ from \eqref{eq:ADHMquiverC3} as well as the framed representation of ${\sf L}_4$ from \cite[eq.~(2.45)]{Szabo:2023ixw}.

\subsubsection*{Stability and Quot Schemes}

As pointed out by \cite{m4}, the complex ADHM equations $\mu_{ab}^\FC=0$ are equivalent to the EJ-term equations 
\begin{align} \label{eq:Bcomm4}
[B_a,B_b]=0
\end{align}
for all $a,b\in\ulfour\,$,
 through the identity
\begin{align}\label{eq:muabCholomorphic}
    \sum_{1\leq a< b\leq 4}\, \big\| \mu^\FC_{ab}\big\|_{\textrm{\tiny F}}^2&=\sum_{1\leq a< b\leq 4}\, \big\| [B_a,B_b] \big\|_{\textrm{\tiny F}}^2 \ ,
\end{align}
where $\|\ccdot\|_{\textrm{\tiny F}}$ is the Frobenius norm on $\sEnd_\FC(V)$.

Instead, the D-term equation $\mu^\FR=\zeta\,\ident_V$ of \eqref{eq:ADHM_tetra} is equivalent to a stability condition similar to the stability condition of Section~\ref{subsec:ADHMC3}: there is no proper subspace $S\subset V$ such that
\begin{align}
B_a(S)\subset S \qquad \mbox{and} \qquad I_A(W_A)\subset S 
\end{align}
for all $a\in\ulfour$ and $A\in\ulfour^\perp$.

It follows that the moduli space of tetrahedron instantons $\frM_{\mbf{r}, k}$ is given by the noncommutative Quot scheme
\begin{align}\label{eq:moduli_space_tetra}
\frM_{\mbf{r}, k} \ \simeq \ \mbf{\mu}^{\FC-1}(0)^{\rm stable} \, \big/ \, \sGL(V) \ ,
\end{align}
where \smash{$\mbf \mu^\FC :=(\mu_{ab}^\FC,\sigma_A)_{(a,b)\,\in\,\ulthree^\perp\,,\, A\,\in\,\ulfour^\perp}$}, the superscript ${}^{\rm stable}$ indicates the stable solutions of \eqref{eq:Bcomm4} and the third equation in \eqref{eq:ADHM_tetra}, while $g\in\sGL(V)\simeq\sGL(k,\FC)$ acts on the ADHM data as
\begin{align}
g\cdot(B_a,I_A)_{a\,\in\,\ulfour\,,\,A\,\in\,\ulfour^\perp} = (g\,B_a\,g^{-1},g\,I_A)_{a\,\in\,\ulfour\,,\,A\,\in\,\ulfour^\perp} \ .
\end{align}

\begin{remark} \label{rem:tetrastability}
The stability condition is equivalent to the statement
\begin{align}
V=\sum_{A\,\in\,\ulfour^\perp}\, V_A:=\sum_{A=(a\,b\,c)\,\in\,\ulfour^\perp}\, \FC[B_a,B_b,B_c]\,I_A(W_A) \ .
\end{align}
The vector space $V_A$ is the smallest subspace of $V$ containing ${\rm im}(I_A)$ which is invariant under the actions of $B_a$, $B_b$ and $B_c$;
its complex dimension $k_A:=\dim V_A$ is the instanton number on the stratum $\FC_A^3\subset\FC_\triangle^3$. 
The equation $\sigma_A=0$ in \eqref{eq:ADHM_tetra} together with \eqref{eq:Bcomm4} then imply 
\begin{align}\label{eq:BAVA=0}
B_{\bar A}(V_A) = 0 \ ,
\end{align}
for all $A\in\ulfour^\perp$. 
\end{remark}

Let \smash{$\iota_A:\, \FC^3_A\lhook\joinrel\longrightarrow \FC_\triangle^3$} be the inclusion of the irreducible codimension one strata in the singular three-fold \eqref{eq:C4strat} for each $A\in\ulfour^\perp$. Let $\CE_{\mbf{r}}$ be the torsion sheaf on $\FC_\triangle^3$ defined by
\begin{align}
\CE_{\mbf{r}}=\bigoplus_{A\,\in\,\ulfour^\perp}\,\iota_{A*}\,\CO^{\oplus\, r_A}_{\FC^3_A} \ .
\end{align} 
As shown by \cite{Fasola:2023ypx}, the description \eqref{eq:moduli_space_tetra} implies that the moduli space of tetrahedron instantons $\frM_{\mbf{r},k}$ is isomorphic to the Quot scheme $\Quot_{\mbf r}^k (\FC_\triangle^3)$ of zero-dimensional quotients of $\CE_{\mbf{r}}$ with length~$k$:
\begin{align}\label{eq:quot_tetra}
\frM_{\mbf{r}, k} \, \simeq \, \Quot_{\mbf r}^k (\FC_\triangle^3) \ .
\end{align}
There are natural closed immersions among Quot schemes
\begin{align}
\frM_{\mbf{r}, k} \lhook\joinrel\longrightarrow \Quot_r^k(\FC_\triangle^3)\lhook\joinrel\longrightarrow \Quot_r^k(\FC^4) \ ,
\end{align}
where $r=|\mbf r|$ and $\Quot_r^k(\FC^4)$ is isomorphic to the moduli space of $\sU(r)$  instantons on $\FC^4$ with charge~$k$~\cite{Cazzaniga:2020xru,Szabo:2023ixw}.

\begin{remark}[{\bf Instantons on $\boldsymbol{\FC_A^3}$}]
In the sector where only $r_A=r$ is non-zero, for some $A\in\ulfour^\perp$, the Quot scheme \eqref{eq:quot_tetra} coincides with the Quot scheme
\begin{align}
\Quot_{\mbf r}^k (\FC_\triangle^3) \, \simeq \, \Quot_r^k(\FC_{A}^3) \ ,
\end{align}
which is isomorphic to the moduli space $\frM_{r,k}$ of $\sU(r)$ instantons on $\FC_{A}^3$ with charge $k$ (cf.~\eqref{eq:quot3d}).
\end{remark}

\subsection{Tangent-Obstruction Theory}\label{sec:integration}

The local geometry of the tetrahedron instanton moduli space $\frM_{\mbf{r},k}$ is captured by the instanton deformation complex~\cite{Baulieu:1997jx}
\begin{align}\label{eq:defcomplex_tetra}
\midwedge^0\,T^*M_4\otimes\frg \xrightarrow{ \ \bar\partial_\CA \ } \midwedge^{0,1}\,T^*M_4\otimes\frg \xrightarrow{ \ \bar\partial_\CA^- \ } \midwedge_-^{0,2}\,T^*M_4\otimes\frg \ .
\end{align}
The first differential is a linearized complex gauge transformation, while the second differential is the linearization of the first equation in \eqref{eq:gauge_tetra}, where $\bar\partial_\CA^-:=P_\varOmega^-\circ\bar\partial_\CA$ and
\begin{align}
P_\varOmega^-=\tfrac12\,(\ident-\star_\varOmega)\,:\,\midwedge^{0,2}\,\FC^4\longrightarrow \midwedge_-^{0,2}\,\FC^4 
\end{align}
is the projection onto the real $-1$-eigenspace of the involution $\star_\varOmega$. 
 
The degree one cohomology $\ker(\bar\partial_\CA^-)/\mathrm{im}(\bar\partial_\CA)$ of the cochain complex \eqref{eq:defcomplex_tetra} describes the complex tangent bundle $T\frM_{\mbf{r},k}\longrightarrow\frM_{\mbf{r},k}$ over a fixed holomorphic connection $\CA$ of the form \eqref{eq:tetrasols}. The second cohomology \smash{$\mathrm{coker}(\bar\partial_\CA^-)$} defines the real self-dual obstruction bundle \smash{$\Ob^-_{\mbf{r},k}\longrightarrow\frM_{\mbf{r},k}$}, which is orientable. We assume the degree zero cohomology vanishes, i.e. $\ker(\bar\partial_\CA) = 0$, so that there are no infinitesimal automorphisms.

The \emph{virtual tangent bundle} $T^{\rm vir}\frM_{\mbf{r},k}$ is the two-term elliptic complex 
\begin{align} \label{eq:tetraTvir}
T^{\rm vir}\frM_{\mbf{r},k} := \big[T\frM_{\mbf{r},k} \longrightarrow \Ob^-_{\mbf{r},k}\big] \ .
\end{align}
When $M_4=\FC^4$, it is easy to show using the ADHM parametrization that the complex virtual dimension of the moduli space $\frM_{\mbf{r},k}$ vanishes:
\begin{align}
\mathrm{vdim}\,\frM_{\mbf{r},  k} = \bigg(4\,k^2+\sum_{A\,\in\,\ulfour^\perp}\,r_A\, k\bigg)-\bigg(3\,k^2 + \sum_{A\,\in\,\ulfour^\perp}\,r_A\, k\bigg)-k^2=0 \ .
\end{align} 

The definition of the Euler class of $T^{\rm vir}\frM_{\mbf{r},k}$ is a bit more involved now because the self-dual obstruction bundle \smash{$\Ob_{\mbf{r},k}^-$} is a real vector bundle. As explained in~\cite{Szabo:2023ixw}, we identify its Euler class through a square root of the Euler class of the complexification $\Ob_{\mbf{r},k}:=\Ob_{\mbf{r},k}^-\otimes_\FR \FC$:
\begin{align}\label{eq:Euler_square}
e(\Ob^-_{\mbf{r},k}) = \sqrt{e}(\Ob_{\mbf{r},k}) \ .
\end{align}
The square root is defined up to a sign determined by a choice of orientation  of \smash{$\Ob_{\mbf{r},k}^-$}; the virtual fundamental class $[\frM_{\mbf{r},k}]^{\rm vir}$ also depends on this choice, though it is customary not to indicate the dependence explicitly. In the cohomological gauge theory, the sign choice corresponds to a choice of lexicographic ordering of the real antighosts in the path integral measure. 
We now define the Euler class
\begin{align}
 \sqrt{e} (T^{\rm vir}\frM_{ \mbf{r}, k}) := \frac{e\big(T\frM_{\mbf{r},k}\big)}{\sqrt{e}\big(\Ob_{\mbf{r},k}\big)} \ .
\end{align}

The tangent-obstruction theory can also be described in terms of ADHM data. For this, 
we introduce vector  bundles 
\begin{align}
\CCV=\mbf{\mu}^{\FC-1}(0)^{\rm stable}\,\times_{\sGL(V)} \, V \qquad \mbox{and} \qquad \CCW_A=\frM_{\mbf{r},k}\,\times\,W_A 
\end{align}
on the moduli space $\frM_{\mbf{r},k}$, whose fibres over a gauge orbit $[\CA]$ are the complex vector spaces $V$ and $W_A$, for $A\in\ulfour^\perp$,  which are part of the generalized ADHM parametrization of $\frM_{\mbf{r}, k}$ discussed in Section \ref{subsec:ADHMtetra}. 

Then there is a three-term cochain complex of vector bundles over $\frM_{\mbf{r},k}$ given by
\begin{align}\label{complex_bundle}
    \sEnd(\CCV)\xrightarrow{ \ \dd_1 \ }\begin{matrix}
        \sHom(\CCV, \CCV \otimes  Q_4 ) \\[1ex] \oplus\\[1ex] \bigoplus\limits_{A\,\in\,\ulfour^\perp}\,\sHom(\CCW_A,\CCV) 
    \end{matrix}
    \xrightarrow{ \ \dd_2 \ }\begin{matrix} \sHom(\CCV,\CCV\otimes\midwedge_-^2\,  Q_4  )
    \\[1ex] \oplus\\[1ex] \bigoplus\limits_{A\,\in\,\ulfour^\perp}\, \sHom(\CCV, \CCW_A \otimes\midwedge^3\,  Q_A )
    \end{matrix}  	\ ,    
\end{align}
where the vector bundle homomorphisms $\dd_1$ and $\dd_2$ act fibrewise over a point $\big[(B_a,I_A)_{a\,\in\,\ulfour\,,\,A\,\in\,\ulfour^\perp}\big]$ of $\frM_{\mbf{r},k}$  as
\begin{align}
\begin{split}
\dd_1(\phi)&=
\big( [B_a,\phi]\,,\,-\phi\, I_A\big)_{a\,\in\,\ulfour\,,\,A\,\in\,\ulfour^\perp} \ ,\\[4pt]
\dd_2(b_a,i_A)_{a\,\in\,\ulfour\,,\,A\,\in\,\ulfour^\perp}&= \big([b_a,B_b]+[B_a,b_b]\,,\,I_A^\dagger\,b^\dagger_{\bar A} +i_A^\dagger\, B_{\bar A}^\dagger \big)_{(a,b)\,\in\,\ulthree^\perp\,,\,A\,\in\,\ulfour^\perp} \ .
\end{split}
\end{align}
Again the stability condition implies $\ker(\dd_1)=0$.

\subsection{Equivariant Generating Functions}

Since the virtual dimension is zero, the tetrahedron instanton partition function is given by
\begin{align}\label{eq:generic_pf_2}
Z^{\mbf{r},k}_{\FC^4}= \int_{[\frM_{\mbf{r},k}]^{\rm vir}} \, 1 \ ,
\end{align}
which again we interpret as the $\sT$-equivariant volume of the moduli space $\frM_{\mbf r,k}$ with respect to the action of some torus group $\sT$ on $\frM_{\mbf{r},k}$, i.e. as the pushforward of $1$ to a point in $\sT$-equivariant cohomology. The $\sT$-action on the moduli space induces $\sT$-equivariant structures on the vector bundles $\CCV$ and $\CCW_A$ for $A\in\ulfour^\perp$. Let  $\frM_{\mbf{r},k}^\sT$ be the subscheme of $\sT$-fixed points of $\frM_{\mbf{r},k}$. As a set, it stratifies into $\sT$-invariant connected components as 
\begin{align}\label{eq:complex_normal}
\frM_{\mbf{r}, k}^\sT =\bigsqcup_{\ttF\,\in\,\pi_0(\frM_{ \mbf{r}, k}^\sT)} \, \frM_\ttF \ .
\end{align}

In Appendix~\ref{app:compact} we prove that the moduli schemes $\frM_\ttF$ for \smash{$\ttF\in\pi_0\big(\frM_{\mbf{r},k}^\sT\big)$} are compact in the complex analytic topology induced by the Frobenius norm on the affine space of ADHM data \smash{$(B_a,I_A)_{a\,\in\,\ulfour\,,\,A\,\in\,\ulfour^\perp}$}, for the torus actions considered in this paper. Since the only fixed point of our toric actions is $0\in\FC^4$, we believe that this is also true in the Zariski topology, i.e. that $\frM_\ttF$ are {proper}. This enables application of the virtual localization formula from~\cite{Graber,Oh:2020rnj} to evaluate the partition function \eqref{eq:generic_pf_2} as follows.

The \emph{$\sT$-fixed} part of the pullback of the two-term complex \eqref{eq:tetraTvir} over  $\frM_\ttF$ is the virtual tangent bundle \smash{$T^{\rm vir}\frM_\ttF = T^{\rm vir}\frM_{\mbf{r},k}\big|_{\frM_\ttF}^{\rm fix}$}, which defines a virtual fundamental class $[ \frM_\ttF]^{\rm vir}$. The equivariant virtual normal bundle to $\frM_\ttF\subset\frM_{\mbf{r},k}$ is the \textit{$\sT$-moving} part
\begin{align}
\CCN_{\frM_\ttF}^{\rm vir}=T^{\rm vir}\frM_{\mbf{r},k}\big|_{\frM_\ttF}^{\rm mov} = \big[T\frM_{\mbf{r},k} \longrightarrow \Ob^-_{\mbf{r},k}\big] \big|_{\frM_\ttF}^{\rm mov}
\end{align}
of the pullback of the virtual tangent bundle $T^{\rm vir}\frM_{\mbf{r},k}$ over $\frM_\ttF$.
The partition function \eqref{eq:generic_pf_2} is then evaluated as a sum over the $\sT$-invariant connected components $\frM_\ttF$ of $\frM_{ \mbf{r}, k}$ given by~\cite{Graber,Oh:2020rnj}
\begin{align}\label{eq:localization}
Z_{\FC^4}^{\mbf{r}, k} =   \sum_{\ttF\,\in\,\pi_0(\frM_{\mbf{r}, k}^{\sT})} \ \int_{[\frM_\ttF]^{\rm vir}} \, \frac{1}{ \sqrt{e_{\sT}}\big(\CCN^{\rm vir}_{\frM_\ttF}\big)} \ .
\end{align}

The equivariant square root Euler class \smash{$\sqrt{e_{\sT}}\big(\CCN^{\rm vir}_{\frM_{\ttF}}\big)$} is defined up to a sign that depends explicitly on the orientation of the pullback of the obstruction bundle $\Ob^-_{\mbf{r},k}$ to the connected component $\frM_\ttF$. It can be obtained from the square root of the equivariant Chern character of the virtual tangent bundle $T^{\rm vir}\frM_{\mbf{r},k}$, which is computed from the index of the complex of vector bundles \eqref{complex_bundle}, regarded as an element in the $\sT$-equivariant K-theory of the moduli space $\frM_{\mbf{r},k}$. This index bundle is given by
\begin{align}\label{eq:index}\begin{split}
\sInd_{\mbf r,k}^-&=-\CCV^*\otimes\CCV\otimes\big(\FC - Q_4 + \midwedge_-^2\,Q_4\big)  +\sum_{A\,\in\,\ulfour^\perp}\, \CCW_A^*\otimes\CCV - \CCV^*\otimes\CCW_A\otimes\midwedge ^3\,Q_A  \ ,  \end{split}
\end{align}
where here $\FC$ denotes the trivial $\sT$-representation.

The $\sT$-equivariant Euler class is now extracted along the lines discussed in~\cite{Nekrasov:2016ydq}. We expand the index bundle \eqref{eq:index} as
\begin{align}
\sInd^-_{\mbf r,k}=\sum_{i_+\in \ttI_+} \, \CCE_{i_+}-\sum_{i_-\in \ttI_-} \, \CCE_{i_-} \ ,
\end{align}
where $\CCE_{i_\pm}$ are $\sT$-equivariant vector bundles on $\frM_{\mbf{r},k}$ labelled by two sets of indices $\ttI_\pm$. After a gauge transformation, the character of the pullback of \eqref{eq:index} to $\frM_\ttF$ can then be expressed in terms of ordinary Chern characters as
\begin{align} \label{eq:sqrtchTgen}
\sqrt{\ch_\sT}\big(T^{\rm vir}\frM_{\mbf r,k}\big|_{\frM_\ttF}\big) := \ch_\sT\big(\sInd^-_{\mbf r,k}\big|_{\frM_\ttF}\big)=\sum_{i_+\in \ttI_+}  \, \e^{\,\textsf w^\ttF_{i_+}}\,\ch\big(\CCE_{i_+}^\ttF\big)-\sum_{i_-\in \ttI_-} \, \e^{\,\textsf w^\ttF_{i_-}}\, \ch\big(\CCE_{i_-}^\ttF\big) \ ,
\end{align}
where $\textsf w^\ttF_{i_\pm}$ are the corresponding $\sT$-weights of \smash{$\CCE^\ttF_{i_\pm}:=\CCE_{i_\pm}\big|_{\frM_\ttF}$}. 

Since the virtual normal bundle $\CCN_{\frM_\ttF}^{\textrm{vir}}$ involves only non-zero $\sT$-weights, it follows by computing the $\sT$-equivariant top Chern class from \eqref{eq:sqrtchTgen} that its $\sT$-equivariant square root Euler class reads
\begin{align}\label{eq:euler_calss}
\sqrt{e_\sT}\big(\CCN^{\rm vir}_{\frM_\ttF}\big) = \prod_{\substack{i_\pm\in \ttI_\pm \\ \textsf w^\ttF_{i_\pm}\neq0}} \ \frac{c\big(\CCE^\ttF_{i_+};\textsf w^\ttF_{i_+}\big)}{c\big(\CCE^\ttF_{i_-};\textsf w^\ttF_{i_-}\big)} \ ,
\end{align}
where
\begin{align}
c(\CCE;\textsf w)=\sum_{j=0}^{\textrm{rk}(\CCE)}\,{\textsf w}^j \, c_{\textrm{rk}( \CCE)-j}(\CCE )
\end{align}  
is the usual Chern polynomial of the vector bundle $\CCE$, with $c_j(\CCE)$ its $j$-th Chern class. 

The character \eqref{eq:sqrtchTgen} is a square root of the equivariant Chern character $ \ch_{\sT}\big(T^{\rm vir}\frM_{\mbf{r},k}\big)$ of the virtual tangent bundle
\begin{align}\begin{split} 
T^{\rm vir}\frM_{\mbf{r},k}&=\sInd_{\mbf r,k}^- + \sInd_{\mbf r,k}^+ = \sInd_{\mbf r,k}^- + (\sInd_{\mbf r,k}^-)^*  \\[4pt]
&= -\CCV^*\otimes \CCV \otimes\Big(\FC-\sum_{a\,\in\,\ulfour} \, (-1)^{a-1} \, \midwedge^{a-1}\,Q_4\Big) \\
& \qquad \, + \sum_{A\,\in\,\ulfour^\perp} \, \CCW_A^*\otimes\CCV\otimes\big(\FC-Q_{(\bar A\,)}\big) + \CCV^*\otimes\CCW_A\otimes\big(\FC-\midwedge^3\,Q_A\big) \ ,
\end{split}
\end{align}
regarded as an element of the $\sT$-equivariant K-theory of $\frM_{\mbf r,k}$, where we used triviality of the determinant representation $\midwedge^4\,Q_4\simeq\FC$ to identify $\midwedge^2_-Q_4^*\simeq\midwedge_+^2Q_4$ and $\midwedge^3\,Q_A^*\simeq Q_{(\bar A\,)}$. Every sign choice for the square root  is equivalent to a choice of local orientation on each $\sT$-invariant connected component $\frM_\ttF$ of the instanton moduli space, which produces a  sign factor $(-1)^{\texttt{O}_\ttF}$. 

Finally, the partition function assumes the form
\begin{align}\label{eq:localization_2}
Z_{\FC^4}^{\mbf{r}, k} =   \sum_{\ttF\,\in\,\pi_0(\frM_{\mbf{r}, k}^{\sT})} \, (-1)^{\texttt{O}_\ttF} \  \int_{[\frM_{ \ttF}]^{\rm vir}}  \  \prod_{\substack{i_\pm\in \ttI_\pm \\ \textsf w^\ttF_{i_\pm}\neq0}} \ \frac{c\big(\CCE^\ttF_{i_-};\textsf w^\ttF_{i_-}\big)}{c\big(\CCE^\ttF_{i_+};\textsf w^\ttF_{i_+}\big)}  \ .
\end{align}
The full instanton partition function is  given by  a weighted sum over the instanton number $k$ as
\begin{align} \label{eq:fulltetra}
Z_{\FC^4}^{\mbf{r}}(\qu) = \sum_{k=0}^{\infty}\,\qu^k \ Z_{\FC^4}^{\mbf{r}, k} \ .
\end{align}

\subsubsection*{$\mbf{\Omega}$-Background}

As in Section~\ref{sec:pf_C3}, the natural choice for the torus group $\sT$ comes from defining the equivariant gauge theory on an $\Omega$-deformation of the affine Calabi--Yau four-fold~\cite{Nekrasov:2002qd,Nekrasov:2003rj}.
The global symmetry group of the tetrahedron instanton moduli space is 
\begin{align}
\sG=\, {\sf P}\sU(\mbf{r})\,\times \,\sH_{\mbf r}  \ .
\end{align}
It can be rotated to its maximal torus
\begin{align} \label{eq:tetrasT}
\sT=\sT_{{\vec\tta}}\,\times\,\sT_{\vec \epsilon} \ ,
\end{align}   
where
\begin{align}
\sT_{{\vec\tta}} = \Timesbig_{A\,\in\,\ulfour^\perp}\,\sT_{\vec \tta _{A}} \ .
\end{align}

The maximal torus of the unbroken holonomy group \eqref{eq:sHmbfr} is $\sT_{\vec\epsilon\,}$, which preserves the stratification $\FC_\triangle^3$ and whose coordinates $\vec \epsilon =(\epsilon_1,\epsilon_2,\epsilon_3,\epsilon_4)$ are the equivariant parameters of $\sSU(4)$ obeying
\begin{align}
\epsilon_{1234} = \epsilon_1+\epsilon_2+\epsilon_3+\epsilon_4=0 \ .
\end{align}
The Coulomb moduli \smash{${\vec\tta} = (\vec\tta_A)_{A\,\in\,\ulfour^\perp}$} with $\vec \tta_A=(\tta_{A\,1},\dots, \tta_{A\, r_A})$ are the vacuum expectation values of the complex Higgs field $\Phi=\bigoplus_{A\,\in\,\ulfour^\perp}\,\Phi_A$ parametrizing the positions of the $r_A$ D7$_A$-branes; they are defined modulo the overall shifts $\tta_{A\,l}\longmapsto\tta_{A\,l} + {\tt c}$ for ${\tt c}\in\FC$. 

The \smash{$\sT$}-fixed points $\vec{\mbf\pi}$ of the instanton moduli space $\frM_{\mbf r, k}$ are all isolated and finite in number~\cite{Pomoni:2021hkn}, hence the fixed point loci are compact in this instance. It follows that $\CCN_{\vec{\mbf\pi}}^{\rm vir} = T_{\vec{\mbf\pi}}^{\rm vir}\frM_{\mbf r,k}$ and the formula  \eqref{eq:euler_calss} for the equivariant square root Euler class agrees with the top form operation \eqref{eq:top}. The localization formula for the instanton partition function \eqref{eq:localization_2} then simplifies to
\begin{align}\label{eq:ZC4tetra}
Z_{\FC^4}^{\mbf r, k}(\vec \tta,\vec \epsilon\,)=\sum_{\vec{\mbf\pi}\,\in\,\frM^{\sT}_{\mbf{r}, k}} \, (-1)^{\mathtt{O}_{\vec{\mbf\pi}}} \ \prod_{\substack{i_\pm\in \ttI_\pm \\ \textsf w^{\vec{\mbf\pi}}_{i_\pm}\neq0}} \ \frac{\big(\textsf w^{\vec{\mbf\pi}}_{i_-}\big)^{\dim(\CCE_{i_-}^{\vec{\mbf\pi}})}}{\big(\textsf w^{\vec{\mbf\pi}}_{i_+}\big)^{\dim(\CCE^{\vec{\mbf\pi}}_{i_+})}} \ .
\end{align}  
The sign $(-1)^{\mathtt{O}_{\vec{\mbf\pi}}}$ depends on the local orientation of the obstruction bundle \smash{$\Ob_{\mbf r, k}^-$} at the fixed point $\vec{\mbf\pi}$, as in the case of noncommutative instantons on $\FC^4$ \cite{Cao:2023gvn,Szabo:2023ixw}. The explicit form  of the sign factor is evaluated in \cite{Fasola:2023ypx} for any choice of $\mbf r$. We  discuss this sign factor within our approach in Section~\ref{sec:pf_tetra}, as well as the explicit evalutation of \eqref{eq:ZC4tetra}.

\subsection{Quiver Matrix Model}\label{sec: tetra_matrix_model}

The maximal torus \eqref{eq:tetrasT} acts on the ADHM data as
 \begin{align}
 \big(\vec t\,,\,\mbf h\big)\cdot \big(B_a\,,\,I_A\big)_{a\,\in\,\ulfour\,,\,A\,\in\,\ulfour^\perp} = \big(t_a^{-1}\,B_a\,,\,I_A\,h_A^{-1}\big)_{a\,\in\,\ulfour\,,\,A\,\in\,\ulfour^\perp} \ ,
 \end{align}
 where $\vec t=(t_a)_{a\,\in\,\ulfour}\in\sT_{\vec\epsilon}$ with $t_a=\e^{\,\ii\,\epsilon_a}$, and $\mbf h=(h_A)_{A\,\in\, \ulfour^\perp}\in\sT_{\vec\tta}$. The cohomological field theory on the $\Omega$-background is equivariant with respect to this torus action and is constructed using the BRST formalism, which produces a quiver matrix model based on the framed quiver representation~\eqref{eq:tetraquiver}.
 
The BRST transformations are analogous to the ones for Donaldson--Thomas theory on $\FC^4$~\cite{m4,Szabo:2023ixw}. They read as 
 \begin{align}\label{eq:BRSTeq}
\begin{split}
    \mathcal{Q}B_a=\psi_a\qquad , & \qquad \mathcal{Q}\psi_a=[\phi,B_a]-\epsilon_a\,B_a \ , \\[4pt]
    \mathcal{Q}I_A=\varpi_A\qquad , & \qquad  \mathcal{Q}\varpi_A=\phi \, I_A-I_A\,\underline{\tta}_A \ ,\\[4pt]
    \mathcal{Q}\chi^\FC_{ab}=H^\FC_{ab} \qquad , & \qquad \mathcal{Q}H^\FC_{ab}=[\phi,\chi^\FC_{ab}]-\epsilon_{ab}\,\chi^\FC_{ab} \ , \\[4pt]
\mathcal{Q}\chi^\FR=H^\FR \qquad , & \qquad \mathcal{Q}H^\FR=[\phi,\chi^\FR] \ , \\[4pt]
 \mathcal{Q}\phi=0 \quad , & \quad \mathcal{Q}\bar{\phi}=\eta  \quad , \quad \mathcal{Q}\eta=[\bar{\phi},\phi] \ ,
      \end{split}
 \end{align}
for  $a\in\ulfour$, $(a,b)\in\ulthree^\perp$ and $A\in\ulfour^\perp$. Here $\phi$ is the generator of $\sU(k)$ gauge transformations and $\underline{\tta}_A = \diag (\tta_{A\,1},\dots, \tta_{A\,r_A})$ is a background field which parametrizes an element of the (complex) Cartan subalgebra of $\sU(r_A)$. The Fermi multiplets $(\vec \chi, \vec H)$ implement the equations $\mu^\FC_{ab}=0$ and $\mu^\FR=\zeta\,\ident_V$, where \smash{$\vec \chi =(\chi^\FC_{ab},\chi^\FR)_{(a,b)\,\in\,\ulthree^\perp}$} are antighost fields in $\sEnd_\FC(V)$ and \smash{$\vec H=(H^\FC_{ab}, H^\FR)_{(a,b)\,\in\,\ulthree^\perp}$} are the auxiliary fields. The scalar multiplet $(\phi,\bar\phi, \eta)$ is needed to close the BRST algebra.

In addition to the BRST transformations \eqref{eq:BRSTeq}, the equations $\sigma_A=0$ for $A\in\ulfour^\perp$ are included by adding  Fermi multiplets  $(\Upsilon_A,\xi_A)_{A\in\ulfour^\perp}$, with \smash{$\Upsilon_A\in\sHom_\FC(V,W_A)$}. These fields transform as
 \begin{align}\label{eq:BRSTeq_extra}
   \mathcal{Q}\Upsilon_A=\xi_A \qquad \text{and} \qquad \mathcal{Q}\xi_A=\underline{\tta}_A\,\Upsilon_A-\Upsilon_A\,\phi+\epsilon_{\bar A}\, \Upsilon_A \ ,
 \end{align}
 for $ A\in\ulfour^\perp$.

The action functional corresponding to this system of symmetries, fields and equations is 
\begin{align}
 \begin{split}
     S=\mathcal{Q}\,\Tr_V\, &\bigg(\sum_{(a,b)\,\in\,\ulthree^\perp}\,\chi_{ab}^{\FC\,\dagger }\,\big(H_{ab}^{\FC}-\mu_{ab}^\FC\big) +\chi^{\FR}\,\big( H^\FR-\mu^\FR-\zeta\,\ident_V\big) +\sum_{A\,\in\,\ulfour^\perp}\,\Upsilon_{A}^{\dagger}\,(\xi_{A}-\sigma_{A}) \\
     & \qquad +\sum_{a\,\in\,\ulfour}\,\psi_a\,[\Bar{\phi},B_a]+\sum_{A\,\in\,\ulfour^\perp}\, \Bar{\phi}\,I_A\,\varpi_A^{\dagger}+\eta\,[\phi,\Bar{\phi}]+\text{c.c.}\bigg) \ ,
         \end{split}
 \end{align}
where $\text{c.c.}$ means complex conjugate. The evaluation of the corresponding path integral is now a routine computation which follows by combining the calculations from \cite{Jafferis:2007sg,Cirafici:2008sn} (for the field theory on $\FC^3$) with those of \cite[Section 2.5]{Szabo:2023ixw} (for the field theory on $\FC^4$). 

The matrix model representation of the sum over equivariant Euler classes in \eqref{eq:ZC4tetra} is given by 
\begin{equation}
\begin{split}
Z_{\FC^4}^{\mbf r,k}(\vec \tta \, , \vec\epsilon\,)  &=  \,
\frac{(-1)^k}{k!} \, \Big(\frac{\epsilon_{12} \, \epsilon_{13} \, \epsilon_{23}}{\epsilon_1\,\epsilon_2\,\epsilon_3\,\epsilon_{123}}\Big)^k  \ \oint_{\varGamma_{\mbf r,k}} \ \prod_{i=1}^k \, \frac{\dd\phi_i}{2\pi\,\ii} \ \prod_{A\,\in\,\ulfour^\perp}\, \frac{P_{r_A}(-\phi_i-\epsilon_{\bar A}|-\vec \tta_A)}{P_{r_A}(\phi_i|\vec \tta_A)}  \\
&\hspace{8cm} \times \prod_{\stackrel{\scriptstyle i,j=1}{\scriptstyle i\neq j}}^k \, R_-(\phi_i-\phi_j|\vec\epsilon\,) \ ,
\end{split} \label{eq:Tetra_Matrix_Model}
\end{equation}
where $\phi_i$ for $i=1,\dots,k$ are the components of $\phi$ in a Cartan subalgebra of $\sU(k)$, and $R_-(x|\vec\epsilon\,)$ is the rational function from \eqref{eq:CPCRdef} evaluated with the opposite sign of $\epsilon_4=-\epsilon_{123}$. Similarly to \cite{Szabo:2023ixw,Cirafici:2010bd}, the ADHM matrix model \eqref{eq:Tetra_Matrix_Model} is interpreted as a contour integral. The closed contour $\varGamma_{\mbf r,k}\subset\FC^k$ encircles all poles of the integrand, which are located along the hyperplanes 
\begin{align}\label{Tetra_fixed_points}\begin{split}
    \phi_i-\phi_j-\epsilon_a=0 \qquad \text{and} \qquad 
    \phi_i-\tta_{A\,l}=0 \end{split}
\end{align}
in $\FR^k$, for $i,j=1,\dots,k$, $a\in\ulfour\,$, $A\in\ulfour^\perp$ and $l=1,\dots,r$. 

\subsubsection*{Fixed Points and Plane Partitions}

The intersections of the hyperplanes \eqref{Tetra_fixed_points} are precisely the BRST fixed points of the cohomological field theory, which by construction are the $\sT$-fixed points of the ADHM moduli space. 
The residue formula then reproduces the sum over $\sT$-fixed points in \eqref{eq:ZC4tetra}. The full instanton partition function $Z_{\FC^4}^{\mbf r}(\qu;\vec\tta,\vec\epsilon\,)$ is given by the sum \eqref{eq:fulltetra} over all instanton numbers $k$.

The fixed points \eqref{Tetra_fixed_points} are in one-to-one correspondence with collections of plane partitions~\cite{Pomoni:2021hkn,Cirafici:2008sn}
\begin{align}\label{eq:plane_partition_A}
    \vec{\mbf\pi}=(\vec\pi_A)_{A\,\in\,\ulfour^\perp} =(\pi_{A\,1},\dots,\pi_{A\,r_A})_{A\,\in\,\ulfour^\perp} \ ,
\end{align}
where the total size of $\vec{\mbf\pi}$ is the instanton number
\begin{align}
    k=|\vec{\mbf\pi}|=\sum_{A\,\in\,\ulfour^\perp} \, |\vec\pi_A| = \sum_{A\,\in\,\ulfour} \ \sum_{l=1}^{r_A}\,|\pi_{A\,l}| \ .
\end{align}
As we will see in Section~\ref{sec:pf_tetra}, the total number of boxes $k_A:=|\vec\pi_A|$ for each $A\in\ulfour^\perp$ is the complex dimension of the vector space $V_A$ introduced in Remark~\ref{rem:tetrastability} at the fixed point $\vec{\mbf\pi}$.

\subsubsection*{Dimensional Reduction}

The structure of \eqref{eq:Tetra_Matrix_Model} is very similar to that of the matrix integral $\mbf \CZ_{\FC^4}^{r,k}(\vec a,\vec\epsilon,\vec m)$ for the rank $r$ Donaldson--Thomas invariants of $\FC^4$ that was obtained in~\cite[eq.~(2.68)]{Szabo:2023ixw}, where $\vec a=(a_1,\dots,a_r)$ are the Coulomb moduli and $\vec m=(m_1,\dots,m_r)$ are the masses of $r$ fundamental matter fields. This similarity is made precise through

\begin{proposition}\label{prop:Tetra_reduction}
There exist Coulomb parameter and mass specialisations such that the equivariant instanton partition functions of the $\sU(r)$ cohomological field theory with a massive fundamental hypermultiplet on $\FC^4$ and the cohomological field theory for tetrahedron instantons of rank $|\mbf r|=r$ are related as
\begin{align} \label{eq:ZC4rtetra}
\mbf\CZ_{\FC^4}^{r}(\qu;\vec a,\vec\epsilon ,\vec m) = Z_{\FC^4}^{\mbf r}\big((-1)^{r}\,\qu;\vec \tta,\vec \epsilon\,)   \ .
\end{align}
\end{proposition}

\begin{proof}
Choose a partition of the set of colour labels \smash{$\{1,\dots,r\} = \bigsqcup_{A\,\in\,\ulfour^\perp}\,\varsigma_A$} into disjoint subsets $\varsigma_A$ of cardinalities $\#\varsigma_A = r_A$ for \smash{$A\in\ulfour^\perp$}. 
In $\mbf\CZ_{\FC^4}^{r,k}(\vec a,\vec\epsilon,\vec m)$ we substitute
\begin{align}\label{eq:sigma_function}
(a_l,m_l) = \big(\tta_{A\,l}\,,\,\tta_{A\,l}+\epsilon_{\bar A}\big) \qquad \mbox{for} \quad l\in\varsigma_A \ .
\end{align}
Using the Calabi--Yau condition $\epsilon_4=-\epsilon_{123}$ on $\FC^4$ one then finds that the matrix integral from~\cite[eq.~(2.68)]{Szabo:2023ixw} is exactly the  integral~\eqref{eq:Tetra_Matrix_Model}, up to an overall sign $(-1)^{r\,k}$:
\begin{align} \label{eq:ZC4rktetra}
\mbf\CZ_{\FC^4}^{r,k}(\vec a,\vec\epsilon ,\vec m) =(-1)^{r\,k} \ Z_{\FC^4}^{\mbf r,k}(\vec \tta,\vec \epsilon\,)   \ ,
\end{align}
and the result follows by taking the weighted sum over $k\in\RZ_{\geq0}$ of \eqref{eq:ZC4rktetra}.
\end{proof}

\begin{remark} \label{rem:indeppart}
From the matrix model representation~\eqref{eq:Tetra_Matrix_Model} we deduce that the partition function $Z_{\FC^4}^{\mbf r}(\qu;\vec\tta, \vec\epsilon\,)$ for tetrahedron instantons is invariant under permutations of  the entries $r_A$ in $\mbf r$. It follows that the result of Proposition~\ref{prop:Tetra_reduction} is independent of the choice of partition $\{\varsigma_A\}_{A\,\in\,\ulfour^\perp}$.

Proposition~\ref{prop:Tetra_reduction} generalizes~\cite[Proposition~2.71]{Szabo:2023ixw}. In the T-dual type~IIA picture~\cite{Pomoni:2023nlf}, the specialisations can be interpreted as particular configurations of D8-branes and anti-D8-branes which decay via tachyon condensation into intersecting D6$_A$-branes for $A\in\ulfour^\perp$, whose bound states with D0-branes correspond to tetrahedron instantons.
\end{remark}

\subsection{Tetrahedron Instanton Partition Function}\label{sec:pf_tetra}

We explicitly evaluate the equivariant partition function in the $\Omega$-background from the formula \eqref{eq:ZC4tetra}, elaborating on the calculation that appears in~\cite{Pomoni:2021hkn}, which in particular does not incorporate the sign dependence on the choice of orientations. For this, we regard the vector space $Q_4$ as the four-dimensional $\sT_{\vec \epsilon}\,$-module  with weight decomposition
\begin{align}\label{eq:Q_tetra}
Q_4=t_1^{-1}+t_2^{-1}+t_3^{-1}+t_4^{-1} \ .
\end{align}
The weights $t_a = \e^{\,\ii\,\epsilon_a}$ for $a\in\ulfour$ satisfy the Calabi–Yau condition
\begin{align}
t_1\, t_2\, t_3\, t_4 = 1 \ .
\end{align}

The fibre of the index bundle \eqref{eq:index} over the fixed point \smash{$\vec{\mbf\pi}\in\frM_{\mbf r,k}^\sT$} computes the square root of the $\sT$-equivariant Chern character of the virtual tangent bundle $T^{\rm vir}\frM_{\mbf r,k}$ at $\vec{\mbf\pi}$. It is given by
\begin{align} \begin{split}
\sqrt{\ch_{\sT}}\big(T^{\rm vir}_{\vec{\mbf\pi}}\frM_{\mbf r,k}\big) &=-V^*_{\vec{\mbf\pi}}\otimes V_{\vec{\mbf\pi}} \ \Big(1-\sum_{a\,\in\,\ulfour}\,t_a^{-1}+t_1^{-1}\,t_2^{-1}+t_1^{-1}\,t_3^{-1}+t_2^{-1}\,t_3^{-1}\Big) \\ & \qquad \, +\sum_{A\,\in\,\ulfour^\perp}\,W^*_A{}_{\vec{\mbf\pi}}\otimes V_{\vec{\mbf\pi}}-V^*_{\vec{\mbf\pi}}\otimes W_A{}_{\vec{\mbf\pi}} \ t_{\bar A} \ ,
\end{split}
\label{eq:tetra_chi}
\end{align}
where
\begin{align} \label{eq:VWpitetra}
\begin{split}
V_{\vec{\mbf\pi}}&=\sum_{A=(a\,b\,c)\,\in\,\ulfour^\perp} \ \sum_{l=1}^{r_A} \, e_{{A\,l}} \ \sum_{\vec p_A\,\in\,\pi_{A\,l}}\, t_a^{p_a-1}\,t_b^{p_b-1}\,t_c^{p_c-1} \qquad \mbox{and} \qquad  W_A{}_{\vec{\mbf\pi}}=\sum_{l=1}^{r_A}\,e_{{A\,l}}
\end{split}
\end{align}
as elements of the representation ring of $\sT$, with $e_{A\,l}=\e^{\,\ii\,\tta_{A\,l}}$ for $A\in\ulfour^\perp$ and $l=1,\dots,r$. We used the stability condition, see Remark~\ref{rem:tetrastability}, along with a suitable gauge transformation.

\begin{proposition}\label{prop:signfactor}
For the choice of square root \eqref{eq:tetra_chi}, the sign factor $(-1)^{\texttt{O}_{\vec{\mbf\pi}}}$ is given by
\begin{align}\label{eq:signfactor}
\texttt{O}_{\vec{\mbf\pi}}= \rk\, \big( V_{\vec{\mbf\pi}}^*\otimes V_{\vec{\mbf\pi}} \ t_4^{-1}\big)^{\rm fix} \mod  2 \ .
\end{align}
\end{proposition}

\proof
By Proposition \ref{prop:Tetra_reduction}, the Donaldson--Thomas partition function on $\FC^4$ reduces to the tetrahedron instanton partition function though specialisations of the parameters $(\vec a,\vec m)$; by Remark~\ref{rem:indeppart} the reduction is independent of the choice of partitioning of the index set $\{1,\dots,r\}$. By regarding the plane partitions $\pi_{A\,l}$ as solid partitions via the natural embedding $\RZ^3_{\geq0} \lhook\joinrel\longrightarrow\RZ^4_{\geq0}$, the sign factor is~\cite[Remark~2.103]{Szabo:2023ixw}
\begin{align}\label{sign_condition}
{\texttt{O}}_{{\vec{\mbf\pi}}} = \sum_{A\,\in\,\ulfour^\perp} \ \sum_{l=1}^{r_A} \, \#\big\{(p,p,p,p')\in \pi_{A\,l} \ \big| \ p<p'\big\} - \rk\,\big( V_{\vec{\mbf\pi}}^*\otimes V_{\vec{\mbf\pi}} \ t_4^{-1}\big)^{\rm fix}  \mod 2 \ ,
\end{align}
where the summands in the first term correspond to the Nekrasov--Piazzalunga sign prescription~\cite{m4,m4c}.
Since $\pi_{A\,l}$  are true plane partitions, $p'=0$ and each summand in the first term of \eqref{sign_condition} is zero modulo $ 2$.
\endproof

\begin{remark} 
The sign factor \eqref{eq:signfactor} is consistent with the sign factor evaluated by Fasola and Monavari in \cite{Fasola:2023ypx}. It can be evaluated explicitly from \eqref{eq:VWpitetra} by counting the zeroes of the combination $\tta_{A\,l}-\tta_{A'\,l'}+(\vec p_{A\,l}-\vec p^{\,\prime}_{A'\,l'}) \cdot\vec\epsilon$ for $A,A'\in\ulfour^\perp$, $l\in\{1,\dots,r_A\}$, $l'\in\{1,\dots,r_{A'}\}$, $\vec p_{A\,l}\in\pi_{A\,l}$ and $\vec p^{\,\prime}_{A'\,l'}\in\pi_{A'\,l'}$, where $\vec p_A\cdot\vec\epsilon := \sum_{a\in A}\,p_a\,\epsilon_a$. For generic equivariant parameters, the result is given by the sum of cardinalities
\begin{align}\label{eq:signexplicit}
\texttt{O}_{\vec{\mbf\pi}}= \sum_{A\,\in\,\ulfour^\perp} \ \sum_{l=1}^{r_A} \, \#\big\{(\vec p_{A\,l},\vec p^{\,\prime}_{A\,l})\in\pi_{A\,l}\times\pi_{A\,l} \ \big| \ p_{a\,l} = p_{a\,l}'+1 \ , \ a\in A\big\} \ .
\end{align}
\end{remark}

Finally, the equivariant partition function for tetrahedron instantons can be easily  evaluated from the character \eqref{eq:tetra_chi}. It is given by the combinatorial expression 
\begin{align}
Z_{\FC^4}^{\mbf{r}}(\qu;\vec{\tta},\vec{\epsilon}\,)=\sum_{k= 0}^\infty \ \sum_{\vec{\mbf\pi}\,\in\,\frM_{\mbf r,k}^\sT} \,  (-1)^{\texttt{O}_{\vec{\mbf\pi}}} \ \qu^{|\vec{\mbf\pi}|}  \ \widehat{\texttt e}\Big[-\sqrt{\ch_{\sT}}(T_{\vec{\mbf\pi}}^{\rm vir} {\frM}_{\mbf r,k}) \Big] \ ,\label{Z_tetra_inst_gen}
\end{align}
where  
\begin{align}
\begin{split}
  \widehat{\texttt e}\Big[-\sqrt{\ch_{\sT}}(T_{\vec{\mbf\pi}}^{\rm vir} {\frM}_{\mbf r,k}) \Big] &=\prod_{A,A'\,\in\,\ulfour^\perp} \ \prod_{l=1}^{r_A} \ \prod_{\vec p_{A\,l}\,\in\,\pi_{A\,l}}^{\neq0} \, \frac{P_{r_{A'}}(-\tta_{A\,l}-\vec p_{A\,l}\cdot\vec\epsilon+\epsilon_{\bar A}\,|-\vec \tta_{A'})}{P_{r_{A'}}(\tta_{A\,l}+\vec p_{A\,l}\cdot\vec\epsilon\,|\vec \tta_{{A'}})} \\
& \hspace{1.5cm} \times  \prod_{l'=1}^{r_{A'}} \ \prod_{\vec p^{\,\prime}_{{A'},l'}\,\in\,\pi_{{A'},l'}}^{\neq0} \, R_-\big(\tta_{A\,l}-\tta_{{A'},l'}+(\vec p_{A\,l}-\vec p^{\,\prime}_{{A'},l'})\cdot\vec\epsilon\,\big|\vec\epsilon\,\big) \ .
\end{split}\label{Zk}
\end{align}

\subsubsection*{Refined Partition Function}

The structure \eqref{eq:plane_partition_A} of the fixed points $\vec{\mbf\pi}$ suggests a refinement of the counting parameter $\qu$ in \eqref{Z_tetra_inst_gen} with four independent fugacities $\qu_A$ weighing the contributions from $\vec\pi_A$ for each face $A\in\ulfour^\perp$. We set $\textbf{q} = (\qu_A)_{A\,\in\,\ulfour^\perp}$ and define
\begin{align}
\textbf{q}^{\vec{\mbf\pi}} := \prod_{A\,\in\,\ulfour^\perp} \, \qu_A^{|\vec\pi_A|} \ .
\end{align}
The \emph{refined} partition function for tetrahedron instantons enumerates instantons on each of the codimension one strata $\FC_A^3\subset\FC_\triangle^3$ for $A\in\ulfour^\perp$ and is defined by
\begin{align}
Z_{\FC^4}^{\mbf{r}}(\textbf{q};\vec{\tta},\vec{\epsilon}\,)=\sum_{k= 0}^\infty \ \sum_{\vec{\mbf\pi}\,\in\,\frM_{\mbf r,k}^\sT} \,   (-1)^{\texttt{O}_{\vec{\mbf\pi}}} \ \textbf{q}^{\vec{\mbf\pi}}  \ \widehat{\texttt e}\Big[-\sqrt{\ch_{\sT}}(T_{\vec{\mbf\pi}}^{\rm vir} {\frM}_{\mbf r,k}) \Big] \ .
\label{Z_tetra_inst_ref}
\end{align}
The generating function \eqref{Z_tetra_inst_gen}, in which $\qu_A=\qu$ for all $A\in\ulfour^\perp$, will sometimes be referred to as the \emph{unrefined} partition function.

\subsubsection*{Instantons on $\mbf{\FC^3_A}$}

If only $r_A=r$ is non-zero for some $A=(a\,b\,c)\in\ulfour^\perp$, we write $\mbf r_A$ for the rank vector and $\vec{\mbf\pi}_A$ for the fixed points. The character \eqref{eq:tetra_chi} is given by
\begin{align}\begin{split}\label{Chir}
\sqrt{\ch_{\sT}}\big(T_{\vec{\mbf\pi}_A}^{\rm vir} {\frM}_{\mbf r_A,k}\big) &= W^*_{A}{}_{\vec{\mbf\pi}_A}\otimes V_{\vec{\mbf\pi}_A}-\frac{V^*_{\vec{\mbf\pi}_A}\otimes W_{A}{}_{\vec{\mbf\pi}_A}}{t_a\,t_b\,t_c}+V^*_{\vec{\mbf\pi}_A}\otimes V_{\vec{\mbf\pi}_A} \ \frac{(1-t_a)\,(1-t_b)\,(1-t_c)}{t_a\,t_b\,t_c} \\
& \qquad \, + C_{\vec{\mbf\pi}_A}-C_{{\vec{\mbf\pi}_A}}^* \ ,
\end{split}
\end{align}
where
\begin{align}
C_{\vec{\mbf\pi}_A} := V^*_{\vec{\mbf\pi}_A}\otimes V_{\vec{\mbf\pi}_A} \ t_a\,t_b\,t_c \ .
\end{align}

The only contribution of the term \smash{$C_{\vec{\mbf\pi}_A}-C^*_{\vec{\mbf\pi}_A}$} to the partition function \eqref{eq:ZC4tetra} is by the sign factor \smash{$(-1)^{\rk\, C_{\vec{\mbf\pi}_A}}$}, where
\begin{align}\label{eq:RankC}
 \rk\, C_{\vec{\mbf\pi}_A} = |\vec\pi_A| + \rk\, \big( V_{\vec{\mbf\pi}_A}^*\otimes V_{\vec{\mbf\pi}_A} \ t_{\bar A}^{-1}\big)^{\rm fix} \mod 2 \ .
\end{align}
Note that the second term of \eqref{eq:RankC} coincides with the sign factor \eqref{eq:signfactor}.
Instead, the remaining terms of \eqref{Chir} form the equivariant character of the instanton deformation complex \eqref{eq:complex_C3} for noncommutative instantons on $\FC^3_A$. 

By Theorem~\ref{thm:C3inst} it  follows that the refined tetrahedron instanton partition function sums to
\begin{align}
Z^{\mbf{r}_A}_{\FC^4}(\qu_{A};\vec{\epsilon}\,)=M\big((-1)^{r+1}\,\qu_{A}\big)^{-r\,\frac{\epsilon_{ab}\,\epsilon_{ac}\,\epsilon_{bc}}{\epsilon_a\,\epsilon_b\,\epsilon_c}} \ , \label{eq:ZrC3}
\end{align}
which after redefinition $\qu_{A}=-\qu$ coincides with the partition function for noncommutative instantons on  $\FC_A^3$ with holonomy group $\sH_{\mbf r_A}=\sU(3)_A$. This agrees with the discussion of Section~\ref{sec: tetra_matrix_model}, and suggests that the partition function for tetrahedron instantons is related to the partition functions for instantons on $\FC^3$ and $\FC^4$. These expectations are borne out below.

\subsubsection*{Generic $\mbf r$}

The previous considerations generalize to

\begin{proposition}\label{prop:pf_tetra}
The unrefined equivariant partition function $Z_{\FC^4}^{\mbf{r}}(\qu;\vec{\tta},\vec{\epsilon}\,)$ for tetrahedron instantons is independent of the Coulomb moduli $\vec\tta$ and can be expressed as 
\begin{align}\label{PfRankpure}
Z_{\FC^4}^{\mbf r}(\qu; \vec\epsilon\,)= M\big((-1)^{|\mbf r|+1}\,\qu\big)^{-\sum\limits_{A\in\ulfour^\perp}\,r_A\,\epsilon_{\bar A}\,\frac{\epsilon_{12}\,\epsilon_{13}\,\epsilon_{23}}{\epsilon_1\,\epsilon_2\,\epsilon_3\,\epsilon_4}} \ .
\end{align}
\end{proposition}

\proof 
The equivariant instanton partition function of the cohomological gauge theory with a massive fundamental hypermultiplet on $\FC^4$ is given by \cite{m4c,Szabo:2023ixw,KRinprep}
\begin{align}
\mbf\CZ_{\FC^4}^r(\qu;\vec{a},\vec{\epsilon},\vec{m})= M(-\qu)^{-r\,m\,\frac{\epsilon_{12}\,\epsilon_{13}\,\epsilon_{23}}{\epsilon_1\,\epsilon_2\,\epsilon_3\,\epsilon_4}} \qquad \mbox{with} \quad m=\frac{1}{r} \, \sum_{l=1}^r\,(m_l-a_l) \ . 
\end{align}
The result now follows immediately from Proposition~\ref{prop:Tetra_reduction}.
\endproof

\section{Tetrahedron Instantons on Local Calabi--Yau Four-Orbifolds}\label{sec:Tetra_orb}

In this section we extend our considerations of tetrahedron instantons from Section \ref{sec:tetra_insta} to twisted Calabi--Yau orbifold resolutions of quotient singularities $\FC^4/\,\sGamma^\tau$, where $\tau:\sGamma\longrightarrow\sH_{\mbf r}$ is a homomorphism from a finite group $\sGamma$ to the unbroken holonomy group \eqref{eq:sHmbfr} fixing the smooth strata $\FC_A^3\subset\FC^3_\triangle$ of the singular three-fold \eqref{eq:C4strat} supporting the instanton type $\mbf r$. That is, as opposed to generic $\sSU(4)$-instantons on $\FC^4$, for tetrahedron instantons we consider only defect-preserving orbifold group actions, which generally restricts the allowed  dimension vectors \smash{$\mbf r = (r_A)_{A\,\in\,\ulfour^\perp}$} in order to allow for non-trivial groups $\sGamma^\tau$ inside $\sH_{\mbf r}$. In this construction both the singular three-fold $\FC_\triangle^3\subset\FC^4$ and its normal bundle may be subjected to the orbifold projection. 
We handle separately the cases where $\sGamma$ is an abelian and a non-abelian group, expanding the analysis and results of Section \ref{sec:C3orbifold}.
  
\subsection{Tetrahedron Instantons on Abelian Orbifolds}

Let $\Ab$ be a finite abelian group. It is straightforward to define (non-effective) actions of $\Ab$ on $\FC^4$ analogously to what we did in Section~\ref{sec:orb_3d_inst}, and compute orbifold instanton partition functions similarly to Section~\ref{sec:orb_3d_pf}. However, for clarity and to streamline notation a bit, we will restrict our considerations of abelian orbifolds to the cases where $\tau$ is a monomorphism. The McKay quivers in these instances have been described in detail in~\cite{Szabo:2023ixw}.

Let $\Ab$ be a finite abelian subgroup of $\sH_{\mbf r}\subset\sSU(4)$ which commutes with the maximal torus $\sT_{\vec\epsilon\,}$; it is of the form $\Ab=\RZ_{n_1}\times\RZ_{n_2}\times\RZ_{n_3}$ with order $n=n_1\,n_2\,n_3$ and is diagonally embedded in $\sSU(4)$. Then $\FC^4/\,\Ab$ is a toric Calabi--Yau four-orbifold. Let $\rho_s$ denote the irreducible representation of $\Ab$ with weight $s$; the trivial representation is $\rho_0$. The restriction to $\Ab$ of the fundamental representation $Q_4$ of $\sSU(4)$ branches into irreducible $\Ab$-modules as
\begin{align}
Q_4 \ \simeq \ \rho_{s_1}\oplus\rho_{s_2}\oplus\rho_{s_3}\oplus\rho_{s_4} \ .
\end{align}
By the Calabi--Yau condition, $\rho_{s_1}\otimes\dots\otimes\rho_{s_4}\simeq \rho_0$.
Under the group isomorphism $\widehat{\sGamma}_{\textsf{ab}}\simeq\Ab$, this induces a corresponding colouring  $\RZ_{\geq 0}^{\oplus 4}\longrightarrow\Ab$  given by
\begin{align}\label{eq:coloring4d}
(n_1,n_2,n_3,n_4)\longmapsto \rho_{s_1}^{\otimes n_1}\otimes\rho_{s_2}^{\otimes n_2}\otimes\rho_{s_3}^{\otimes n_3}\otimes\rho_{s_4}^{\otimes n_4} \ .
\end{align}

The choice of an abelian orbifold group $\Ab$ leaves unbroken the maximal torus $\sT$ of the theory in the $\Omega$-background, because we assume $\Ab$ commutes with $\sT_{\vec \epsilon}\subset\sSU(4)$. In this case, there is no restriction on the type $\mbf r$ labelling the solutions of \eqref{eq:tetrasols}. Since the irreducible representations of $\Ab$ are all one-dimensional, the $\sT$-fixed points of the tetrahedron instanton moduli space are all isolated and are in one-to-one correspondence with plane partitions coloured via the map \eqref{eq:coloring4d}.

\subsubsection*{ADHM Data}
 
The $\Ab$-action on $\FC^4$ induces an equivariant decomposition of the vector spaces
\begin{align}\label{eq:decomVW_ab_tetra}
    V=\bigoplus_{s\in\Abw}\, V_s\otimes \rho^*_s \qquad \mbox{and} \qquad
    W_A=\bigoplus_{s\in\Abw}\, W_{A\,s}\otimes\rho^*_s 
\end{align}
for $A\in \ulfour^\perp$, where $V_s$ and $W_{A\,s}$ are Hermitian vector spaces of complex dimensions $k_s$ and $r_{A\,s}$, respectively, each carrying a trivial $\Ab$-action. 
 We define dimension vectors \smash{$\vec k=(k_s)_{s\in\Abw}$} and \smash{$\vec{\mbf r}=(\vec r_{A})_{A\,\in\,\ulfour^\perp} = (r_{A\,s})_{A\,\in\,\ulfour^\perp\,,\,s\in\Abw}$}, with
\begin{align}\begin{split}
    k = |\vec k\,|:=\sum_{s\in\Abw}\, k_s \qquad \mbox{and} \qquad r=\sum_{A\,\in\,\ulfour^\perp} \, r_A=|\vec{\mbf r}\,| := \sum_{A\,\in\,\ulfour^\perp} \, |\vec r_A| = \sum_{A\,\in\,\ulfour^\perp} \  \sum_{s\in\Abw} \, r_{A\,s}\ .
    \end{split}
\end{align}

By Schur’s Lemma, the decompositions \eqref{eq:decomVW_ab_tetra} induce equivariant decompositions of the ADHM variables as
\begin{align}
\begin{split}
B&=\bigoplus_{s\in\Abw}\, (B_a^s)_{a\,\in\,\ulfour} \ \in \  \sHom_\Ab(V, V\otimes Q_4) \qquad \mbox{with} \quad B^s_a: V_{s}\longrightarrow V_{s+s_a} \ , \\[4pt]
I_A&=\bigoplus_{s\in\Abw}\, I_{A}^s \ \in \ \sHom_\Ab(W_A, V) \qquad \mbox{with} \quad I^s_A: W_{A\,s}\longrightarrow V_s \ .
\end{split}
\end{align} 
Consequently, the ADHM equations \eqref{eq:ADHM_tetra} for tetrahedron instantons decompose into
\begin{align}\begin{split}\label{eq:ADHM_tetra_orb}
\mu^{\FC s}_{ab}&=B_a^{s+s_b}\,B_{b}^s - B_b^{s+s_a}\,B_a^s - \tfrac12\,\varepsilon_{abcd}\,\big(B_c^{s-s_{cd} \dag}\,B_d^{s-s_d \dag} - B_d^{s-s_{cd} \dag}\,B_c^{s-s_c \dag}\big) = 0 \ , \\[4pt]
   \mu^{\FR s} &= \sum_{a\,\in\,\underline{4}}\, \big(B_a^{s-s_a}\,B_a^{ s-s_a \dag}-B_a^{ s \dag}\,B_a^{s}\big)+\sum_{A\,\in\,\ulfour^\perp}\,I_{A}^s\,I_{A}^{s\dagger} =\zeta_s\,\ident_{V_s} \ , \\[4pt]    
    \sigma_{A}^s&=\big(B_{\bar A}^s\,I_{A}^s\big)^\dagger=0 \ ,
    \end{split}
\end{align}
for $s\in\Abw$, $a,b\in\ulfour$ and $A\in\ulfour^\perp$. 

The symmetry group of the system of equations \eqref{eq:ADHM_tetra_orb} is
\begin{align}
\sU(\vec k\,)\,\times\,\sU(\vec{\mbf r}\,) := \Timesbig_{s\in\Abw}\,\sU(k_s) \ \times \ \Timesbig_{s\in\Abw} \ \Timesbig_{A\,\in\,\ulfour^\perp}\,\sU(r_{A\,s}) \ ,
\end{align}
which acts on the ADHM variables as
\begin{align}
\big(g_s\, ,\,h_A^s\big)_{\stackrel{\scriptstyle s\in\Abw}{\scriptstyle A\,\in\,\ulfour^\perp}} \cdot \big(B_a^s\,,\, I^s_A\big)_{\stackrel{\scriptstyle s\in\Abw}{\scriptstyle a\,\in\,\ulfour\,,\,A\,\in\,\ulfour^\perp}} = \big( g_{s+s_a}\,B_a^s\, g^{-1}_{s}\ ,  \, g_s\, I_{A}^s\, h^s_A \big)_{\stackrel{\scriptstyle s\in\Abw}{\scriptstyle a\,\in\,\ulfour\,,\,A\,\in\,\ulfour^\perp}} \ , 
\end{align}
for $g_s\in\sU(k_s)$ and $h_A^s\in\sU(r_{A\,s})$. There is an additional $\sH_{\mbf r}$ symmetry which acts in the fundamental representation $Q_4$ on $(B_a^s)_{a\,\in\,\ulfour}$ for all $s\in\Abw$ and trivially on all $I_A^s$.

\subsection{Abelian Orbifold Partition Functions}
\label{sub:abelianpartfn}

The equivariant generating functions for tetrahedron instantons on abelian orbifolds of $\FC^4$ can be evaluated by equivariant decompositions of the theory in the $\Omega$-background from Section~\ref{sec:tetra_insta}.

\subsubsection*{Cohomological Field Theory on $\mbf{[\FC^4/\,\Ab]}$}

By decomposing all fields as equivariant maps, the $\Ab$-module structure splits the BRST transformations \eqref{eq:BRSTeq} and \eqref{eq:BRSTeq_extra} into irreducible representations labelled by $s\in\Abw$. They read as
\begin{align}\label{eq:BRSTeq_orb}
\begin{split}
    \mathcal{Q}_\Ab B^s_a=\psi^s_a\qquad , & \qquad \mathcal{Q}_\Ab\psi^s_a=\phi^{s+s_a}\,B_a^s-B_a^s\,\phi^s-\epsilon_a\,B^s_a \ , \\[4pt]
    \mathcal{Q}_\Ab I^s_A=\varpi^s_A\qquad , & \qquad  \mathcal{Q}_\Ab\varpi^s_A=\phi^s \, I^s_A-I^s_A\,\underline{\tta}^s_A \ , \\[4pt]
    \mathcal{Q}_\Ab\chi^{\FC s}_{ab}=H^{\FC s}_{ab} \qquad , & \qquad \mathcal{Q}_\Ab H^{\FC s}_{ab}=\phi^{s+s_{ab}}\,\chi^{\FC s}_{ab}-\chi^{\FC s}\,\phi^s-\epsilon_{ab}\,\chi^{\FC s}_{ab} \ , \\[4pt]
\mathcal{Q}_\Ab\chi^{\FR s}=H^{\FR s} \qquad , & \qquad \mathcal{Q}_\Ab H^{\FR s}=[\phi^s,\chi^{\FR s}] \ , \\[4pt]
 \mathcal{Q}_\Ab\Upsilon^s_A=\xi^s_A \qquad , & \qquad \mathcal{Q}_\Ab\xi^s_A=\underline{\tta}^s_A\,\xi^s_A-\Upsilon^s_A\,\phi^s+\epsilon_{\bar A}\,\Upsilon^s_A \ ,\\[4pt]
 \mathcal{Q}_\Ab\phi^s=0 \quad & , \quad \mathcal{Q}_\Ab\bar{\phi}^s=\eta^s \quad , \quad \mathcal{Q}_\Ab\eta^s=[\bar{\phi}^s,\phi^s] \ ,
      \end{split}
 \end{align}
for $a\in\ulfour\,$, $(a,b)\in\ulthree^\perp$, $A\in\ulfour^\perp$ and $s\in\Abw$. Here $\phi^s$  parametrizes infinitesimal $\sU(k_s$) gauge transformations, while $\underline{\tta}_A^s$ collects the $r_{A\,s}$ Coulomb moduli $\tta_{A\,l}$ all associated to the irreducible representation $\rho_s$. This defines a map 
\begin{align}
l\longmapsto \tts_A(l) \ \in \ \Abw \qquad \mbox{for} \quad l\in\{1,\dots,r_A\} \ . 
\end{align}

The BRST transformations \eqref{eq:BRSTeq_orb} can be used to construct the abelian orbifold matrix model with standard techniques, as we did in Section \ref{sec: tetra_matrix_model}. This results in the partition function
\begin{equation}
\begin{split}
    Z^{\vec{\mbf r} ,\vec k}_{[\FC^4/\,\Ab]}(\vec{\tta},\vec{\epsilon}\,)&= \oint_{\varGamma_{\vec{\mbf r},\vec k}} \ \prod_{s\in\Abw} \, \frac1{k_s!} \ \prod_{i=1}^{k_s} \, \frac{\dd\phi_i^{s}}{2\pi\, \ii} \ \prod_{A\,\in\,\ulfour^\perp} \, \frac{P_{r_{A\,s}}(-\phi_i^s-\epsilon_{\bar A}|-\vec{\tta}_A^{\,s+s_{\bar A}})}{P_{r_{A\,s}}(\phi_i^s|\vec{\tta}_A^{\,s})} \ 
    \prod_{\stackrel{\scriptstyle i,j=1}{\scriptstyle i\neq j}}^{k_s} \, \big(\phi^{s}_i-\phi^{s}_j\big) \\
    & \hspace{6cm} \times \ \prod_{i,j=1}^{k_s} \ \frac{\displaystyle \prod_{(a,b)\,\in\,\ulthree^\perp}\,\big(\phi_i^{s+s_{ab}}-\phi_j^{s}-\epsilon_{ab}\big)}{\displaystyle \prod_{a\,\in\,\ulfour}\,\big(\phi_i^{s+s_a}-\phi_j^{s}-\epsilon_a\big)} \ . \label{eq:abelian_matrix_model}
    \end{split}
\end{equation}

The contour $\varGamma_{\vec{\mbf r},\vec k}\subset\FC^k$ encloses the poles of the integrand, which are situated along the hyperplanes
\begin{align}
\phi_i^{s+s_a}-\phi_j^s-\epsilon_a=0 \qquad \mbox{and} \qquad \phi_i^s-\tta_{A\,l}^s=0
\end{align}
in $\FR^k$, for $i,j=1,\dots,k$, $s\in\Abw$, $a\in\ulfour$, $A\in\ulfour^\perp$ and $l=1,\dots, r$. Similarly to the matrix model of Section~\ref{sec: tetra_matrix_model}, as well as that of~\cite{Szabo:2023ixw}, these are the fixed points of the orbifold ADHM data \smash{$(B_a^s,I_A^s)_{a\,\in\,\ulfour,A\,\in\,\ulfour^\perp\,,\,  s\in\Abw}$} under the equivariant action of the symmetry group $\sU(\vec k\,)\times\sU(\vec{\mbf r}\,)\times\sH_{\mbf r}$. They reside on the locus of fixed points of the BRST charge $\CQ_\Ab$ of the cohomological gauge theory on $[\FC^4/\Ab]$, and correspond to \smash{$\Ab$}-coloured plane partitions $\vec{\mbf\pi}$ as defined in Section~\ref{sec:orb_3d_pf}.

In the notation of Section~\ref{sec:orb_3d_pf}, the full partition function for orbifold tetrahedron instantons is
\begin{align}
Z^{\vec{\mbf r}}_{[\FC^4/\,\Ab]}(\vec\qu\,;\vec{\tta},\vec{\epsilon}\,) = \sum_{\vec k\in \RZ_{\geq 0}^{\#\Ab}} \, \vec\qu^{\,\vec k} \ Z^{\vec{\mbf r},\vec k}_{[\FC^4/\,\Ab]}(\vec{\tta},\vec{\epsilon}\,) \ ,
\end{align}
where
\begin{align}
\vec \qu^{\,\vec k}=\prod_{s\in\Abw}\,\qu_{s}^{k_{s}} \ .
\end{align}

\begin{remark}[{\bf Broken Permutation Symmetry}]
From the matrix model representation \eqref{eq:abelian_matrix_model} we deduce that, in contrast to the cases of Section~\ref{sec:tetra_insta}, the  partition functions for tetrahedron instantons on abelian orbifolds $\FC^4/\,\Ab$ are not  invariant under all permutations of the entries of the dimension vectors \smash{$\vec{\mbf r}=(\vec r_A)_{A\,\in\,\ulfour^\perp}$}. In fact, such a permutation generically generates a permutation of the Coulomb moduli \smash{$\vec\tta = (\vec{\tt a}_A^{\,s})_{A\,\in\,\ulfour^\perp\,,\,s\,\in\,\Abw}$} which is associated to different irreducible representations of $\Ab$ for different faces \smash{$A\in\ulfour^\perp$}.
\end{remark}

\subsubsection*{Dimensional Reduction}

Similarly to Section~\ref{sec: tetra_matrix_model}, we can compare \eqref{eq:abelian_matrix_model} with the matrix integral \smash{$\mbf\CZ_{[\FC^4/\,\Ab]}^{\vec r,\vec k}(\vec a,\vec\epsilon,\vec m)$} for the rank $r$ orbifold Donaldson--Thomas invariants of $[\FC^4/\,\Ab]$ of type $\vec r$ which was obtained in~\cite[eq.~(3.62)]{Szabo:2023ixw}. In particular, the analogue of Proposition~\ref{prop:Tetra_reduction} is given in the following way by restricting to solutions with at most two intersecting stacks of D7-branes. Prior to gauging, these types of tetrahedron instantons generalize the folded instantons of~\cite{Nekrasov:2016qym} and are called \emph{generalized folded instantons} by~\cite{Pomoni:2021hkn}.

For distinct fixed $A_1,A_2\in\ulfour^\perp$, let \smash{$\FC^2_{A_1,A_2}:=\FC^3_{A_1}\cap\FC^3_{A_2}$} denote the intersection of the corresponding codimension one strata in \smash{$\FC_\triangle^3\subset\FC^4$}; we write \smash{$\FC^2 = \FC^4\setminus \FC^2_{A_1,A_2} = \FC_{\bar A_1}\times\FC_{\bar A_2}$} for the remaining affine plane. We take as rank vector $\mbf r = \mbf r_{A_1,A_2} := (r_{A_1},r_{A_2},0,0)$; then the unbroken holonomy group is
\begin{align}
\sH_{\mbf r_{A_1,A_2}}=\sU(2)_{A_1,A_2}\,\times\,\sU(1) \ .
\end{align} 
We restrict the holonomy to $\sSU(2)_{A_1,A_2}\subset\sH_{\mbf r_{A_1,A_2}}$. Its only finite abelian subgroups are the cyclic groups $\Ab=\RZ_n$ of order $n$, with generator
\begin{align} \label{eq:ZnSU2generator}
g = {\small \begin{pmatrix} \xi_n & 0 \\ 0 & \xi_n^{-1} \end{pmatrix} } \normalsize \ ,
\end{align}
where $\xi_n=\e^{\,2\pi\,\ii/n}$ is a primitive $n$-th root of unity. If $A_1\cap A_2=(a_1\,a_2)$ with $a_1,a_2\in\ulfour\,$, then the weights of the fundamental representation $Q_4$ are $s_{a_1}=1$, $s_{a_2}=n-1$ and $s_{\bar A_1}=s_{\bar A_2}=0$. 

The McKay quiver assumes the form~\cite{Szabo:2023ixw}
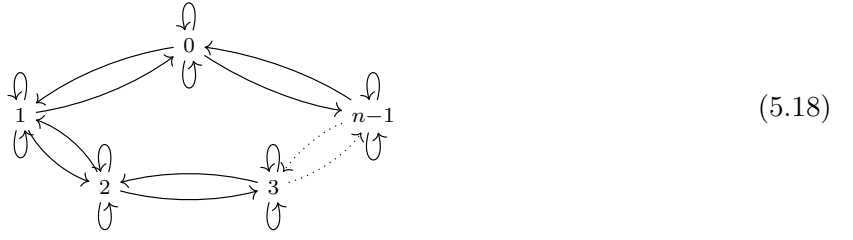
\begin{equation}
{\scriptsize
\begin{tikzcd}
	&& 0\arrow[,out=75,in=105,loop,swap,]    
	\arrow[,out=255,in=285,loop,swap,] \\
	1\arrow[,out=75,in=105,loop,swap,]    
	\arrow[,out=255,in=285,loop,swap,] &&&& {n{-}1}\arrow[,out=75,in=105,loop,swap,]    
	\arrow[,out=255,in=285,loop,swap,] \\
	& 2\arrow[,out=75,in=105,loop,swap,]    
	\arrow[,out=255,in=285,loop,swap,] && 3\arrow[,out=75,in=105,loop,swap,]    
	\arrow[,out=255,in=285,loop,swap,]
	\arrow[curve={height=6pt}, from=1-3, to=2-1]
	\arrow[curve={height=6pt}, from=2-1, to=1-3]
	\arrow[curve={height=6pt}, from=2-1, to=3-2]
	\arrow[curve={height=6pt}, from=3-2, to=2-1]
	\arrow[curve={height=6pt}, from=3-2, to=3-4]
	\arrow[curve={height=6pt}, from=3-4, to=3-2]
	\arrow[curve={height=6pt}, dotted, from=3-4, to=2-5]
	\arrow[curve={height=6pt}, dotted, from=2-5, to=3-4]
	\arrow[curve={height=6pt}, from=2-5, to=1-3]
	\arrow[curve={height=6pt}, from=1-3, to=2-5]
\end{tikzcd} }
\normalsize
\end{equation}
This is obtained from any orientation for the affine Dynkin diagram of type $\sA_{n-1}$, and adding a pair of edge loops at each node (cf. Example~\ref{ex:gammaSU(2)}).

\begin{proposition}\label{prop:Tetra_reduction_orb}
There exist Coulomb parameter and mass specialisations such that the equivariant instanton partition functions for the cohomological field theory with a massive fundamental hypermultiplet  on $[\FC^2_{A_1,A_2}/\,\RZ_n]\times\FC^2$  of type $\vec r=\vec r_{A_1}+\vec r_{A_2}$ and the cohomological field theory for orbifold tetrahedron instantons of type
\begin{align}
\vec{\mbf r}_{A_1,A_2} = \big(r_{A_1,s}\,,\,r_{A_2,s}\,,\,0\,,\,\dots\,,\,0\big)_{s=0,1,\dots,n-1}
\end{align}
are related as
\begin{align}
\mbf\CZ_{[\FC^2_{A_1,A_2}/\,\RZ_n]\times\FC^2}^{\vec r_{A_1}+\vec r_{A_2}}(\vec\qu\,;\vec a,\vec\epsilon,\vec m) = Z_{[\FC^2_{A_1,A_2}/\,\RZ_n]\times\FC^2}^{\vec{\mbf r}_{A_1,A_2}}(\vec\qu^{\,\prime}\,;\vec\tta,\vec\epsilon\,) \ ,
\end{align}
where $\vec\qu^{\,\prime} = \big((-1)^{r_{A_1,s}+r_{A_2,s}}\,\qu_s\big)_{s=0,1,\dots,n-1}$.
\end{proposition}

\proof
In \smash{$\mbf\CZ_{[\FC^2_{A_1,A_2}/\,\RZ_n]\times\FC^2}^{\vec r_{A_1}+\vec r_{A_2},\vec k}(\vec a,\vec\epsilon,\vec m)$} we specialise the substitution \eqref{eq:sigma_function} to
\begin{align}
\begin{split}
(a_l^s,m_l^s) = \begin{cases} \, \big(\tta_{A_1,l}^s\,,\,\tta_{A_1,l}^s+\epsilon_{\bar A_1}\big) \qquad \mbox{for} \quad l=1,\dots,r_{A_1} \ ,  \\[4pt]  \, \big(\tta_{A_2,l}^s\,,\,\tta_{A_2,l}^s+\epsilon_{\bar A_2}\big) \qquad \mbox{for} \quad l=r_{A_1}+1,\dots,r_{A_1}+r_{A_2} \ . \end{cases}
\end{split}
\end{align}
Using again $\epsilon_4 = -\epsilon_{123}$, together with \smash{$s_{\bar A_1}=s_{\bar A_2}=0$}, the matrix integral from~\cite[eq.~(3.62)]{Szabo:2023ixw} then coincides with the matrix integral \eqref{eq:abelian_matrix_model}, up to a sign factor $\prod_{s=0}^{n-1}\,(-1)^{(r_{A_1,s}+r_{A_2,s})\,k_s}$:
\begin{align}
\mbf\CZ_{[\FC^2_{A_1,A_2}/\,\RZ_n]\times\FC^2}^{\vec r_{A_1}+\vec r_{A_2},\vec k}(\vec a,\vec\epsilon,\vec m) = \e^{\,\sum\limits_{s=0}^{n-1}\,(r_{A_1,s}+r_{A_2,s})\,k_s} \ Z_{[\FC^2_{A_1,A_2}/\,\RZ_n]\times\FC^2}^{\vec{\mbf r}_{A_1,A_2},\vec k}(\vec\tta,\vec\epsilon\,) \ .
\end{align}
The result now follows by taking the weighted sum over $\vec k\in\RZ_{\geq0}^n$.
\endproof

\begin{remark}[{\bf Instantons on $\mbf{[\FC_A^3/\,\Ab]}$}] \label{rmk:reduction}
For fixed $A=(a\,b\,c)\in\ulfour^\perp$ with rank vector taken to be $\mbf r=\mbf r_A:=(r_A,0,0,0)$, the unbroken holonomy group is
\begin{align}
\sH_{\mbf r_A} = \sU(3)_A \ .
\end{align}
Proposition~\ref{prop:Tetra_reduction_orb} can then be extended to orbifolds $\FC_A^3/\,\Ab\times\FC$, where $\Ab$ is a finite subgroup of $\sSU(3)_A\subset\sH_{\mbf r_A}$ and \smash{$\FC=\FC^4\setminus\FC_A^3=\FC_{\bar A}$}. This recovers the partition function for orbifold tetrahedron instantons of type \smash{$\vec{\mbf  r}_A=(r_{A\,s},0,\dots,0)_{s\in\Abw}$}, which by~\cite[Proposition~3.63]{Szabo:2023ixw} reduces to the generating function for noncommutative Donaldson--Thomas invariants of type $\vec r_A$ on the toric Calabi--Yau three-orbifold $\FC_A^3/\,\Ab$ with $\sU(3)_A$ holonomy: 
\begin{align}
\mbf\CZ_{[\FC^3_{A}/\,\Ab]\times\FC}^{\vec r_A}(\vec\qu\,;\vec a,\vec\epsilon,\vec m) = Z_{[\FC^3_{A}/\,\Ab]\times\FC}^{\vec{\mbf r}_{A}}(\vec\qu^{\,\prime}\,;\vec\tta,\vec\epsilon\,) = Z^{\vec r_A}_{[\FC_A^3/\,\Ab]}(\vec\qu^{\,\prime}\,;\vec a,\epsilon_a,\epsilon_b,\epsilon_c) \ ,
\end{align}
where $\vec\qu^{\,\prime} = \big((-1)^{r_{A\,s}}\,\qu_s\big)_{s\in\Abw}$.
The restriction of this orbifold theory to $\sSU(3)_A$ holonomy is thoroughly discussed in \cite{Cirafici:2010bd}.
\end{remark}

\subsubsection*{Instanton Partition Functions}

The partition function for orbifold tetrahedron instantons can be evaluated by 
considering the $\Ab$-invariant part of the index \eqref{eq:tetra_chi}.
The natural inclusion $\Ab\lhook\joinrel\longrightarrow\sT_{\vec\epsilon}$  defines the irreducible representations of $\Ab$ associated to the toric generators $t_a$ for $a\in\ulfour\, $. Consequently, after a gauge transformation the vector spaces $V$ and $W_A$  at the fixed point \smash{$\vec{\mbf\pi}\in \frM_{\vec{\mbf r},\vec k}^\sT$} decompose into
\begin{align}\begin{split}   
  V_{\vec{\mbf\pi}}=\sum_{A=(a\,b\,c)\,\in\,\ulfour^\perp} \ \sum_{l=1}^{r_A}\, e_{A\,l} \ \sum_{\vec p_A\in {\pi}_{A\,l}} \, t_a^{p_a-1}\,t_b^{p_b-1}\,t_c^{p_c-1}\otimes \rho^*_{l;\vec p_A}
   \end{split} \label{decompositionV}
\end{align}
and
\begin{align}\begin{split}   
 W_{A{\vec{\mbf\pi}}}=\sum_{l=1}^{r_A}\, e_{A\,l}\otimes \rho^*_{\tts_A(l)}
   \end{split} \label{decompositionW}
\end{align}
as elements of the representation ring of $\sT\times\Ab$, where 
\begin{align}\label{eq:rho}
 \rho_{l;\vec p_{A}}:=\rho_{\tts_{A}(l)}\otimes\rho_{s_a}^{\otimes (p_a-1)}\otimes\rho_{s_b}^{\otimes (p_b-1)}\otimes\rho_{s_c}^{\otimes (p_c-1)} \qquad \mbox{for}  \quad A=(a\,b\,c) \ .
\end{align}

Proceeding as in Section~\ref{sec:orb_3d_pf}, we obtain the combinatorial formula
\begin{align}\label{eq:Orb_tetra_pf}
Z^{\vec{\mbf r}}_{[\FC^4/\,\Ab]}(\vec\qu\,;\vec{\tta},\vec{\epsilon}\,)=\sum_{\vec k\in\RZ_{\geq 0}^{\#\Ab}}\,\vec \qu^{\,\vec k} \ \sum_{\vec{\mbf \pi}\in\frM_{\vec{\mbf r},\vec k}^\sT} \, (-1)^{\texttt{O}^\Ab_{\vec{\mbf\pi}}} \ \widehat{\texttt e}\Big[-\sqrt{\ch^\Ab_{\sT}}(T_{\vec{\mbf\pi}}^{\rm vir} {\frM}_{\vec{\mbf r},\vec k}) \Big]    \ , 
\end{align}
where the superscript $^\Ab$ stands for the $\Ab$-invariant part and 
\begin{align}
\begin{split}
 \widehat{\texttt e}\Big[-\sqrt{\ch^\Ab_{\sT}}(T_{\vec{\mbf\pi}}^{\rm vir} {\frM}_{\vec{\mbf r},\vec k}) \Big] &= \prod_{A,A'\,\in\,\ulfour^\perp} \  \prod_{l=1}^{r_A} \ \prod_{\vec p_{A\,l}\in{\pi}_{A\,l}}^{\neq0} \, \frac{P_{r_{A'}}\circ\delta^\Ab_0(-\tta_{A\,l}-\vec p_{A\,l}\cdot\vec\epsilon+\epsilon_{\bar A'}\,|-\vec \tta_{A'})}{P_{r_A}\circ\delta^\Ab_0(\tta_{A\,l}+\vec p_{A\,l}\cdot\vec\epsilon\,|\vec \tta_A)} \\
& \quad\, \times\prod_{ l'=1}^{r_{A'}} \ \prod_{\vec p^{\,\prime}_{A'\,l'}\in{\pi}_{A'\,l'}}^{\neq0} \, R_-\circ\delta^\Ab_0\big(\tta_{A\,l}-\tta_{A'\,l'}+(\vec p_{A\,l}-\vec p^{\,\prime}_{A'\,l'})\cdot\vec\epsilon\,\big|\vec\epsilon\,\big) \ .
\end{split}\label{chiorb}
\end{align}

As in Section~\ref{sec:pf_tetra}, we may also introduce refined fugacities \smash{$\vec{\textbf{q}} = (\qu_{A\,s})_{A\,\in\,\ulfour^\perp\,,\,s\in\Abw}$} and reorganise the dimension vector $\vec k$ as \smash{$\vec{\mbf k} = (\vec k_{A})_{A\,\in\,\ulfour^\perp} = (k_{A\,s})_{A\,\in\,\ulfour^\perp\,,\,s\in\Abw}$}, where $k_{A\,s}$ is the total number of boxes of $\vec\pi_A$ of colour $s$, or equivalently the complex dimension of the isotypical component of the vector space $V_A$ from Remark~\ref{rem:tetrastability} labelled by $s\in\Abw$. A refined partition function, enumerating fractional instantons on each of the strata \smash{$\FC_A^3\subset\FC^3_\triangle$} for $A\in\ulfour\,$, may then be defined by replacing the counting weights \smash{$\vec\qu^{\,\vec k}$} in \eqref{eq:Orb_tetra_pf} with 
\begin{align}\label{eq:redefine_orb}
\vec{\textbf{q}}^{\,\vec{\mbf k}}:=\prod_{A\,\in\,\ulfour^\perp} \ \prod_{s\in\Abw}\,\qu_{A\,s}^{k_{A\,s}} \ ,
\end{align}
and writing
\begin{align}\label{eq:Orb_tetra_pf_refined}
Z^{\vec{\mbf r}}_{[\FC^4/\,\Ab]}(\vec{\textbf{q}}\,;\vec{\tta},\vec{\epsilon}\,)=\sum_{\vec{\mbf k}\in\RZ_{\geq 0}^{4\,\#\Ab}}\,\vec{\textbf{q}}^{\,\vec{\mbf k}} \ \sum_{\vec{\mbf \pi}\in\frM_{\vec{\mbf r},\vec{\mbf k}}^\sT} \, (-1)^{\texttt{O}^\Ab_{\vec{\mbf\pi}}} \ \widehat{\texttt e}\Big[-\sqrt{\ch^\Ab_{\sT}}(T_{\vec{\mbf\pi}}^{\rm vir} {\frM}_{\vec{\mbf r},\vec{\mbf k}}) \Big]    \ .
\end{align}

\begin{remark}[{\bf Sign Factors}]
Since the character \smash{$\sqrt{\ch^\Ab_{\sT}}(T_{\vec{\mbf\pi}}^{\rm vir}{\frM}_{\vec{\mbf r},\vec k})$} is evaluated by  projecting onto the $\Ab$-invariant part of the character~\eqref{eq:tetra_chi}, it seems reasonable to assume that the sign factor does not depend on the $\Ab$-coloring, i.e.  \smash{$\texttt{O}^{\Ab}_{\vec{\mbf\pi}}=\texttt{O}_{\vec{\mbf\pi}}$} is also given by \eqref{eq:signexplicit}. This is the same assertion made in \cite{Cao:2023gvn,Szabo:2023ixw} for instantons on toric Calabi--Yau four-orbifolds. 
\end{remark}

\begin{remark}[{\bf Permutation Symmetry}]\label{rmk:permutation}
Looking at the combinatorial formula \eqref{eq:Orb_tetra_pf}, it is easy to see that given a framing vector \smash{$\mbf{r}_s=(r_{A\,s},0,\dots,0)_{A\,\in\,\ulfour^\perp}$} for a fixed weight $s\in\Abw$, any other framing vector obtained from $\mbf{r}_s$ by varying $s\in\Abw$ yields an equivalent partition function. The same result naturally descends from Proposition \ref{prop:Tetra_reduction_orb} for the orbifolds of type $\mathbbm{C}^2/\RZ_n\times\mathbbm{C}^2$. Indeed, the partition function of type $\vec r$ for the cohomological field theory with a massive fundamental hypermultiplet on $[\mathbbm{C}^4/\,\Ab]$ is invariant under permutations of the entries of the dimension vector $\vec{r}=(r,0,\dots,0)$ \cite[Remark~4.15]{Szabo:2023ixw}.
\end{remark}

\begin{example}\label{ex:C2Z2C2}
Consider the rank two cohomological gauge theory on the orbifold resolution $[\mathbbm{C}^2/\mathbbm{Z}_2]\times\mathbbm{C}^2$ where $\mathbbm{Z}_2$ acts on $\mathbbm{C}^4$ with weights
\begin{align}
s_1=s_2=1 \qquad \mbox{and} \qquad s_3=s_4=0 \ .
\end{align} 
For the framing vector $\vec{\mbf r}$ we take
\begin{align}
    \mbf r_0 = (r_{123\,0},r_{124\,0},0,0,0,0,0,0)=(1,1,0,0,0,0,0,0) \ .
\end{align}
By Remark \ref{rmk:permutation}, this framing yields the same theory as the framing  
\begin{align}
\mbf r_1 = (0,0,r_{123\,1},r_{124\,1},0,0,0,0)=(0,0,1,1,0,0,0,0) \ .
\end{align}
The leading terms of the instanton partition function \eqref{eq:Orb_tetra_pf} are given by
\begin{align}\begin{split}\label{eq:example_Z2_NB}
 Z^{\mbf r_0}_{[\FC^2/\RZ_2]\times \FC^2}(\vec \qu\,;\vec\tta,\vec{\epsilon}\,) &= \frac{\epsilon_{12}\,\epsilon_{34}}{\epsilon_3\,\epsilon_4} \ \qu_0 +\frac{\epsilon_{12}\,\epsilon_{34}\, (\epsilon_{12}\,\epsilon_{34}-\epsilon_3\,\epsilon_4)}{2\,\epsilon_3^2\,\epsilon_4^2} \ \qu_0^2 \\
 & \qquad\,-\frac{\epsilon_{12}\,\epsilon_{34}\,(4\,\epsilon_1\,\epsilon_2-\epsilon_3\, \epsilon_4)}{2\,\epsilon_1\,\epsilon_2\,\epsilon_3\,\epsilon_4} \ \qu_0\,\qu_1+\cdots \ , \end{split}
\end{align}
independently of the Coulomb moduli $\vec\tta$.

On the other hand, after substituting \eqref{eq:redefine_orb}, the refined partition function \eqref{eq:Orb_tetra_pf_refined} takes the more complicated form
\begin{align}\begin{split}\label{eq:example_Z2}
    Z^{\mbf r_0}_{[\FC^2/\RZ_2]\times \FC^2}(\vec{\textbf{q}}\,;\vec\tta,\vec{\epsilon}\,) &= \frac{\epsilon_{12}\,(\tta -\epsilon_3)}{\epsilon_3\,\tta } \ \qu_{123\,0} +\frac{\epsilon_{12}\,(\tta  +\epsilon_4)}{\epsilon_4\,\tta } \ \qu_{124\,0} \\
    & \qquad \, + \frac{\epsilon_{12}\,(\epsilon_{12}-\epsilon_3)\,(\tta -\epsilon_3)}{2\,\epsilon_3^2\,(\tta +\epsilon_3)} \ \qu_{123\,0}^2 + \frac{\epsilon_{12}\,(\epsilon_{12}-\epsilon_4)(\tta  +\epsilon_4)}{2\,\epsilon_4^2\,(\tta -\epsilon_4)} \ \qu_{124\,0}^2 \\
    & \qquad \,  + \frac{\epsilon_{12}\,(4\,\epsilon_1\,\epsilon_2-\epsilon_3\,\epsilon_4)}{2\,\epsilon_1\,\epsilon_2\,\tta } \, \left(\frac{\epsilon_3
-\tta }{\epsilon_3} \ \qu_{123\,0}\,\qu_{123\,1} -\frac{\epsilon_4+\tta }{\epsilon_4} \ \qu_{124\,0}\,\qu_{124\,1} \right) \\
& \qquad \, + \frac{\epsilon_{12}^2}{\epsilon_3\,\epsilon_4}\,\frac{(\tta -\epsilon_{12})\,(\tta +\epsilon_{12})}{(\tta -\epsilon_4)\,(\tta +\epsilon_3)} \ \qu_{123\,0}\,\qu_{124\,0}   + \cdots \ . \end{split}
\end{align}
In particular, it depends explicitly  on the Coulomb moduli $\vec\tta = (\tta_1,\tta_2)$ through the combination  $\tta =\tta_1-\tta_2$.
\end{example}

\subsection{The Orbifolds $\FC_{A_1,A_2}^2/\RZ_n\times\FC^2$ and $\FC_A^3/(\RZ_2\times\RZ_2)\times\FC$}
\label{sub:age1orbifolds}

The finite subgroups of $\sSU(2)\subset\sSU(3)$ and $\sSO(3)\subset\sSU(3)$ play a special role in the Donaldson--Thomas theory of Calabi--Yau orbifolds~\cite{Young:2008hn,Cao:2023gvn}: these are the only orbifold groups whose elements all have age $\leq1$ and for which the theory can be subjected to a crepant resolution correspondence; we will return to this point in Section~\ref{subsec:crepant}. Of these the only abelian groups $\Ab$ are the cyclic groups $\RZ_n$ and the Klein four-group $\RZ_2\times\RZ_2$. 
Proposition \ref{prop:Tetra_reduction_orb} and Remark~\ref{rmk:reduction}  allow us to compute the orbifold partition functions for tetrahedron instantons based on the partition functions for the cohomological gauge theory with a massive fundamental hypermultiplet on $[\FC^4/\,\Ab]$. Utilizing the explicit results for the latter presented in~\cite[Sections~4~and~5]{Szabo:2023ixw}, we can immediately infer corresponding closed formulas for the unrefined partition functions for tetrahedron instantons on the orbifolds  $\FC_{A_1,A_2}^2/\RZ_n\times\FC^2$ and $\FC_A^3/(\RZ_2\times\RZ_2)\times\FC$.

\subsubsection*{Tetrahedron Instantons on $\mbf{\FC_{A_1,A_2}^2/\RZ_n\times\FC^2}$}

Consider the quotient singularity $\FC_{A_1,A_2}^2/\RZ_n\times \FC^2$ for distinct $A_1,A_2\in\ulfour^\perp$ in the notation of Proposition~\ref{prop:Tetra_reduction_orb}. Again we write $A_1\cap A_2=(a_1\,a_2)$ with $a_1,a_2\in\ulfour\,$.

\begin{proposition} \label{prop:C2ZnU4}
Assume \cite[Conjecture~4.11]{Szabo:2023ixw} is true. Then the unrefined partition function for  tetrahedron instantons of type~$\vec{\mbf r}_{A_1,A_2\,0}=(r_{A_1\,0},r_{A_2\,0},0,\dots,0)$ on the orbifold  $\FC_{A_1,A_2}^2/\RZ_n\times\FC^2$ with $\sH_{\mbf r_{A_1,A_2}}$ holonomy is given by
\begin{align}\label{eq:C2ZnU4}
\begin{split}
& Z^{\vec{\mbf r}_{A_1,A_2\,0}}_{[\FC_{A_1,A_2}^2/\RZ_n]\times\FC^2}(\vec\qu\,;\vec \epsilon\, ) \\[4pt]
& \hspace{2cm} = M\big((-1)^{n+r}\,\Qu\big)^{-n\,\frac{\epsilon_{12}\,\epsilon_{13}\,\epsilon_{23}\,(r_{A_1\,0}\,\epsilon_{\bar A_1}+r_{A_2\,0}\,\epsilon_{\bar A_2})}{\epsilon_1\,\epsilon_2\,\epsilon_3\epsilon_4}-\frac{n^2-1}{n}\,\frac{\epsilon_{a_1a_2}\,(r_{A_1\,0}\,\epsilon_{\bar A_1}+r_{A_2\,0}\,\epsilon_{\bar A_2})}{\epsilon_{a_1}\,\epsilon_{a_2}}} \\
& \hspace{4cm} \times \prod_{0<p\leq s<n}\,\widetilde{M}\big((-1)^{p-s+1}\,\qu_{[p,s]},(-1)^{n+r}\,\Qu\big)^{\frac{\epsilon_{a_1a_2}\,(r_{A_1\,0}\,\epsilon_{\bar A_1}+r_{A_2\,0}\,\epsilon_{\bar A_2})}{\epsilon_{\bar A_1}\,\epsilon_{\bar A_2}}} \ ,\end{split}
\end{align}
where
\begin{align}
\Qu=\qu_0\,\qu_1\cdots \qu_{n-1} \qquad \mbox{and} \qquad \qu_{[p,s]}=\qu_p\,\qu_{p+1}\cdots \qu_{s-1}\,\qu_s \ ,
\end{align}
while $r=r_{A_1\,0}+r_{A_2\,0}$.
\end{proposition}

\proof
This follows straightforwardly from \cite[Conjecture~4.11]{Szabo:2023ixw} and Proposition~\ref{prop:Tetra_reduction_orb}.
\endproof

\begin{remark} [{\bf Refined Partition Functions}]
The expansion of the formula \eqref{eq:C2ZnU4} for $n=2$, $A_1=(123)$, $A_2=(124)$ and $r_{A_1\,0}=r_{A_2\,0}=1$ reproduces the explicit expansion \eqref{eq:example_Z2_NB} from Example~\ref{ex:C2Z2C2}. On the other hand, it is not clear if the refined partition functions can also be expressed by a closed formula, see \eqref{eq:example_Z2}.
\end{remark}

\begin{proposition}\label{prop:r0r1_tetra}
Assume \cite[Conjecture~4.21]{Szabo:2023ixw} is true. Then the unrefined partition function for tetrahedron instantons  of  type $\vec{\mbf r}_{A_1,A_2}=(r_{A_1\,0},r_{A_2\,0},r_{A_1\,1},r_{A_2\,1},0,0,0,0)$ on the orbifold  $\FC_{A_1,A_2}^2/\RZ_2\times\FC^2$ with $\sH_{\mbf r_{A_1,A_2}}$ holonomy is given by
\begin{align}\label{eq:C2Z2r0r1_tetra}
\begin{split}
Z_{[\FC_{A_1,A_2}^2/\RZ_2]\times\FC^2}^{\vec{\mbf r}_{A_1,A_2}}(\vec \qu\,;\vec{\epsilon}\,) &= M\big((-1)^r\,\qu_0\,\qu_1\big)^{-2\,\frac{\epsilon_{12}\,\epsilon_{13}\,\epsilon_{23}\,(r_{A_1}\,\epsilon_{\bar A_1}+r_{A_2}\,\epsilon_{\bar A_2})}{\epsilon_1\,\epsilon_2\,\epsilon_3\,\epsilon_4}-\frac{3}{2}\,\frac{\epsilon_{a_1a_2}\,(r_{A_1}\,\epsilon_{\bar A_1}+r_{A_2}\,\epsilon_{\bar A_2})}{\epsilon_{a_1}\,\epsilon_{a_2}}} \\
 &\hspace{1cm} \times \,  \widetilde{M}(-\qu_1,(-1)^r\,\qu_0\,\qu_1)^{-\frac{\epsilon_{a_1a_2}\,(r_{A_1\,0}\,\epsilon_{\bar A_1}+r_{A_2\,0}\,\epsilon_{\bar A_2})}{\epsilon_{\bar A_1}\,\epsilon_{\bar A_2}}} \\
 &\hspace{2cm} \times \,  \widetilde{M}(-\qu_0,(-1)^r\,\qu_0\,\qu_1)^{-\frac{\epsilon_{a_1a_2}\,(r_{A_1\,1}\,\epsilon_{\bar A_1}+r_{A_2\,1}\,\epsilon_{\bar A_2})}{\epsilon_{\bar A_1}\,\epsilon_{\bar A_2}}} \ ,
\end{split}
\end{align}
where $r_A=r_{A\,0}+r_{A\,1}$ for $A\in\{A_1,A_2\}$ and $r=r_{A_1}+r_{A_2}$.
\end{proposition}

\proof
This follows straightforwardly from  \cite[Conjecture~4.21]{Szabo:2023ixw} and Proposition \ref{prop:Tetra_reduction_orb}.
\endproof

\subsubsection*{Tetrahedron Instantons on $\mbf{\FC_A^3/(\RZ_2\times\RZ_2)\times\FC}$}

Consider the action of $\RZ_2\times\RZ_2$ on $\FC_A^3$ for fixed $A=(a\,b\,c)\in\ulfour^\perp$ given as in Example \ref{ex:C3Z2Z2}. Using Remark~\ref{rmk:reduction} we can recover the unrefined partition function for tetrahedron instantons of type $\vec{\mbf r}_A=(r,0,\dots,0)$ on the orbifold \smash{$\FC_A^3/(\RZ_2\times\RZ_2)\times\FC$} with holonomy group $\sH_{\mbf r_A}=\sU(3)_A$. It coincides, up to signs after taking the $r$-th power, with the closed formula \eqref{eq:C3Z2Z2U3}:
\begin{align} \label{eq:tetraZ2Z2}
\begin{split}
& Z^{\vec{\mbf r}_A=(r,0\dots,0)}_{[\FC_A^3/\RZ_2\times\RZ_2]\times\FC}(\vec \qu\,;\vec \epsilon\,) \\[4pt]
& \hspace{0.5cm} = \frac{M\big((-1)^{r}\,\Qu\big)^{r\,\frac{\epsilon_a\,\epsilon_b\,\epsilon_c-\epsilon_a^2\,\epsilon_b-\epsilon^2_a\,\epsilon_c-\epsilon^2_b\,\epsilon_c-\epsilon_a\,\epsilon_b^2-\epsilon_a\,\epsilon_c^2-\epsilon_b\,\epsilon^2_c}{\epsilon_a\,\epsilon_b\,\epsilon_c}}}{\widetilde M\big(-\qu_1,(-1)^r\,\Qu\big)^r \, \widetilde M\big(-\qu_2,(-1)^r\,\Qu\big)^r \, \widetilde M\big(-\qu_3,(-1)^r\,\Qu\big)^r\,\widetilde M\big(-\qu_1\,\qu_2\,\qu_3,(-1)^r\,\Qu\big)^r} \\[2pt]
& \hspace{4cm} \times
 \,  \prod_{\stackrel{\scriptstyle p,s\in A}{\scriptstyle p<s}}\,\widetilde{M}\big(\qu_p\,\qu_s,(-1)^{r}\,\Qu\big)^{r\,\frac{\epsilon_{(ps)^-}-\epsilon_{ps}}{2\,\epsilon_{(ps)^-}}}  \ ,
 \end{split}
\end{align}
where $(ps)^-=A\setminus\{p,s\}$.

\subsection{Tetrahedron Instantons on Non-Abelian Orbifolds}\label{sec:treta_nonabelian}

We now turn to the case where the orbifold group $\sGamma\subset\sH_{\mbf r}\subset\sSU(4)$ is a finite non-abelian group. 
In analogy with the discussion of Section~\ref{sec:quiver_variety}, we can associate a McKay quiver $\ttQ^\sGamma=\big(\ttQ^\sGamma_0,\ttQ^\sGamma_1\big)$  to the action of a finite subgroup $\sGamma $ of the unbroken holonomy group \eqref{eq:sHmbfr} on $\FC^4$ in the fundamental representation $Q_4$ of $\sSU(4)$. Each vertex $\sfi\in\ttQ^\sGamma_0$ corresponds to an irreducible representation $\lambda_\sfi\in\widehat\sGamma$, while the number $a_{\sfi\sfi'}$ of arrows from vertex $\sfi$ to vertex $\sfi'$ is given by the decomposition of $\sGamma$-modules
\begin{align} \label{eq:Q4lambda}
Q_4\otimes\lambda_\sfi=\bigoplus_{\sfi'\in\ttQ^\sGamma_0}\,a_{\sfi\sfi'}\,\lambda_{\sfi'} = \bigoplus_{e\in\sfs^{-1}(\sfi)} \, \lambda_{\sft(e)} \ \oplus \ \bigoplus_{e\in\sft^{-1}(\sfi)} \, \lambda_{\sfs(e)}  \ .
\end{align}

The ADHM parametrization is constructed as a stable framed linear representation of the bounded McKay quiver.
To each vertex $\sfi\in\ttQ_0^\sGamma$ we assign vector spaces $V_\sfi$ and $W_{A\,\sfi}$, together with  linear maps $I_{A\,\sfi}\in\sHom_{\FC}(W_{A\,\sfi},V_\sfi)$ for $A\in\ulfour^\perp$. We introduce dimension vectors \smash{$\vec k =(k_\sfi)_{\sfi\in\ttQ_0^\sGamma}$} and \smash{$\vec{\mbf r} = (\vec r_{A})_{A\in\ulfour^\perp}=(r_{A\,\sfi})_{A\in\ulfour^\perp\,,\,\sfi\in\ttQ_0^\sGamma}$}, with  $k_\sfi=\dim V_\sfi $ and $r_{A\,\sfi}=\dim W_{A\,\sfi}$. With $d_\sfi = \dim\lambda_\sfi$, we define $r_A=\sum_{\sfi\in\ttQ^\sGamma_0}\,d_\sfi\, r_{A\,\sfi}$ for any $A\in\ulfour^\perp$, and set
\begin{align}
  k = |\vec k\,|:=\sum_{\sfi\in\ttQ^\sGamma_0}\, d_\sfi \, k_\sfi \qquad \mbox{and}\quad r=\sum_{A\,\in\,\ulfour^\perp} \, r_A=|\vec{\mbf r}\,| := \sum_{A\,\in\,\ulfour^\perp} \, |\vec r_A| = \sum_{A\,\in\,\ulfour^\perp} \  \sum_{\sfi\in\ttQ^\sGamma_0} \, d_\sfi\, r_{A\,\sfi}\ .
\end{align}
Finally, to each arrow $e\in\ttQ^\sGamma_1$ we assign a linear map $B_e\in\sHom_\FC(V_{\sfs(e)},V_{\sft(e)})$. The linear maps \smash{$(B_e,I_{A\,\sfi})_{\sfi\in\ttQ_0^\sGamma\,,\,e\in\ttQ_1^\sGamma\,,\,A\in\ulfour^\perp}$} are required to satisfy relations for the McKay quiver given by the orbifold ADHM equations, obtained as $\sGamma$-equivariant decomposition of the tetrahedron instanton equations \eqref{eq:ADHM_tetra}, similarly to Section~\ref{sec:quiver_variety}.

Similarly to Section \ref{sec:orb_3d_inst}, we can also consider arbitrary finite non-abelian groups $\sGamma$ and define their actions on $\FC^4$ via a homomorphism 
\begin{align}
\tau:\sGamma\longrightarrow \sH_{\mbf r} \ .
\end{align}
Generically, this leads  to non-effective orbifolds of $\FC^4$. In particular, for the choice of framing vector $\mbf r = {\mbf r}_A$, the cohomological gauge theory is BRST-localized on noncommutative instantons in the twisted orbifold resolution $[\FC_A^3/\,\sGamma]\times\rmB\sK^{\tau}$ of the quotient singularity $\FC_A^3/\,\sGamma^{\tau}$, with holonomy $\sH_{\mbf r_A}=\sU(3)_A$, that we discussed in Section~\ref{sec:orb_3d_inst}.
Generally, the constraint that the image $\tau(\sGamma)$ lands in the defect-preserving subgroup $\sH_{\mbf r}\subset\sSU(4)$ of the holonomy group ensures that the strata \smash{$\FC^3_A\subset\FC_\triangle^3$} for $A\in\ulfour^\perp$ are invariant under the $\sGamma$-action and restricts our considerations to only two admissible classes. We follow the terminology and notation of Section~\ref{sec:orb_3d_inst} throughout, and call these $\sSU(m)\,\times\,$abelian  orbifolds for $m=2,3$.

\subsubsection*{$\mbf{\sSU(2)\,\times\,{\rm Abelian}}$ Orbifolds}

For fixed distinct face labels $A_1,A_2\in\ulfour^\perp$, we write $A_1\cap A_2=(a_1\,a_2)$ for $a_1,a_2\in\ulfour\,$. We take as framing vector $\vec{\mbf r}_{A_1,A_2} = (\vec r_{A_1},\vec r_{A_2},\vec 0,\vec 0\,)$; then the unbroken holonomy group \eqref{eq:sHmbfr} is given by \smash{$\sH_{\mbf r_{A_1,A_2}} = \sU(2)_{A_1,A_2}\times\sU(1)$}.
Let $\sGamma_2= \sUps_2\times \sGamma_{\textsf{ab}}$, where $\sUps_2$ is a finite non-abelian  subgroup of $\sSU(2)$ acting on \smash{$\FC_{A_1,A_2}^2$} in the fundamental representation $Q_2$. 

Let $\sGamma_2$ act on $\FC^4$ via the homomorphism $\tau_{\vec s} :\, \sGamma_2\longrightarrow \sH_{\mbf r_{A_1,A_2}}$ defined by
\begin{align}\label{Gamma_action}
\tau_{\vec s}\,(\sGamma_2)= \big(
\sUps_2\times \rho_{s_1}(\Ab)\big) \, \times \, \rho_{-s_1+s_2}(\Ab) \, \times \, \rho_{-s_{12}}(\Ab) \ \subset \ \sU(2)_{A_1,A_2}\times\sU(1) \ \subset \ \sSU(4) \ ,
\end{align}
where $\vec s = (s_1,s_2)$ and $\rho_s: \,\Ab\longrightarrow \sU(1)$ is the unitary irreducible representation of $\Ab$ with weight $s$.  
This defines the action of $\sGamma_2$ on $\FC^4=\FC^2_{A_1,A_2}\times\FC_{\bar A_1}\times\FC_{\bar A_2}$ as the four-dimensional $\sGamma_2$-module
\begin{align}
Q_4^{\vec s} = (Q_2\otimes\rho_{s_1}) \, \oplus \, (\lambda_0\otimes\rho_{-s_1+s_2}) \, \oplus \, (\lambda_0\otimes\rho_{-s_{12}}) \ .
\end{align}
The kernel of $\tau_{\vec s}$ is the normal subgroup
\begin{align}
\sK^{\vec s} := \ker(\tau_{\vec s}) = \big\{ (g,\xi)\in \sUps_2\times\Ab \ \big| \ g=\rho_{-s_1}(\xi)\,\ident_2\in\sUps_2 \ , \ \xi\in\ker(\rho_{-s_1+s_2})\cap\ker(\rho_{s_{12}}) \big\} 
\end{align}
of $\sGamma_2$.

The centralizer of $\tau_{\vec s}\,(\sGamma_2)$ in $\sT_{\vec \epsilon}$ is
\begin{align}\label{eq:tori1}
\sC^{\vec s} =\sU(1)_{\vec\epsilon}^{\times 2} \ \subset \ \sT_{\vec\epsilon} \ ,
\end{align}
where $\vec\epsilon = (\epsilon_1,\epsilon_2)$ are the equivariant parameters. 
Thus the unbroken maximal torus of the equivariant gauge theory is
\begin{align}\label{eq:toriaction_tetra}
\sT_{A_1,A_2}:=\sT_{\vec{\mbf \tta}_{A_1}}\times \sT_{\vec{\mbf \tta}_{A_2}}\times\sU(1)^{\times 2}_{\vec\epsilon} \ 	.
\end{align}

Let \smash{$\mathsf{Dynk}_{\sUps_2}$} be the oriented affine Dynkin diagram associated to $\sUps_2$, with adjacency matrix \smash{$A_{\sUps_2}=\big(a^{\sUps_2}_{\sfi\sfi'}\big)$}. Each vertex of the McKay quiver \smash{$\ttQ^{\tau_{\vec s}\,(\sGamma_2)}$} is labelled by a pair $(\sfi,s)$, where $\sfi$ is a vertex of $\mathsf{Dynk}_{\sUps_2}$ and \smash{$s\in\widehat \sGamma_{\textsf{ab}}$}. Then the number of arrows $a_{(\sfi,s)\,(\sfi',s')}$ from vertex $(\sfi,s)$ to vertex $(\sfi',s')$ is given by
\begin{align}
a_{(\sfi,s)\,(\sfi',s')} = a^{\sUps_2}_{\sfi\sfi'} \ \delta_{s',s+s_1} + \delta_{\sfi',\sfi} \ \big(\delta_{s',s-s_1+s_2} + \delta_{s',s-s_{12}} \big) \ .
\end{align}

In the notation of Example~\ref{ex:gammaSU(2)}, with $r_{A\,\sfi,s}=0$ unless \smash{$A\in\{A_1, A_2\}$}, the ADHM variables $(B,I_A)\in\sHom_{\sGamma_2}(V,V\otimes Q_4^{\vec s\,})\times\sHom_{\sGamma_2}(W_A,V)$ decompose into linear maps
\begin{align}
\begin{split}
B^s_e \in\sHom_\FC(V_{\sfs(e),s}, V_{\sft(e),s+s_1}) \quad , \quad \bar B^s_e\in \sHom_\FC&(V_{\sft(e),s}, V_{\sfs(e),s+s_1}) \ , \\[4pt]
L^s_{\bar A_1\,\ii}\in\sHom_\FC(V_{\sfi,s}, V_{\sfi,s-s_{12}}) \ ,  \quad
L^s_{\bar A_2\,\ii}\in\sHom_\FC(V_{\sfi,s}, V_{\sfi,s-s_1+s_2})\ , &\quad I^s_{A\,\sfi}\in\sHom_\FC(W_{A\,\sfi,s}, V_{\sfi,s})
\end{split}
\end{align}
for each arrow $e$ and vertex $\sfi$ of \smash{$\mathsf{Dynk}_{\sUps_2}$}, $s\in\Abw$, and $A\in\{A_1,A_2\}$. 

The field content is required to satisfy the orbifold ADHM equations
\begin{align}\label{eq:ADHM1}
\begin{split}
\mu_{\sfi}^{\FC s}&=\sum_{e\,\in\, \sfs^{-1}(\sfi)} \, \bar B^{s+s_1}_{e}\,B^s_{e}- \sum_{e\,\in\, \sft^{-1}(\sfi)} \, B^{s+s_1}_{e}\,\bar B^s_{e} \\
& \hspace{1cm} + L_{\bar A_2\,\sfi}^{s+2s_1 \dagger}\, L_{\bar A_1\,\sfi}^{s+s_{12}\dagger}-L_{\bar A_1\,\sfi}^{s+2s_1\dagger}\,L_{\bar A_2\,\sfi}^{s+s_1-s_2\dagger} \ = \ 0 \  ,
\\[4pt]
\mu_{e}^{\FC s}&= L^{s+s_1}_{\bar A_2\,\sft(e)}\,B^s_{e}-  B_{e}^{s-s_1+s_2}\,L^s_{\bar A_2\,\sfs(e)} + \bar B_{e}^{s+s_2\dagger}\, L_{\bar A_1\,\sfs(e)}^{s+s_{12}\dagger}-L_{\bar A_1\,\sft(e)}^{s+s_2\dagger}\, \bar B_{e}^{s-s_1\dagger} \ =  \ 0 \ ,
\\[4pt]
\bar\mu_{e}^{\FC s}&= L^{s+s_1}_{\bar A_2\,\sfs(e)}\,\bar B^s_{e}-  \bar B^{s-s_1+s_2}_{e}\,L^s_{\bar A_2\,\sft(e)}-B_e^{s+s_{2}\dag}\,L_{\bar A_1\,\sft(e)}^{s+s_{12} \dagger}  + L_{\bar A_1\,\sfs(e)}^{s+s_2\dagger}\,B_e^{s-s_1\dagger} \ = \ 0 \ ,
\\[4pt]
\mu_{\sfi}^{\FR s}&=\sum_{e\,\in\, \sft^{-1}(\sfi)} \, \big(B^{s-s_1}_{e}\,B_{e}^{s-s_1\dagger}-\bar B_{e}^{s\dagger}\, \bar B^{s}_{e}\big) - \sum_{e\,\in\, \sfs^{-1}(\sfi)} \, \big( B^{s\dagger}_{e}\,B^s_{e} -\bar B^{s-s_1}_{e}\,\bar B_{e}^{s-s_1\dagger}\big)\\
&\hspace{1cm} +L_{\bar A_2\,\sfi}^{s+s_1-s_2}\, L_{\bar A_2\,\sfi}^{s+s_1-s_2\dagger} - L_{\bar A_2\,\sfi}^{s\dagger}\,L_{\bar A_2\,\sfi}^{s} +L_{\bar A_1\,\sfi}^{s+s_{12}}\, L_{\bar A_1\,\sfi}^{s+s_{12}\dagger}- L_{\bar A_1\,\sfi}^{s\dagger}\,L_{\bar A_1\,\sfi}^{s} \\
&\hspace{2cm} + I^s_{A_1\,\sfi}\,I_{A_1\,\sfi}^{s\dagger} + I^s_{A_2\,\sfi}\,I_{A_2\,\sfi}^{s\dagger} \ = \ \zeta_{\sfi,s} \, \ident_{ V_{\sfi,s}} \ ,  
\\[4pt]
\sigma_{A_1\,\sfi}^s&=I_{A_1\,\sfi}^{s+s_{12}\dagger}\, L_{\bar A_1\,\sfi}^{s\dagger} \ = \ 0 \quad , \quad \sigma_{A_2\,\sfi}^s=I_{A_2\,\sfi}^{s+s_{1}-s_2\dagger}\, L_{\bar A_2\,\sfi}^{s\dagger} \ = \ 0 \ ,
\end{split}
\end{align}
where $\zeta_{\sfi,s}\in\FR_{>0}$.

The action of the torus $\sC^{\vec s} = \sU(1)_{\vec\epsilon}^{\times 2} $ from \eqref{eq:toriaction_tetra} on the ADHM data is given by
\begin{align}\begin{split}\label{eq:adhm_tetra_trans}
\big(B\,,\,\bar B\,,\,L_{\bar A_1}\,,\,L_{\bar A_2}\,,\,I_{A_1}\,,\,I_{A_2}\big)\longmapsto \big(t_1^{-1} \, B\,,\, t_1^{-1}\, \bar B\,,\, t_1^2\,t_2\, L_{\bar A_1}\,,\,t_2^{-1}\,L_{\bar A_2}\,,\,I_{A_1}\,,\,I_{A_2}\big) \ ,
\end{split}
\end{align}
where $t_a=\e^{\,\ii\,\epsilon_a}$.

\subsubsection*{$\mbf{\sSU(3)\,\times}\,$Abelian Orbifolds}

For a fixed face label $A=(a\,b\,c)\in\ulfour^\perp$, we take as framing vector $\vec{ \mbf r}_A =(\vec r_{A},\vec 0,\vec 0,\vec 0\,)$; then the unbroken holonomy group \eqref{eq:sHmbfr} is $\sH_{\mbf r_A} = \sU(3)_A$. 
Let $\sGamma_3=\sUps_3\times\sGamma_{\textsf{ab}}$, where $\sUps_3$ is a finite non-abelian subgroup of $\sSU(3)$ acting on $\FC_A^3$ in the fundamental representation $Q_3$. 

Let $\sGamma_3$ act on $\FC^4$ via the homomorphism $\tau_{\tilde s}:\sGamma_3 \longrightarrow \sH_{\mbf r_A}$ defined by
\begin{align}
\tau_{\tilde s}\,(\sGamma_3) = 
\big(\sUps_3\times \rho_{\tilde s}(\Ab)\big)\,\times\,
\rho_{-3 \tilde s} (\Ab) \ \subset \ \sU(3)_A \ \subset \ \sSU(4) \ .
\end{align}
This defines the action of $\sGamma_3$ on $\FC^4=\FC_A^3\times\FC_{\bar A}$ as the four-dimensional $\sGamma_3$-module
\begin{align}
Q_4^{\tilde s} = (Q_3\otimes\rho_{\tilde s})\,\oplus\,(\lambda_0\otimes\rho_{-3\tilde s}) \ .
\end{align}
The kernel of $\tau_{\tilde s}$ is the normal subgroup
\begin{align}
\sK^{\tilde s} := \ker(\tau_{\tilde s}) = \big\{(g,\xi)\in\sUps_3\times\Ab \ \big| \ g=\rho_{-\tilde s}(\xi)\,\ident_3 \ , \ \xi\in\ker(\rho_{3\tilde s})\big\}
\end{align}
of $\sGamma_3$. 

The centralizer of $\tau_{\tilde s}(\sGamma_3)$ in $\sT_{\vec \epsilon}$ is
\begin{align}\label{eq:tori2}
\sC^{\tilde s} =\sU(1)_\epsilon \ \subset \ \sT_{\vec\epsilon} \ ,
\end{align}
where $\epsilon$ is the equivariant parameter, and the unbroken maximal torus of the equivariant gauge theory is
\begin{align}
\sT_A:=\sT_{\vec{ \mbf \tta}_{A}}\times\sU(1)_{\epsilon} \ .
\end{align}
The adjacency matrix of the McKay quiver \smash{$\ttQ^{\tau_{\tilde s}(\sGamma_3)}$} is given by
\begin{align}
a_{(\sfi,s)\,(\sfi',s')} = a^{\sUps_3}_{\sfi\sfi'} \ \delta_{s',s+\tilde s} + \delta_{\sfi,\sfi'} \ \delta_{s',s-3\tilde s} \ ,
\end{align}
where \smash{$(\sfi,s),(\sfi',s')\in \ttQ_0^{\sUps_3}\times\widehat \sGamma_{\textsf{ab}}$}.

The ADHM field content $(B,I_A)\in\sHom_{\sGamma_3}(V,V\otimes Q_4^{\tilde s})\times\sHom_{\sGamma_3}(W_A,V)$ decomposes into linear maps
\begin{align}
\begin{split}
B_e^s\in\sHom_\FC(V_{\sfs(e),s}, V_{\sft(e),s+\tilde s}) \quad , \quad L_{\bar A\,\sfi}^s\in\sHom_\FC(V_{\ii,s}, V_{\sfi,s-3\tilde s}) \quad , \quad I^s_{A\,\sfi}\in\sHom_\FC(W_{A\,\sfi,s},V_{\sfi,s}) \ ,
\end{split}
\end{align}  
for $e\in\ttQ_1^{\sUps_3}$, $s\in\Abw$, and $\sfi\in\ttQ_0^{\sUps_3}$.
They satisfy the orbifold ADHM equations
\begin{align}\begin{split}\label{eq:ADHM2}
 \mu^\FC&=B\wedge B -\star_{\varOmega} \,(B\wedge B) \ = \ 0 \ , \\[4pt]
\mu_{\sfi}^{\FR s}&=\sum_{e\,\in\, \sft^{-1}(\sfi)}\,B_{e}^{s-\tilde s}\,B_{e}^{s-\tilde s\dagger} - \sum_{e\,\in\, \sfs^{-1}(i)}B^{s\dagger}_{e}\,B^s_{e} \\
&\hspace{1cm} + L_{\bar A\,\sfi}^{s+3\tilde s}\,L_{\bar A\,\sfi}^{s+3\tilde s\dagger}-L^{s\dagger}_{\bar A\,\sfi}\,L_{\bar A\,\sfi}^s+ I^s_{A\,\sfi}\,I_{A\,\sfi}^{s\dagger} \ = \ \zeta_{\sfi,s} \, \ident_{V_{\sfi,s}}  \ , \\[4pt]
\sigma_{A\,\sfi,s}&=I_{A\,\sfi}^{s+3\tilde s\dagger}\, L_{\bar A\,\sfi}^{s\dagger} \ = \ 0 \ , \end{split}
\end{align}
where $\zeta_{\sfi,s}\in\FR_{>0}$, and the complex equation $\mu^\FC\in\sHom_{\sGamma_3}(V,V\otimes\midwedge_-^2Q_4^{\tilde s})$ is written using the involution $\star_\varOmega$ from \eqref{eq:starOmega}.

The torus $\sC^{\tilde s}=\sU(1)_\epsilon$ from \eqref{eq:tori2} transforms the ADHM data as
\begin{align}\begin{split}\label{eq:adhm_tetra_trans1}
\big(B\,,\,L_{\bar A}\,,\,I_A\big) \longmapsto \big(t^{-1}\,B \,,\, t^3\,L_{\bar A}  \,,\, I_{A}\big) \ ,
\end{split}
\end{align}
where $t=\e^{\,\ii\,\epsilon}$.

\subsubsection*{Stability and Quot Schemes}
\label{rmk:stability}

As discussed  in \cite{Nekrasov:2016ydq,Nakajima:1994nid} for the case of Nakajima quiver varieties, the D-term equations $\mu_{\sfi}^{\FR}=\zeta_\sfi\,\ident_{V_\sfi}$ in~\eqref{eq:ADHM1} and \eqref{eq:ADHM2}, for \smash{$\sfi\in \ttQ_0^{\tau(\sGamma_m)}$}, are equivalent to the following stability condition: if there is a collection of subspaces $S_{\sfi}\subset V_{\sfi}$ for $\sfi\in\ttQ_0^{\tau(\sGamma_m)}$ such that
\begin{align}
I_{A\,\sfi}(W_{A\,\sfi})\subset S_{\sfi} \qquad \text{and} \qquad
B_{e} (S_{\sfs(e)})\subset S_{\sft(e)} \ ,
\end{align}
for all $\sfi\in\ttQ_0^{\tau(\sGamma_m)}$, $A\in\ulfour^\perp$ and $e\in\ttQ_1^{\tau(\sGamma_m)}$,
then $S_\sfi=V_\sfi$ for all $\sfi\in\ttQ_0^{\tau(\sGamma_m)}$. In the present case, the proof is similar to the stability proof for spiked instantons given in~\cite[Section~8]{Nekrasov:2015wsu} .

Similarly to~\cite{Nekrasov:2016ydq}, let $\CP_\sfi^\sfj[\ttQ^{\tau(\sGamma_m)}]$ denote the set of all paths along the McKay quiver $\ttQ^{\tau(\sGamma_m)}$ starting at vertex \smash{$\sfi\in\ttQ_0^{\tau(\sGamma_m)}$} and ending at vertex \smash{$\sfj\in\ttQ_0^{\tau(\sGamma_m)}$}.  A path \smash{$\gamma=(e_{\gamma_1},\dots, e_{\gamma_n})\in\CP_\sfi^\sfj[\ttQ^{\tau(\sGamma_m)}]$} of length $\ttl(\gamma)=n$ is described by a sequence of $n$ arrows $e_{\gamma_i}\in\ttQ_1^{\tau(\sGamma_m)}$, with $\sfs(e_{\gamma_1})=\sfi$, $\sft(e_{\gamma_n})=\sfj$ and $\sfs(e_{\gamma_i}) = \sft(e_{\gamma_{i-1}})$ for $2\leq i\leq n$. We indicate by $\CB_\gamma$ the composition of linear maps $B_e$ defined by the path $\gamma$:
\begin{align}
\CB_{\gamma}=B_{e_{\gamma_{n}}}\,B_{e_{\gamma_{n-1}}}\,\cdots\,B_{e_{\gamma_1}} \ .
\end{align}
Then the stability condition implies that
\begin{align}
V_\sfj= \sum_{A\,\in\,\ulfour^\perp}\,V_{A\,\sfj}:=\sum_{A\,\in\,\ulfour^\perp} \ \sum_{\sfi\in\ttQ_0^{\tau(\sGamma_m)}} \ \sum_{\gamma\in\CP_\sfi^\jj[\ttQ^{\tau(\sGamma_m)}]} \, \CB_\gamma\, I_{A\,\sfi }\big( W_{A\,\sfi}\big) \ ,
\end{align}
for all $\jj\in\ttQ_0^{\tau(\sGamma_m)}$. 

The same argument used in Section \ref{sec:quiver_variety} shows that the equations $ \mu^\FC=0$ in \eqref{eq:ADHM1} and \eqref{eq:ADHM2} are equivalent to the EJ-term relations
\begin{align}\begin{split}
\mu_{\sfi}^{\FC s}&=\sum_{e\,\in\, \sfs^{-1}(\sfi)} \, \bar B^{s+s_1}_{e}\,B^s_{e}- \sum_{e\,\in\, \sft^{-1}(\sfi)} \, B^{s+s_1}_{e}\,\bar B^s_{e} = 0  \ ,
\\[4pt]
\mu_{e}^{\FC s} &= L^{s+s_1}_{\bar A_2\,\sft(e)}\,B^s_{e}-  B_{e}^{s-s_1+s_2}\,L^s_{\bar A_2\,\sfs(e)} = 0  \ ,
\\[4pt]
\bar\mu_{e}^{\FC s} &= L^{s+s_1}_{\bar A_2\,\sfs(e)}\,\bar B^s_{e}-  \bar B^{s-s_1+s_2}_{e}\,L^s_{\bar A_2\,\sft(e)} = 0 \ ,
\end{split}
\end{align}
for $\sSU(2)\,\times\,$abelian orbifolds, and
\begin{align}
 \mu^\FC&=B\wedge B=0 \ , 
\end{align}
for $\sSU(3)\,\times\,$abelian orbifolds.

Using the stability condition, we can now express the instanton moduli space $\frM_{\vec{\mbf r},\vec k}$, regarded as a quiver variety in the ADHM parametrization, as a noncommutative $\sGamma_m$-Quot scheme
\begin{align}\label{eq:quiver_variety_tetra}
\frM_{\vec{\mbf r},\vec k} \, \simeq \, \vec{\mbf\mu}_m^{\,\FC -1}(0)^{\rm{stable}} \, \big/ \, \sG_{\vec k} \ , 
\end{align}
for the $\sSU(m)\,\times\,$abelian orbifolds, where
\begin{align}
\vec{\mbf\mu}_m^{\,\FC} =\begin{cases} \ 
\big(\mu_\sfi^{\FC s}\,,\,\mu^{\FC s}_{e} \,,\, \bar\mu^{\FC s}_{e} \,,\, \sigma_{A_1\,\sfi}^s\,,\, \sigma_{A_2\,\sfi}^s\big)_{ \sfi\in\ttQ_0^{\sUps_2}\,,\,e\in\ttQ_1^{\sUps_2}\,,\, s\in\widehat \sGamma_{\textsf{ab}}} \quad &\mbox{for}\quad m=2 \ , \\[4pt]
 \ \big(\mu^\FC \,,\, \sigma_{A\,\sfi}^s\big)_{\sfi\in\ttQ_0^{\sUps_3}\,,\, s\in \widehat \sGamma_{\textsf{ab}  }} \quad &\mbox{for}\quad m=3 \ .
\end{cases}
\end{align}
The complex gauge group
\begin{align}
\sG_{\vec k}=\Timesbig_{\sfi\,\in\,\ttQ_0^{\tau(\sGamma_m)}} \, \sGL(k_\sfi,\FC) 
\end{align}
acts on the ADHM data as
\begin{align}
 g\cdot \big(B_{e}\,,\,I_{A\,\sfi}\big) = \big(g_{\sft(e)}\,B_{e}\,g_{\sfs(e)}^{-1} \, , \,  g_\sfi\,I_{A\,\sfi}\big) \ , 
\end{align}
with $g_\sfi\in \sGL(k_\sfi,\FC)$.

\subsection{Equivariant Fixed Points on Quiver Varieties}
\label{sub:fixedptsquiver}

The quiver variety $\frM_{\vec{\mbf r},\vec k}$ has a symmetry group
\begin{align}\label{eq:symmgroup}
\sU(\vec{\mbf r}\,) = \Timesbig_{A\,\in\,\ulfour^\perp} \ \Timesbig_{\sfi\in\ttQ_0^{\tau(\sGamma_m)}}\,\sU(r_{A\,\sfi}) \ ,
\end{align}
acting by framing rotations $I_{A\,\sfi}\longmapsto I_{A\,\sfi}\,h_{A\,\sfi}^{-1}$ with $h_{A\,\sfi}\in\sU(r_{A\,\sfi})$. Its maximal torus can be expressed as 
\begin{align}\label{eq:maxtoriUr}
\sT_{\vec{\mbf \tta}}\, = \Timesbig_{A\,\in\,\ulfour^\perp} \, \sT_{\vec \tta_A}\, = \, \Timesbig_{A\,\in\,\ulfour^\perp} \ \Timesbig_{\sfi\in\ttQ_0^{\tau(\sGamma_m)}} \,  \sT_{\vec \tta_{A\,\sfi}} \ \subset \ \sT^\tau \ ,
\end{align}
where $\vec \tta_{A\,\sfi}=(\tta_{A\,\sfi\,1},\dots,a_{A\,\sfi\,r_{A\,\sfi}})$  are the equivariant parameters of the maximal torus $\sT_{\vec \tta_{A\,\sfi}}\subset\sU(r_{A\,\sfi})$. 

With respect to the action of the maximal torus \eqref{eq:maxtoriUr} on the moduli space, a connected component of the fixed point locus labelled by $\ttF\in\pi_0\big(\frM_{\vec{\mbf r},\vec k}^{\sT^\tau}\big)$ corresponds to a set 
\begin{align}
\ttF=\big(\ttF_{A\,\sfi\,l}\big)_{\stackrel{A\,\in\,\ulfour^\perp\,,\,\sfi\in\ttQ_0^{\tau(\sGamma_m)}}{ l=1,\dots,r_{A\,\sfi}}} \ .
\end{align} 
 The fixed point locus for the action  of the maximal tours $\sT^\tau = \sT_{\vec\tta}\times\sC^\tau$ is the disjoint union 
\begin{align}\label{eq:complex_normal1}
\frM_{\vec{\mbf{r}}, \vec k}^{\sT^\tau} =\bigsqcup_{\ttF\,\in\,\pi_0(\frM_{ \vec{\mbf{r}}, \vec k}^{\sT^\tau})} \ \Timesbig_{A\,\in\,\ulfour^\perp} \ \Timesbig_{\sfi\in\ttQ_0^{\tau(\sGamma_m)}} \ \Timesbig_{l=1}^{r_{A\,\sfi}}  \, \frM_{\ttF_{A\,\sfi\,l} } \ .
\end{align}
We describe these component sets explicitly below.

\subsubsection*{$\mbf{\sSU(2)\,\times\,}$Abelian Orbifolds: Linear Partitions}

Consider the setup of the orbifold group $\sGamma_2=\sUps_2\times \Ab$ and the equivariant gauge theory with maximal torus $\sT_{A_1,A_2}=\sT_{\vec{\mbf \tta}_{A_1}} \times \sT_{\vec{\mbf \tta}_{A_2}}\times\sU(1)^{\times 2}_{\vec\epsilon}$. We may characterise the equivariant fixed points of the torus action on the quiver variety by considering first the action of $\sT_{A_1,A_2}$ on the ADHM data $(B_a, I_A)_{a\in\ulfour\,,\,A\in\ulfour^\perp}$ of Section~\ref{subsec:ADHMtetra} with $\mbf r=\mbf r_{A_1,A_2}= (r_{A_1},r_{A_2},0,0)$. The action of $\sGamma_2$ decomposes the  ADHM data into its irreducible representations labelled by the vertices of the McKay quiver \smash{$\ttQ_0^{\tau(\sGamma_2)}$}. Since the actions of the groups $\sGamma_2$ and $\sT_{A_1,A_2}$ commute by construction, the degeneracy structure of the $\sT_{A_1,A_2}$-fixed point loci, whether they be isolated points or admit continuous deformations,  remain unchanged in the orbifold theory and are parametrized by the same combinatorial data.  For the same reason the compactness results of Appendix~\ref{app:compact} descend to the orbifold projections.

The torus action is given by
\begin{align}
\begin{split}
& \big(B_a\,,\,I_{A}\big)_{a\,\in\,\ulfour\,,\,A\in\{A_1,A_2\}}  \longmapsto  \big(t_1^{-1}\, B_{a} \,,\, t_2^{-1}\,  B_{\bar A_2} \, , \, t_1^2\,t_2\, B_{\bar A_1}\,,\, I_A\, h^{-1}_A \big)_{a\in\{a_1,a_2\}\,,\,A\in\{A_1,A_2\}} \ ,
\end{split}
\end{align}
for $h_A\in\sT_{\vec{\mbf \tta}_A}$. 
The equivariant $\sT_{A_1,A_2}$-fixed point equations are 
\begin{align}\begin{split}\label{eq:fixedpoints}
g\,B_a\,g^{-1}=t_1^{-1}\,B_{a} \ , \quad g\,B_{\bar A_2}\,g^{-1}=t_2^{-1}\,B_{\bar A_2} \ , \quad g\,B_{\bar A_1}\,g^{-1}=t_1^2\,t_2\,B_{\bar A_1} \ , \quad g\,I_A=I_A\,\underline{e}_A \ ,
\end{split}
\end{align}
for $a\in\{a_1,a_2\}$ and $A\in\{A_1,A_2\}$, where $g$ is the image of a homomorphism $\sT_{A_1,A_2}\longrightarrow\sGL(k,\FC)$ and $\underline{e}_A=\diag(e_{A\,1},\dots,e_{A\,r_A})$ with \smash{$e_{A\,l}=\e^{\,\ii\,\tta_{A\,l}}$}.

We use the complex version of the ADHM parametrization to determine the general structure of the connected components of the tetrahedron instanton moduli space labelled by \smash{$\ttF\in\pi_0\big(\frM_{{\mbf r}_{A_1,A_2}, k}^{\sT_{A_1,A_2}}\big)$}. For this, we use Remark~\ref{rem:tetrastability} to decompose $V=V_{A_1} + V_{A_2}$, where $V_A=\FC[B_a,B_b,B_c]\,I_A(W_A)$ and $B_{\bar A}(V_A)=0$. For each $A\in\{A_1,A_2\}$ there are weight decompositions for the $\sT_{\vec\tta_A}$-action given by
\begin{align} \label{eq:VAWATadecomp}
V_A = \bigoplus_{l=1}^{r_A}\,V_{A\,l} \qquad \text{and} \qquad W_A = \bigoplus_{l=1}^{r_A}\,W_{A\,l} \ ,
\end{align}
where $W_{A\,l}$ are one-dimensional $\sT_{\vec{\tta}_A}$-modules. We momentarily focus on the rank one submodules $V_{A\,l}$ and $W_{A\,l}$ for fixed $l\in\{1,\dots,r_A\}$.

By the fixed point equations \eqref{eq:fixedpoints}, the $\sGL(k,\FC)$-transformation $g$ is unique and it induces  weight decompositions of the rank one $\sT_{\vec{\tta}_A}$-modules
\begin{align}
V_{A\,l} = \bigoplus_{i,j\in\RZ}\, V_{A\,l}(i,j) \qquad \text{with} \quad
V_{A\,l}(i,j) = \big\{v\in V_{A\,l} \ \big| \ g(v) = t_1^i\,t_2^j\,e_{A\,l} \, v\big\} \ .
\end{align}
From \eqref{eq:fixedpoints} it follows that 
$B_a\big(V_{A\,l}(i,j)\big) \subset V_{A\,l}(i-1,j) $,
$B_{\bar A_2}\big(V_{A_1\,l}(i,j)\big)\subset V_{A_1\,l}(i,j-1)$ and $B_{\bar A_1}\big(V_{A_2\,l}(i,j)\big)\subset V_{A_2\,l}(i+2,j+1)$,
for $a\in\{a_1,a_2\}$ and $A\in\{A_1,A_2\}$, along with the vanishing images $B_{\bar A_1}\big(V_{A_1\,l}(i,j)\big) = B_{\bar A_2}\big(V_{A_2\,l}(i,j)\big) = 0$. The images
\begin{align} \label{eq:IAWAl}
I_A(W_{A\,l})\subset V_{A\,l}(0,0)
\end{align}
are all one-dimensional subspaces.

For each $l\in\{1,\dots,r\}$, $a\in\{a_1,a_2\}$ and $i,j\in\RZ$, we can summarise this weight data in a pair of diagrams: the $A_1$-diagram
\begin{align}
\begin{split}
\small
\xymatrix{
 V_{A_1\,l}(i-1,j)  \ar[dd]_{B_{\bar A_2}} & & V_{A_1\,l}(i,j) \ar[ll]_{B_a} \ar[dd]^{B_{\bar A_2}}  \\ \\
  V_{A_1\,l}(i-1,j-1) & & V_{A_1\,l}(i,j-1) \ar[ll]^{B_a} 
}
\normalsize
\end{split}
\end{align} 
and the $A_2$-diagram
\begin{align}
\begin{split}
\small
\xymatrix{
& &  V_{A_2\,l}(i+1,j+1) & & \\
V_{A_2\,l}(i-1,j) \ar[urr]^{B_{\bar A_1}} & & & & V_{A_2\,l}(i+2,j+1) \ar[ull]_{B_a}  \\ 
  & & V_{A_2\,l}(i,j) \ar[ull]^{B_a} \ar[urr]_{B_{\bar A_1}}  & & 
}
\normalsize
\end{split}
\end{align}
Both are commutative diagrams by the EJ-term relations $[B_a,B_b]=0$.

For the $A_1$-diagrams, we argue exactly as in~\cite{Nakajima:1999,Cirafici:2008sn}. Since $V_{A_1\,l}$ is spanned by the one-dimensional subspaces \smash{$B_{a_1}^p\,B_{a_2}^m\,B_{\bar A_2}^n\,I_{A_1}(W_{A_1\,l})$} with $p,m,n\in\RZ_{\geq0}$, it follows from \eqref{eq:IAWAl} that \smash{$V_{A_1\,l}(i,j)=0$} if either $i>0$ or $j>0$, while each non-trivial weight space is one-dimensional. The commutativity of the $A_1$-diagrams implies that \smash{$V_{A_1\,l}(i,j)\simeq\FC$} is possible in only three instances: $i=0$ and \smash{$V_{A_1\,l}(i,j+1)\simeq\FC$}, or $j=0$ and \smash{$V_{A_1\,l}(i+1,j)\simeq\FC$}, or both \smash{$V_{A_1\,l}(i+1,j)\simeq\FC$} and \smash{$V_{A_1\,l}(i,j+1)\simeq\FC$}. This yields the box stacking description of a Young diagram $\lambda_{A_1\,l}$: we identify each pair $(i,j)$ for which \smash{$V_{A_1\,l}(i,j)\simeq\FC$} with a box at the corresponding location $(i,j)\in\RZ_{\leq0}^2$. 

By reading off the numbers of boxes in each row, a Young diagram may be identified with a linear partition, that is, a sequence $\lambda=(\lambda_i)_{i\geq1}$ of non-negative integers $\lambda_i\in\RZ_{\geq0}$ satisfying
\begin{align}
\lambda_i \geq\lambda_{i+1} \ .
\end{align}
The total number of boxes in the Young diagram is the size $|\lambda|=\sum_{i\geq1}\,\lambda_i$ of the linear partition. By considering the totality of Young diagrams for $l\in\{1,\dots,r_{A_1}\}$, we obtain an array \smash{$\vec\lambda_{A_1} = (\lambda_{A_1\,1},\dots,\lambda_{A_1\,r_{A_1}})$} of linear partitions of size
\begin{align}
\big|\vec\lambda_{A_1}\big| = \sum_{l=1}^{r_{A_1}}\,|\lambda_{A_1\,l}| = k_{A_1} := \dim V_{A_1} \ .
\end{align}

The argument for the $A_2$-diagrams is analogous. In this case it follows from \eqref{eq:IAWAl} that \smash{$V_{A_2\,l}(i,j)=0$} if either $j<0$ or $i>2j$, while commutativity of the $A_2$-diagrams implies that \smash{$V_{A_2\,l}(i,j)\simeq\FC$} is only possible when either $j=0$ and \smash{$V_{A_2\,l}(i-2,j-1)\simeq\FC$}, or $i=2j$ and \smash{$V_{A_2\,l}(i+1,j)\simeq\FC$}, or both \smash{$V_{A_2\,l}(i-2,j-1)\simeq\FC$}  and \smash{$V_{A_2\,l}(i+1,j)\simeq\FC$}. By identifying each pair $(i,j)$ for which \smash{$V_{A_2\,l}(i,j)\simeq\FC$} with a box at the location $(2j-i,j)\in\RZ^2_{\geq0}$, we obtain an array of linear partitions \smash{$\vec\lambda_{A_2}$} of size \smash{$|\vec\lambda_{A_2}|=k_{A_2}:=\dim V_{A_2}$}. 

For generic values of $t_1$, $t_2$ and $e_{A\,l}$, the sets of weights for the actions of $g$ on $V_{A_1}$ and $V_{A_2}$ are disjoint and therefore $V_{A_1}\cap V_{A_2} = 0$ at the fixed points, i.e. $V=V_{A_1}\oplus V_{A_2}$.
Altogether we have shown that a fixed point labelled by \smash{$\ttF\in\pi_0\big(\frM_{{\mbf r}_{A_1,A_2}, k}^{\sT_{A_1,A_2}}\big)$} corresponds to an array of linear partitions $\vec{\mbf\lambda}= (\vec\lambda_{A_1},\vec\lambda_{A_2})$ whose total size is the instanton number
\begin{align}
k = \big|\vec{\mbf\lambda}\,\big| = \big|\vec\lambda_{A_1}\big| + \big|\vec\lambda_{A_2}\big| = \sum_{A\in\{A_1,A_2\}} \ \sum_{l=1}^{r_A} \, |\lambda_{A\,l}| \ .
\end{align}
However, the correspondence is not bijective: the associations of the same Young diagrams $\lambda_{A\,l}$ can be reached through different combinations of the actions of the linear maps $B_{a_1}$ and $B_{a_2}$. Put differently, the virtual tangent space \smash{$T_{\vec{\mbf\lambda}}^{\rm vir}\frM_{\mbf r_{A_1,A_2},k}$} is not movable, i.e. it contains the trivial \smash{$\sT_{A_1,A_2}$}-representation. This generally allows for continuous deformations and the fixed points \smash{$\vec{\mbf\lambda}$} are not isolated.

\begin{example}
Consider $\sU(1)$ gauge theory with $r_{A_1}=1$ and $r_{A_2}=0$ in the sector of instanton charge $k=2$. A solution of the complex ADHM equations from \eqref{eq:ADHM_tetra} is obtained by taking $B_{\bar A_1}=B_{\bar A_2}=0$ and 
\begin{align}
B_{a_1}={\small \bigg(\begin{matrix}
0&0\\
b_1&0
\end{matrix}\bigg) } \normalsize \quad , \quad B_{a_2}= {\small \bigg(\begin{matrix}
0&0\\
b_2&0
\end{matrix}\bigg) } \normalsize \quad , \quad I_{A_1}= {\small \bigg(\begin{matrix}
I \ \\ 0 \
\end{matrix}\bigg) } \normalsize \ ,
\end{align}
with $b_1,b_2,I\in\FC$.
Up to a $\sU(1)$ phase rotation, and using scaling symmetry to set $\zeta=1$, the D-term equation in \eqref{eq:ADHM_tetra} is then uniquely solved by taking $I=\sqrt{2}$ and $(b_1,b_2)\in\FC^2$ to parametrize the three-sphere $|b_1|^2+|b_2|^2=1$, which after quotienting by the $\sU(1)$ phase leaves the complex projective line $\PP^1$. 

For these ADHM data, the fixed point equations \eqref{eq:fixedpoints} are uniquely solved by the complex gauge transformation  
\begin{align}\label{eq:solution}
g = {\small \bigg(\begin{matrix}
e_{A_1} & 0 \\ 0 & t_1^{-1}\,e_{A_1}
\end{matrix}\bigg)
} \normalsize \ ,
\end{align} 
for all $(b_1,b_2)\in\FC^2$.
The fixed point locus $\frM_\ttF\simeq\PP^1$ is thus compact and consists of non-isolated points, parametrizing the center of the $\sT_{A_1,A_2}$-invariant two-instanton solution in $\FC^2_{A_1,A_2}\subset\FC^3_{A_1}$. It corresponds to the Young diagram
\begin{align}
\begin{split}
\lambda \ = \!\!
\tiny
\begin{ytableau}
\none          & &  
\end{ytableau}
\normalsize
\end{split}
\end{align}
\end{example}

\subsubsection*{$\mbf{\sSU(3)\,\times\,}$Abelian Orbifolds: Integer Points}

We can similarly treat the setup of the orbifold group $\sGamma_3=\sUps_3\times\Ab$ and the equivariant gauge theory with maximal torus $\sT_A = \sT_{\vec{\mbf \tta}_{A}}\times\sU(1)_{\epsilon}$. Let $(B_a,I_{A})_{a\in\ulfour}$ be the ADHM data of Section~\ref{subsec:ADHMtetra} with ${\mbf r}=\mbf r_A=( r_A,0,0,0)$. The torus action is given by
\begin{align}
\big(B_a\,,\,I_A\big)_{a\,\in\,\ulfour} \longmapsto \big(t^{-1}\,B_a\,,\,,  t^3\,B_{\bar A}\,,\,I_{A}\,h^{-1}_{A}\big)_{a\in A} \ ,
\end{align}
for $h_{A}\in\sT_{\vec{\mbf \tta}_{A}}$. The same argument as given in Section \ref{subsec:ADHMC3} shows that $B_{\bar A}=0$. Then the equivariant $\sT_{A}$-fixed point equations are 
\begin{align}\label{eq:su3phi}
g\,B_a\,g^{-1} = t^{-1}\,B_a \qquad \text{and} \qquad g\,I_A=I_A\,\underline{e}_A  \ ,
\end{align}
for $a\in A$, where $g$  denotes the image of a homomorphism $\sT_A\longrightarrow \sGL(k,\FC)$.

By decomposing the vector spaces $V=V_A$ and $W_A$ into rank one $\sT_{\vec\tta_A}$ modules as in \eqref{eq:VAWATadecomp}, the $\sGL(k,\FC)$-transformation $g$ from \eqref{eq:su3phi} induces weight decompositions
\begin{align}
V_l = \bigoplus_{n\in\RZ} \, V_l(n) \qquad \text{with} \quad V_l(n) = \big\{v\in V \ \big| \ g(v) = t^n\,e_{A\,l}\,v\big\} \ ,
\end{align}
such that $B_a\big(V_l(n)\big)\subset V_l(n-1)$ for $a\in A$, while $I_A(W_{A\,l})\subset V_l(0)$ are all one-dimensional subspaces. For each $l\in\{1,\dots,r_A\}$, $a\in A$ and $n\in\RZ$, this data is summarised by the diagram
\begin{align}
V_l(n-1) \xleftarrow{ \ \ B_a \ \ } V_l(n)
\end{align}

Since $V_l$ is spanned by the one-dimensional subspaces $B_a^i\,B_b^j\,B_c^p\,I_A(W_A)$ with $i,j,p\in\RZ_{\geq0}$, it follows that $V_l(n)=0$ if $n>0$ and each non-trivial weight space is one-dimensional. There being no other conditions and no structure, we simply count the number of non-zero subspaces $V_l(n)\simeq\FC$ for $n\in\RZ_{\leq 0}$ to obtain the non-negative integer $\eta_l = \dim V_l$. 

The totality of integer points defines an array $\vec\eta = (\eta_1,\dots,\eta_{r_A})$ of non-negative integers $\eta_l\in\RZ_{\geq0}$ partitioning the instanton number
\begin{align}
k = |\vec\eta\,| = \sum_{l=1}^{r_A}\,\eta_l \ ,
\end{align}
and corresponding to a fixed point labelled by $\ttF\in\pi_0\big(\frM_{\mbf r_A,k}^{\sT_A}\big)$. As previously, one can show that the correspondence is not bijective and the $\sT_A$-fixed points are generally not  isolated, as the associations of the same integer points $\eta_l$ can be reached by different combinations of the actions of the linear maps $B_a$ for $a\in A$.

\subsection{Non-Abelian Orbifold Partition Functions}
\label{sub:nonabpartfns}

We now focus on evaluating the equivariant partition functions for tetrahedron instantons  on non-abelian orbifolds. We have seen in Section~\ref{sub:fixedptsquiver} that, for the $\sSU(m)\,\times\,$abelian orbifolds, the torus-fixed points of the instanton moduli space are not isolated. Moreover, unlike the case of abelian orbifolds, the four-dimensional representation of $\sGamma_m$ defined by the homomorphism $\tau$ does not induce a $\widehat\sGamma_m$-colouring of the combinatorial data parametrizing the fixed points. Consequently, we do not refine the counting variable $\vec\qu$ with respect to the irreducible representations of $\sGamma_m$ when defining the partition functions for $\sSU(m)\,\times\,$abelian orbifolds.

\subsubsection*{$\mbf{\sSU(2)\,\times\,}$Abelian Orbifolds} 

We use the stability condition from Section~\ref{sec:treta_nonabelian} together with the parametrization of the fixed point locus in terms of arrays $\vec{\mbf\lambda}$ of linear partitions from Section~\ref{sub:fixedptsquiver} to decompose the $\sGamma_2$-module $V$, following the analagous treatment for spiked instantons from~\cite{Nekrasov:2016ydq}.
For each \smash{$(\sfi,s),(\sfi',s')\in\ttQ_0^{\tau_{\vec s\,}(\sGamma_2)}$} and $\vec n=(n_1,n_2)\in\RZ^2_{\geq0}$, we define the vector spaces
\begin{align}
V_{A\,\sfi',s'}^{\sfi,s}(\vec n)=\sum_{\substack{\gamma\in \CP_{(\sfi,s)}^{(\sfi',s')}[\ttQ^{\tau_{\vec s\,}(\sGamma_2)}]^A_{n_2}\\ \ttl(\gamma)=n_1+n_2}}\,\CB_\gamma\, I_{A\,\sfi}^s\big( W_{A\,\sfi,s} \big) 
\end{align}
for $A\in\{A_1,A_2\}$, 
where \smash{$\CP_{(\sfi,s)}^{(\sfi',s')}[\ttQ^{\tau_{\vec s\,}(\sGamma_2})]^A_{n_2}$} indicates the set of paths along the quiver $\ttQ^{\tau_{\vec s\,}(\sGamma_2)}$ from $(\sfi,s)$ to $(\sfi',s')$ formed by $n_2$ applications of $L_{\bar A}$ with respect to the notation of \eqref{eq:adhm_tetra_trans}. The complex gauge group $\sG_{\vec k}$ acts  on  \smash{$V_{A\,\sfi,s}^{\sfi',s'}(\vec n)$}  as $\sGL(k_{\sfi,s},\FC)$.

Next we introduce corresponding $\sGamma_2$-equivariant vector bundles
\begin{align}
\begin{split}
\CCV &= \bigoplus_{(\sfi,s)\in\ttQ_0^{\tau_{\vec s\,}(\sGamma_2)}} \, \CCV_{\sfi,s}\otimes\CR^*_{(\sfi,s)} \\[4pt]
:\!\!&=\bigoplus_{(\sfi,s)\in\ttQ_0^{\tau_{\vec s\,}(\sGamma_2)}}\ \bigg( \, \bigoplus_{\vec n\in\RZ_{\geq0}^2} \ \sum_{A\in\{A_1,A_2\}} \  \sum_{(\sfi',s')\in\ttQ_0^{\tau_{\vec s\,}(\sGamma_2)}} \,  \CCV_{A\,\sfi,s}^{\sfi',s'}(\vec n)\bigg) \otimes \CR_{(\sfi,s)}^* \ ,
\end{split}
\end{align}
where $\CR_{(\sfi,s)} = \lambda_\sfi\otimes\rho_s$ for $\sfi\in\ttQ_0^{\sUps_2}$ and $s\in\Abw$ while
\begin{align}
\CCV_{A\,\sfi,s}^{\sfi',s'}(\vec n)=\vec{\mbf\mu}_2^{\,\FC -1}(0)^{ \rm stable} \, \times_{\sG_{\vec k}} \, V_{A\,\sfi,s}^{\sfi',s'}(\vec n) \ ,
\end{align}
together with
\begin{align}\begin{split}
\CCW_A=\bigoplus_{(\sfi,s)\in\ttQ_0^{\tau_{\vec s\,}(\sGamma_2)}}\,\CCW_{A\,\sfi,s}\otimes\CR_{(\sfi,s)}^* \qquad \mbox{with} \quad \CCW_{A\,\sfi,s}= \frM_{\vec{\mbf r}_{A_1,A_2}, \vec k} \, \times \, W_{A\,\sfi,s}\ ,
  \end{split}
\end{align}  
for $A\in\{A_1,A_2\}$. The $\sT_{A_1,A_2}$-action on the moduli space \smash{$\frM_{\vec{\mbf r}_{A_1,A_2},\vec k}$} lifts to $\sT_{A_1,A_2}$-equivariant structures on the bundles $\CCV$ and $\CCW_A$.

Similarly to the case of abelian orbifolds, we need to consider the equivariant version of the cochain complex of vector bundles \eqref{complex_bundle}. Since the subgroups $\sGamma_2^{\tau_{\vec s}}$ and $\sC^{\vec s}$ commute, we can consider the equivariant index bundle as the $\sGamma_2$-invariant part of the index bundle \eqref{eq:index}, regarded as an element of the equivariant K-theory of the moduli space \smash{$\frM_{\vec{\mbf r}_{A_1,A_2},\vec k}$}, by replacing the vector space $Q_4$ with
\begin{align}
Q_4^{\vec s}= t^{-1}_1 \,  (Q_2\otimes\rho_{s_1})+t^{-1}_2\, (\lambda_0\otimes \rho_{-s_1+s_2}) + t_1^2\,t_2\,(\lambda_0\otimes\rho_{-s_{12}}) \ , 
\end{align}
as an element of the representation ring of the group $\sT_{A_1,A_2}\times\sUps_2\times\Ab$, where $t_a=\e^{\,\ii\,\epsilon_a}$ for $a=1,2$.

We express the pullback of the index to the connected component parametrized by the array of Young diagrams $\vec{\mbf\lambda}$ as
\begin{align}\begin{split}\label{eq:character_tetra_orbifold}
& \sqrt{\ch_{\sT_{A_1,A_2}}^{{\sUps_2\times\Ab}}}\big(T^{\rm vir}\frM_{\vec{\mbf r}_{A_1,A_2},\vec k} \big|_{\frM_{\vec{\mbf\lambda}}}\big) \\[4pt]
&  :=
      \ch_{T_{A_1,A_2}}\Big[\CCW_{A_1\,\vec{\mbf\lambda}}^*\otimes \CCV_{\vec{\mbf\lambda}}
      + \CCW_{A_2\,\vec{\mbf\lambda}}^*\otimes \CCV_{\vec{\mbf\lambda}} -\CCV_{\vec{\mbf\lambda}}^*\otimes \CCV_{\vec{\mbf\lambda}} \\
      & \hspace{2cm} +\CCV_{\vec{\mbf\lambda}}^*\otimes \CCV_{\vec{\mbf\lambda}}\,\big(t_1^{-1}\,(Q_2\otimes\rho_{s_1}) +t_2^{-1}\,(\lambda_0\otimes\rho_{s_2-s_1})+t_1^{2}\,t_2\,(\lambda_0\otimes\,\rho_{-s_{12}})\big) \\ 
 & \hspace{2.5cm}
     -\CCV_{\vec{\mbf\lambda}}^*\otimes \CCV_{\vec{\mbf\lambda}}\,\big(t_1^{-2}\,(\lambda_0\otimes\rho_{2s_1})+\,t_1^{-1}\,t_2^{-1}\,(Q_2\otimes\rho_{s_2})\big)  \\
     & \hspace{3cm} -\CCV^*_{\vec{\mbf\lambda}}\otimes \CCW_{A_1\,\vec{\mbf\lambda}}\, t_1^{-2}\,t_2^{-1}\,(\lambda_0\otimes\rho_{s_{12}})-\CCV^*_{\vec{\mbf\lambda}}\otimes \CCW_{A_2\,\vec{\mbf\lambda}}\, t_2\,(\lambda_0\otimes\rho_{s_1-s_2})\Big]^{\sUps_2\times\Ab} \ , \end{split}
\end{align}
where $\CCV_{\vec{\mbf\lambda}} := \CCV\big|_{\frM_{\vec{\mbf\lambda}}}$ and $\CCW_{A\,\vec{\mbf\lambda}} := \CCW_A\big|_{\frM_{\vec{\mbf\lambda}}}$. From \eqref{eq:Q4lambda} it follows that
\begin{align}\begin{split}
& \ch_{\sT_{A_1,A_2}}\big[\CCV^*\otimes \CCV\, (Q_2\otimes\rho_{s'})\big]^{\sUps_2\times\Ab} \\[4pt]
& \hspace{3cm} = \sum_{e\in\ttQ_1^{\sUps_2}} \ \sum_{s\in\Abw} \, \Big( \ch_{\sT_{A_1,A_2}}\big(\CCV^*_{\sft(e),s+s'}\big)\, \ch_{\sT_{A_1,A_2}}\big( \CCV_{\sfs(e),s}\big) \\
& \hspace{8cm} + \ch_{\sT_{A_1,A_2}} \big( \CCV^*_{\sfs(e),s+s'} 	\big)\,	 \ch_{\sT_{A_1,A_2}}\big(\CCV_{\sft(e),s}\big) \Big)  \ , \end{split}
\end{align}
and similarly for the other types of contributions to \eqref{eq:character_tetra_orbifold}.

The pullbacks of the equivariant characteristic classes $	\ch_{\sT_{A_1,A_2}}(\CCV)$ and $\ch_{\sT_{A_1,A_2}}(\CCW)$ to the connected component $\frM_{\vec{\mbf \lambda}}$ decompose into
\begin{align}\begin{split}
& \ch_{\sT_{A_1,A_2}}\big( \CCV_{\vec{\mbf\lambda}}\big) \\[4pt]
&  =\sum_{(\sfi,s)\,,\,(\sfi',s')\in\ttQ_0^{\tau_{\vec s\,}(\sGamma_2)}} \bigg( \sum_{l=1}^{r_{A_1\,\sfi',s'}}\,e_{{A_1\,\sfi',s'\,l}} \ \sum_{\vec p\,\in\lambda_{A_1\,\sfi',s'\,l}} \, t_1^{p_1-1}\,t_2^{p_2-1}\\
& \hspace{6.5cm} \times \ch\big(\CCV^{\sfi',s'}_{A_1\,\sfi,s}(p_1-1,p_2-1)\big|_{\frM_{\lambda_{A_1\,\sfi',s'\,l}}}\big)\otimes\CR^*_{(\sfi,s)} \\
 & \hspace{4cm} +  \sum_{l'=1}^{r_{A_2\,\sfi',s'}}\,e_{A_2\,\sfi',s'\,l'} \ \sum_{\vec p^{\,\prime}\in\lambda_{A_2\,\sfi,s'\,l'}}\, t_1^{p'_1-2p'_2-3}\,t_2^{1-p'_2}\\
 & \hspace{6.5cm} \times \ch\big(\CCV^{\sfi',s'}_{A_2\,\sfi,s}(p'_1-1,p'_2-1)\big|_{\frM_{\lambda_{A_2\,\sfi',s'\,l'}}}\big)\otimes\CR^*_{(\sfi,s)} \bigg)
 \end{split}
 \end{align}
 and
 \begin{align}
 \begin{split}
 \ch_{\sT_{A_1,A_2}}\big(\CCW_{A\,\vec{\mbf\lambda}}\big)&=\sum_{(\sfi,s)\in \ttQ_0^{\tau_{\vec s\,}(\sGamma_2)}} \ \sum_{l=1}^{r_{A\,\sfi,s}}\,e_{A\,\sfi,s\,l}\otimes\CR^*_{(\sfi,s)}\ . \end{split}
\end{align}
From these formulas one may now extract the equivariant top Chern classes and compute the equivariant square root Euler class \smash{$\sqrt{e^{\sGamma_2}_{\sT_{A_1,A_2}}}\big(\CCN^{\rm vir}_{\frM_{\vec{\mbf\lambda}}}\big)$} of the virtual normal bundle, as described in~\eqref{eq:euler_calss}.

Then the full twisted partition function for tetrahedron instantons on $\FC^4/\,\sGamma_2^{\tau_{\vec s}}$, for $\sGamma_2=\sUps_2\times\Ab$ with $\sUps_2$ a finite non-abelian subgroup of $\sSU(2)$ and $\vec{\mbf r}_{A_1,A_2}=(\vec r_{A_1},\vec r_{A_2},\vec 0,\vec 0\,)$, is given by
\begin{align} \label{eq:partfnSU2}
\begin{split}
\hspace{-0.3cm} Z_{[\FC^4\,/\,\sGamma_2]\,\times\,\rmB\sK^{\vec s}}^{\vec{\mbf r}_{A_1,A_2}}( \vec \qu\,;\vec{\mbf \tta},\epsilon_1,\epsilon_2)  =\sum_{\vec k\in\RZ_{\geq 0}^{\#\widehat{\sGamma}_2}}\, \vec \qu^{\, \vec k} \ \sum_{\vec{\mbf\lambda} \in\pi_0(\frM^{\sT_{A_1,A_2}}_{\vec{\mbf r}_{A_1,A_2},\vec k})}\, (-1)^{\texttt{O}^{\sGamma_2}_{\vec{\mbf\lambda}}} \  \int_{[\frM_{\vec{\mbf\lambda}}]^{\rm vir}} \ \frac1{\sqrt{e^{\sGamma_2}_{\sT_{A_1,A_2}}}\big(\CCN^{\rm vir}_{\frM_{\vec{\mbf\lambda}}}\big)} \ ,
\end{split}
\end{align} 
where 
\begin{align}
\vec \qu^{\,\vec k}=\prod_{\sfi\in\ttQ_0^{\sUps_2}} \ \prod_{s\in\Abw}\,\qu_{\sfi,s}^{k_{\sfi,s}} \ .
\end{align}

\begin{remark}[{\bf Sign Factors}]
Comparing the actions of the tori $\sU(1)^{\times 2}_{\vec\epsilon}$ and $\sT_{\vec \varepsilon}$ on $\FC^4$, we see they are related through
\begin{align}\begin{split}\label{eq:breakin:operation}
\varepsilon_1 = \epsilon_1 \quad, \quad
\varepsilon_2 = \epsilon_1 \quad , \quad
\varepsilon_3 = \epsilon_2 \quad , \quad
\varepsilon_4=-\varepsilon_1-\varepsilon_2-\varepsilon_3 = -2\,\epsilon_1-\epsilon_2 \ ,
\end{split}
\end{align}
where $\vec \varepsilon=(\varepsilon_1,\varepsilon_2,\varepsilon_3,\varepsilon_4)$ are the generators of the maximal torus $\sT_{\vec \varepsilon}$ of $\sSU(4)$ and $\vec\epsilon = (\epsilon_1,\epsilon_2)$ are the generators of the centralizer $\sU(1)_{\vec\epsilon}^{\times 2}$ of $\tau_{\vec s\,}(\sGamma_2)$.
From this relation we believe that the sign factor can be evaluated by generalizing the sign factor in \eqref{eq:signfactor} as
 \begin{align}\label{eq:signfactor_orbifold}
\texttt{O}_{\vec{\mbf\lambda}}^{\sGamma_2}=\rk\,  \big( \CCV_{\vec{\mbf\lambda}}^*\otimes \CCV_{\vec{\mbf\lambda}} \ t_1^{-2}\,t_2^{-1}\big)^{\rm fix} \mod 2\ . 
\end{align}
\end{remark}

\begin{example}\label{ex:SU(3)}
Consider the  orbifold $\FC^2/\,\sUps_2\times\FC^2$ where $\sUps_2=\mathbbm{T}^*$ is the binary tetrahedral group of order 24. It has three one-dimensional irreducible representations, $\lambda_0$, $\lambda_3$ and $\lambda_6$, three two-dimensional irreducible representations, $\lambda_1=Q_2$, $\lambda_4$ and $\lambda_5$, and one three-dimensional irreducible representation $\lambda_3$. Given an orientation for the affine Dynkin diagram of type $\mathsf{E}_6$, the McKay quiver $\ttQ^{\mathbbm{T}^*}$ is
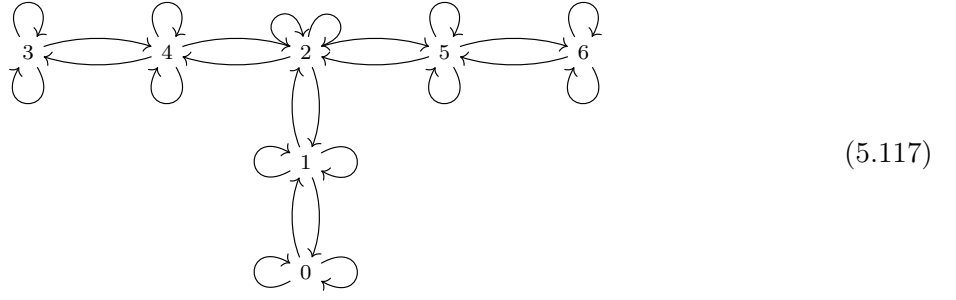
\begin{equation}
{\scriptsize
\begin{tikzcd}
	3 && 4 && 2 && 5 && 6 \\
	\\
	&&&& 1 \\
	\\
	&&&& 0
	\arrow[curve={height=-6pt}, from=5-5, to=3-5]
	\arrow[curve={height=-6pt}, from=3-5, to=1-5]
	\arrow[curve={height=-6pt}, from=3-5, to=5-5]
	\arrow[curve={height=-6pt}, from=1-5, to=3-5]
	\arrow[curve={height=-6pt}, from=1-1, to=1-3]
	\arrow[curve={height=-6pt}, from=1-3, to=1-5]
	\arrow[curve={height=-6pt}, from=1-5, to=1-7]
	\arrow[curve={height=-6pt}, from=1-7, to=1-9]
	\arrow[curve={height=-6pt}, from=1-3, to=1-1]
	\arrow[curve={height=-6pt}, from=1-5, to=1-3]
	\arrow[curve={height=-6pt}, from=1-7, to=1-5]
	\arrow[curve={height=-6pt}, from=1-9, to=1-7]
	\arrow[out=210,in=150,loop,swap,from=5-5,to=5-5]
	\arrow[out=30,in=330,loop,swap,from=5-5,to=5-5]
	\arrow[out=120,in=60,loop,swap,from=1-1,to=1-1]
		\arrow[out=300,in=240,loop,swap,from=1-1,to=1-1]
		\arrow[out=120,in=60,loop,swap,from=1-3,to=1-3]
		\arrow[out=300,in=240,loop,swap,from=1-3,to=1-3]
	\arrow[out=210,in=150,loop,swap,from=3-5,to=3-5]
	\arrow[out=30,in=330,loop,swap,from=3-5,to=3-5]
		\arrow[out=160,in=100,loop,swap,from=1-5,to=1-5]
		\arrow[out=80,in=20,loop,swap,from=1-5,to=1-5]
		\arrow[out=120,in=60,loop,swap,from=1-7,to=1-7]
		\arrow[out=300,in=240,loop,swap,from=1-7,to=1-7]
		\arrow[out=120,in=60,loop,swap,from=1-9,to=1-9]
		\arrow[out=300,in=240,loop,swap,from=1-9,to=1-9]
\end{tikzcd}} \normalsize
\end{equation}

Let  \smash{$\vec{\mbf r}_{A_1,A_2}^{\,0}=\big(r^0_{A\,\sfi}\big)_{A\in\{A_1,A_2\}\,,\,\sfi\in\ttQ^{\mathbbm{T}^*}_0}$} and \smash{$\vec{\mbf r}_{A_1,A_2}^{\,1}=\big(r^1_{A\,\sfi}\big)_{A\in\{A_1,A_2\}\,,\,\sfi\in\ttQ^{\mathbbm{T}^*}_0}$} be two choices of the framing vector $\vec{\mbf r}_{A_1,A_2}$ whose only non-zero entries are $r^i_{A\,\sfi} = r_{A_1\,i}^i = 1$ for $i=0,1$. For these framing vectors the ADHM equations \eqref{eq:ADHM1} are different and inequivalent for any \smash{$\vec k\in\RZ_{\geq 0}^{7}$}. This implies
\begin{align}
\frM_{\vec{\mbf r}_{A_1,A_2}^{\,0},\vec k} \ \not\simeq \ \frM_{\vec{\mbf r}_{A_1,A_2}^{\,1},\vec k} \ .
\end{align}

This example serves to illustrate that, unlike the cases of $\sSU(4)$-instantons on orbifolds studied in \cite{Szabo:2023ixw}, the partition functions for tetrahedron instantons on non-abelian orbifolds of type $\vec{\mbf r} =(r_{A\,\sfi})_{A\in\ulfour^\perp\,,\,\sfi\in\ttQ^\sGamma_0}$ are, in general, not invariant under permutations of the quiver vertices $\sfi\in\ttQ^\sGamma_0$.
\end{example}

\subsubsection*{$\mbf{\sSU(3)\,\times\,}$Abelian Orbifolds} 

Following our stability analysis from Section~\ref{sec:treta_nonabelian} and the parametrization of the fixed point subschemes in terms of arrays of integer points $\vec\eta$ from Section~\ref{sub:fixedptsquiver}, let us introduce vector spaces
\begin{align}
V_{A\,\sfi',s'}^{\sfi,s}(n)=\sum_{\substack{\gamma\in\CP_{(\sfi,s)}^{(\sfi',s')}[\ttQ^{\tau_{\tilde s}(\sGamma_3)}]\\ \ttl(\gamma)=n}}\,\CB_\gamma\, I_{A\,\sfi,s} \big(W_{A\,\sfi,s}\big)
\end{align}
for $(\sfi,s),(\sfi',s') \in \ttQ_0^{\tau_{\tilde s}(\sGamma_3)}$ and $n\in \RZ_{\geq 0}$, where \smash{$\CP_{(\sfi,s)}^{(\sfi',s')}[\ttQ^{\tau_{\tilde s}(\sGamma_3)}]$} is the set of paths along the quiver $\ttQ^{\tau_{\tilde s}(\sGamma_3)}$ from $(\sfi,s)$ to $(\sfi',s')$. Again the complex gauge group $\sG_{\vec k}$ acts  on  $V_{A\,\sfi,s}^{\sfi',s'}(n)$  as \smash{$\sGL(k_{\sfi,s},\FC)$}.
  
We define the $\sGamma_3$-equivariant vector bundles
\begin{align}
\CCV &=\bigoplus_{(\sfi,s)\in\ttQ_0^{\tau_{\tilde s}(\sGamma_3)}}\, \CCV_{\sfi,s}\otimes\CR^*_{(\sfi,s)} := \bigoplus_{(\sfi,s)\in\ttQ_0^{\tau_{\tilde s}(\sGamma_3)}} \ \bigg( \, \bigoplus_{n\in\RZ_{\geq 0}} \ \sum_{(\sfi',s')\in\ttQ_0^{\tau_{\tilde s}(\sGamma_3)}}\,\CCV_{\sfi,s}^{\sfi',s'}(n)\bigg)\otimes \CR^*_{(\sfi,s)} \ ,
\end{align}
where $\CR_{(\sfi,s)} = \lambda_\sfi\otimes\rho_s$ for $\sfi\in\ttQ_0^{\sUps_3}$ and $s\in\Abw$ while
\begin{align}
\CCV_{\sfi,s}^{\sfi',s'}(n)=\vec{\mbf\mu}_3^{\,\FC-1}(0)^{\rm stable} \, \times_{\sG_{\vec k}} \, V_{A\,\sfi,s}^{\sfi',s'}(n) \ ,
\end{align}
along with
\begin{align}
\CCW&=\bigoplus_{(\sfi,s)\in\ttQ_0^{\tau_{\tilde s}(\sGamma_3)}}\, \CCW_{\sfi,s}\otimes\CR^*_{(\sfi,s)} \qquad \mbox{with} \quad \CCW_{\sfi,s}= \frM_{\vec{\mbf r}_A, \vec k} \, \times \, W_{A\,\sfi,s} \ .
\end{align}
Similarly to the $\sSU(2)\times\,$abelian orbifolds,  the $\sT_{A}$-action on the moduli space \smash{$\frM_{\vec{\mbf r}_A,\vec k}$} lifts to $\sT_{A}$-equivariant structures on the bundles $\CCV$ and $\CCW$.

Since $\sGamma_3^{\tau_{\tilde s}}$ and $\sC^{\tilde s}$ commute, the equivariant index bundle is given by the $\sGamma_3$-invariant part of the index \eqref{eq:index} of the cochain complex of vector bundles \eqref{complex_bundle}, by replacing the vector space $Q_4$ with
\begin{align}
Q_4^{\tilde s}= t^{-1} \,  (Q_3 \otimes\rho_{\tilde s})+ t^3\,(\lambda_0\otimes \rho_{-3\tilde s})\ , 
\end{align}
as an element in the representation ring of $\sT_A\times\sUps_3\times\Ab$, where  $t=\e^{\,\ii\,\epsilon}$. 

The pullback of the index to the connected component parametrized by the array of integer points $\vec\eta$ reads
\begin{align}\begin{split}\label{eq:3d_chi_Gamma_ii}
& \sqrt{\ch_{\sT_A}^{\sUps_3\times\Ab}}\big(T^{\rm vir}\frM_{\vec{\mbf r}_A,\vec k}\big|_{\frM_{\vec\eta}}\big) \\[4pt]
& \hspace{1cm}= \ch_{\sT_{A}}\Big[\CCW_{ \vec{ \eta}}^*\otimes \CCV_{\vec{ \eta}} -\CCV_{ \vec{ \eta}}^*\otimes \CCV_{ \vec{ \eta}} +\CCV_{ \vec{\eta }}^*\otimes \CCV_{\vec{\eta }}\,t^{-1}\,(Q_3\otimes\rho_{\tilde s}) + \CCV_{\vec{ \eta}}^*\otimes \CCV_{\vec{ \eta}}\,t^3\,(\lambda_0\otimes\rho_{-3\tilde s}) \\ & \hspace{3cm}
      -\CCV_{ \vec{ \eta}}^*\otimes \CCV_{ \vec{ \eta}}\,t^{-2}\,( Q_3^*\otimes\rho_{2\tilde s})  -\CCV^*_{ \vec{ \eta}}\otimes \CCW_{\vec{ \eta}} \,t^{-3}\,(\lambda_0\otimes \rho_{3\tilde s}) \Big]^{\sUps_3\times\Ab} \ , \end{split}
\end{align}
where $\CCV_{\vec\eta}:=\CCV\big|_{\frM_{\vec\eta}}$ and $\CCW_{\vec\eta}:=\CCW\big|_{\frM_{\vec\eta}}$. From \eqref{eq:Q4lambda} it follows that
\begin{align}
 \ch_{\sT_{A}}\big[\CCV^*\otimes \CCV\, (Q_3\otimes\rho_{s'})\big]^{\sUps_3\times\Ab}=\sum_{e\in\ttQ_1^{\sUps_3}} \ \sum_{s\in\Abw}\,\ch_{\sT_{A}}\big(\CCV^*_{\sft(e),s+s'}\big)\, \ch_{\sT_{A}}\big(\CCV_{\sfs(e),s}\big) \ ,
\end{align}
and similarly for the other types of contributions to \eqref{eq:3d_chi_Gamma_ii}.

The pullbacks of the equivariant Chern characters $\ch_{\sT_{A}}(\CCV)$ and $\ch_{\sT_{A}}(\CCW)$ to the connected component $\frM_{\vec\eta}$ decompose into
\begin{align}
\ch_{\sT_{A}}\big(\CCV_{\vec{ \eta}}\big)&=\sum_{(\sfi,s)\,,\,(\sfi',s')\in\ttQ^{\tau_{\tilde s}(\sGamma_3)}_0} \  \sum_{l=1}^{r_{A\,\sfi,s}}\,e_{A\,\sfi',s'\,l} \ \sum_{p=1}^{\eta_{\sfi',s'\,l}}\, t^{p-1} \, \ch\big(\CCV_{\sfi,s}^{\sfi',s'}(p-1)\big|_{\frM_{\eta_{A\,\sfi',s'\,l}}}\big) \otimes \CR^*_{(\sfi,s)}
\end{align}
and
\begin{align}
\ch_{\sT_{A}}\big(\CCW_{\vec{ \eta}}\big)&=\sum_{(\sfi,s)\in\ttQ^{\tau_{\tilde s}(\sGamma_3)}_0} \ \sum_{l=1}^{r_{A\,\sfi,s}}\, e_{A\,\sfi,s\,l}\otimes \CR^*_{(\sfi,s)} \ .
\end{align}

These formulas may be used to extract the equivariant square root Euler class \smash{$\sqrt{e_{\sT_A}^{\sGamma_3}}\big(\CCN^{\rm vir}_{\frM_{\vec\eta}}\big)$} of the virtual normal bundle using \eqref{eq:euler_calss}, and the full twisted instanton partition function is 
\begin{align} \label{eq:partfnSU3}
Z_{[\FC^4\,/\,\sGamma_3]\,\times\,\rmB\sK^{\tilde s}}^{\vec{\mbf r}_A}( \vec \qu\,;\vec{\tta}_{A}, \epsilon)=\sum_{\vec k\in\RZ_{\geq 0}^{\#\widehat{\sGamma}_3}} \,  \vec\qu^{\,\vec k} \ \sum_{\vec{\eta}\in\pi_0(\frM^{\sT_A}_{\vec{\mbf r}_A,\vec k})} \, (-1)^{\texttt{O}^{\sGamma_3}_{\vec{\eta}}} \ \int_{[\frM_{\vec{\eta}}]^{\rm vir}} \  \frac1{\sqrt{e_{\sT_A}^{\sGamma_3}}\big(\CCN^{\rm vir}_{\frM_{\vec\eta}}\big)} \ ,
\end{align}
where 
\begin{align}
\vec \qu^{\,\vec k}=\prod_{\sfi\in\ttQ_0^{\sUps_3}} \ \prod_{s\in\Abw}\,\qu_{\sfi,s}^{k_{\sfi,s}} \ .
\end{align}

\begin{remark}[{\bf Sign Factors}]
The orbifold by the group $\sGamma_3=\sUps_3\times\Ab$  is equivalent to  the description of instantons on the generally non-effective orbifold $\FC^3/\,\sGamma_3$ from Section~\ref{sec:orb_3d_inst}. By comparing the index in that case with the index \eqref{eq:3d_chi_Gamma_ii}, we find that the sign factor is given by
 \begin{align}
\texttt{O}_{\vec{\eta}}^{\sGamma_3}=\rk\,  \big( \CCV_{\vec{\eta}}^*\otimes \CCV_{\vec{\eta}} \ t^{-3}\big)^{\rm fix} \mod  2\ .
\end{align}
\end{remark}

\subsection{Orbifold Partition Functions from Geometric Crepant Resolutions}
\label{subsec:crepant}

While our constructions from Section~\ref{sub:nonabpartfns} formally solve the problem of computing the partition functions for tetrahedron instantons on non-abelian orbifolds, in practice making the formulas \eqref{eq:partfnSU2} and \eqref{eq:partfnSU3} more explicit like the abelian case is generally still a complicated task due to the remaining integrals over $[\frM_\ttF]^{\rm vir}$ required. We conclude by discussing some classes of non-abelian orbifolds whereby closed formulas for the instanton partition functions can be obtained.

Although our construction of orbifold partition functions for tetrahedron instantons holds generally for any  Calabi–Yau four-orbifold of the types we have discussed, a special role is played  by orbifolds admitting a geometric crepant resolution, which provides a regularization of the orbifold singularities~\cite{crepant}. 
For a finite group $\sGamma$ acting linearly on $\FC^d$, recall that a proper algebraic map $\pi_\sGamma:X_\sGamma\longrightarrow \FC^d/\,\sGamma$ is a crepant resolution if $X_\sGamma$ is smooth and $\pi_\sGamma$ is a birational morphism which preserves the canonical bundles. A necessary but not sufficient condition for the existence of a crepant resolution is that $\sGamma$ is a proper subgroup of $\sSL(d,\FC)$. Crepant resolutions appear in the stringy K\"ahler moduli space of supersymmetric Calabi--Yau orbifolds which have marginal operators that can be used to resolve the singularity.

Resolutions of non-effective orbifolds are discussed in~\cite{Pantev:2005wj}.  While these theories lead to richer BPS spectra at the quotient singularity $\FC^d/\,\sGamma^\tau$, it is not possible to smoothly resolve or deform all singularities within the moduli space of supersymmetric vacua. Henceforth we restrict our considerations to effectively acting orbifold groups, i.e. to subgroups where $\sGamma=\sGamma^\tau\subset\sSL(d,\FC)$.
However, it should be stressed that the absence of a geometric crepant resolution is not a deficiency of the theory: both the twisted orbifold and noncommutative resolutions always exist, and are `desingularizations' in their own contexts.

For $d=2,3$, a crepant resolution is provided by the Hilbert–Chow morphism $\pi_\sGamma$ from the Nakamura $\sGamma$-Hilbert scheme \smash{$X_\sGamma = \Hilb^{\sGamma}(\FC^d) \subset\Hilb^{{\#{\sGamma}}}(\FC^d)$} of $\sGamma$-invariant
zero-dimensional subschemes $Z\subset \FC^d$ of length $\#\sGamma$ whose global sections $H^0(Z,\CO_Z )$ form the regular
representation $\FC[\sGamma]$ of $\sGamma$; this is the moduli space of regular instantons of the rank one orbifold gauge theory. For $d=2$, this crepant resolution is unique and related to an ALE space of type ADE.  For $d=4$, the existence of crepant resolutions for orbifolds of the types $\FC^2/\,\sGamma\times\FC^2$ and $\FC^3/\,\sGamma\times\FC$ is discussed in~\cite{TV2,Cao:2023gvn,Szabo:2023ixw}.

Given a geometric crepant resolution $\pi_\sGamma:X_\sGamma\longrightarrow\FC^4/\,\sGamma$, we now consider its interaction with the orbifold crepant resolution
\begin{align}
\begin{split}
\xymatrix{
\big[\FC^4\,\big/\,\sGamma\big] \ar[dr]_{\pi_{\rm orb}} & & X_\sGamma \ar[dl]^{\pi_\sGamma} \\
 & \FC^4\,\big/\,\sGamma
}
\end{split}
\end{align}
We are interested in those orbifold theories whose partition function on the quotient stack $[\FC^4/\,\sGamma]$ is equivalent to the partition function of the cohomological gauge theory on the crepant resolution $X_\sGamma$ through changes of variables and wall-crossing formulas. This amounts to associating $\sSU(4)$-instantons on $\FC^4/\,\sGamma$ to torsion free sheaves on $X_\sGamma$ along the lines of~\cite{Cirafici:2010bd}, or equivalently fractional D-branes at the orbifold singularity to D-branes wrapping cycles of the exceptional locus of $X_\sGamma$, which underlies an equivalence between the derived categories of coherent sheaves on $[\FC^4/\,\sGamma]$ and $X_\sGamma$. This further restricts the allowed orbifold groups $\sGamma$, and is the physical incarnation of the Donaldson--Thomas crepant resolution correspondence in algebraic geometry~\cite{BCR}.

\subsubsection*{Tetrahedron Instantons on $\mbf{\FC^2_{A_1,A_2}/\,\sGamma\times\FC^2}$}

We start by pointing out that one can always construct a crepant resolution for tetrahedron instantons on orbifolds of the type \smash{$\FC^4/\,\sGamma\simeq\FC_{A_1,A_2}^2/\,\sGamma\times\FC^2$}, where $\sGamma$ is a finite subgroup of $\sSL(2,\FC)$. As previously, this choice of orbifold forces us to consider tetrahedron instantons of type 
\begin{align}
\vec{\mbf r} = \vec{\mbf r}_{A_1,A_2}=\big(\vec r_{A_1\,\sfi},\vec r_{A_2\,\sfi},\vec 0,\vec 0\,\big)_{\sfi\in\widehat\sGamma} \ .
\end{align}

The construction is simple. The ADE singularity $\FC^2_{A_1,A_2}/\,\sGamma$ has a unique minimal crepant resolution  given by the Nakamura $\sGamma$-Hilbert scheme~\cite{ItoNakamura}
\begin{align}
\pi_{A_1,A_2}:X_{A_1,A_2}:=\Hilb^\sGamma\big(\FC^2_{A_1,A_2}\big)\longrightarrow \FC_{A_1,A_2}^2\,\big/\,\sGamma \ .
\end{align}

By regarding $\sGamma$ as a subgroup of $\sSL(3,\FC)$ through the natural embedding $\sSL(2,\FC)\subset\sSL(3,\FC)$, for each stratum \smash{$\FC^3_{A}\subset\FC^3_\triangle$} with $A\in\{A_1,A_2\}$ a crepant resolution of the quotient singularity \smash{$\FC_A^3/\,\sGamma \simeq \FC^2_{A_1,A_2}/\,\sGamma\times\FC_{\bar A}$} is given by letting $X_A = X_{A_1,A_2}\times\FC_{\bar A}\simeq\Hilb^\sGamma(\FC^3_A)$ and defining the two crepant resolutions
\begin{align}
\pi_A := \pi_{A_1,A_2}\times\id_{\FC_{\bar A}} : X_A\longrightarrow \FC_A^3\,\big/\,\sGamma \ ,
\end{align}
for $A\in\{A_1,A_2\}$. Note that $X_{A_1,A_2} \simeq X_{A_1}\cap X_{A_2}$. 

Finally, letting $X_\sGamma = X_{A_1,A_2}\times\FC_{\bar A_1}\times\FC_{\bar A_2}\simeq\Hilb^\sGamma(\FC^4)$, we define
\begin{align}
\pi_\sGamma := \pi_{A_1,A_2}\times\id_{\FC_{\bar A_1}\times\FC_{\bar A_2}} : X_\sGamma\longrightarrow\FC^4\,\big/\,\sGamma \ ,
\end{align}
which by construction is a crepant resolution. The cohomological gauge theory for tetrahedron instantons on the smooth Calabi--Yau four-fold $X_\sGamma$ is now defined by solutions \eqref{eq:tetrasols} of the $\sSU(4)$-instanton equations \eqref{eq:gauge_tetra} on the singular Calabi--Yau three-fold
\begin{align}
X_\triangle = X_{A_1}\cup X_{A_2} \ \subset \ X_\sGamma \ .
\end{align}

On general grounds, the tetrahedron instanton partition function $Z_{X_\sGamma}^{\vec{\mbf r}}$ should follow from a dimensional reduction of the Donaldson--Thomas partition function \smash{$\mbf\CZ_{X_\sGamma}^{ r}$}, similarly to Proposition~\ref{prop:Tetra_reduction}, though we do not yet have available a computation of the latter. For the $\sA_{n-1}$ singularity $\FC^4/\,\RZ_n$, the crepant resolution correspondence of~\cite[Conjecture~5.16]{Cao:2023gvn} relates the $\sU(1)$ orbifold Donaldson--Thomas partition function to \smash{$\mbf\CZ_{X_{\RZ_n}}^{ r=1}$}, where the latter can be computed from the vertex formalism of~\cite{Nekrasov:2023nai}. Extending this correspondence to higher rank and to generic ADE singularities would then enable explicit computation of \eqref{eq:partfnSU2} for any finite subgroup $\sGamma\subset\sSU(2)$. These tasks are beyond the scope of the present paper.

\subsubsection*{Tetrahedron Instantons on $\mbf{\FC^3_A/\,\sGamma\times\FC}$}

A similar construction is available for tetrahedron instantons on orbifolds of the type $\FC^4/\,\sGamma\simeq\FC_A^3/\,\sGamma\times\FC$, where $\sGamma$ is a finite subgroup of $\sSL(3,\FC)$. This restricts to tetrahedron instantons of type
\begin{align}
\vec{\mbf r} = \vec{\mbf r}_A = \big(\vec r_{A\,\sfi},\vec 0,\vec 0,\vec 0\,\big)_{\sfi\in\widehat{\sGamma}} \ .
\end{align}
When $\sGamma$ is a finite subgroup of $\sSO(3)\subset\sSU(3)$, the polyhedral singularity $\FC_A^3/\,\sGamma$ has an irreducible crepant resolution realised by the Nakamura $\sGamma$-Hilbert scheme~\cite{Bridgeland:2001xf}
\begin{align} \label{eq:CA3resolution}
\pi_A:X_A:=\Hilb^\sGamma\big(\FC_A^3\big)\longrightarrow \FC_A^3\,\big/\,\sGamma \ ,
\end{align}
which contracts curves of the exceptional locus of $X_A$ to points in the singular locus of $\FC^3_A/\,\sGamma$.
The Calabi--Yau four-fold $X_\sGamma:=X_A\times\FC_{\bar A}\simeq\Hilb^\sGamma(\FC^4)$ then defines a crepant resolution
\begin{align}
\pi_\sGamma:=\pi_A\times\id_{\FC_{\bar A}} : X_\sGamma\longrightarrow \FC^4\, \big/ \, \sGamma \ .
\end{align} 

We use the crepant resolution \eqref{eq:CA3resolution} to derive a closed formula for the rank one Donaldson--Thomas partition function of the polyhedral singularity in the following way. Let $\sGamma^*\subset\sSU(2)$ be the binary polyhedral group which is the pullback of $\sGamma\subset\sSO(3)$ under the double covering
\begin{align} \label{eq:GammaGamma*}
\begin{split}
\xymatrix{
\sGamma^* \ar[d] \lhook\joinrel\!\ar[r] & \sSU(2) \ar[d] \\
\sGamma \lhook\joinrel\!\ar[r] & \sSO(3)
}
\end{split}
\end{align}
A representation of $\sGamma^*$ which does not descend to a representation of $\sGamma$ is called a \emph{binary representation}~\cite{BoissiereSarti}. Removing the vertices corresponding to binary irreducible representations from the McKay quiver $\ttQ^{\sGamma^*}$ leaves the McKay quiver $\ttQ^\sGamma$.

In addition to the semi-small crepant resolution \eqref{eq:CA3resolution}, for any complex plane $\FC^2\subset\FC^3_A$  there is the minimal resolution of the ADE singularity
\begin{align}
\pi_{\sGamma^*} : X_{\sGamma^*} :=\Hilb^{\sGamma^*}(\FC^2)\longrightarrow\FC^2\,\big/\,\sGamma^* \ .
\end{align}
By the classical McKay correspondence, there are bijections between the nodes $\sfi\neq0$ of the McKay quiver \smash{$\ttQ^{\sGamma^*}$}, the simple roots of the simply-laced Lie algebra $\frg_{\sGamma^*}$ associated to $\sGamma^*$, and the smooth rational curves of the exceptional divisor of $X_{\sGamma^*}$. In particular, denoting by $\sfR^+$ the set of positive roots of $\frg_{\sGamma^*}$, each $\alpha\in \sfR^+$ can be associated to a curve class in $X_{\sGamma^*}$ and there is an injective map
\begin{align}
\vec c_{\sGamma^*}:\sfR^+ \longrightarrow H_2(X_{\sGamma^*},\RZ) \ \simeq \ \RZ^{\#\widehat{\sGamma}^*-1} \ .
\end{align}
The node $\sfi=0$ of $\ttQ^{\sGamma^*}$ corresponds to classes in $H_0(X_{\sGamma^*},\RZ)\simeq\RZ$.

We now construct the map
\begin{align}
\vec c_\sGamma:=f_*\circ \vec c_{\sGamma^*}:\sfR^+\longrightarrow H_2(X_A,\RZ) \ \simeq \ \RZ^{\#\widehat{\sGamma}-1} \ ,
\end{align}
where the morphism $f:X_{\sGamma^*}\longrightarrow X_A$ contracts the curves corresponding to binary irreducible representations of $\sGamma^*$, leaving the exceptional curves of $X_A$~\cite{BoissiereSarti}. The binary irreducible representations of $\sGamma^*$ correspond to the simple roots in $\ker(\vec c_\sGamma)$, and we obtain

\begin{proposition}
For any finite subgroup $\sGamma\subset\sSO(3)$, the partition function for tetrahedron instantons of type $\vec{\mbf r}_A=(1,0,\dots,0)$ on the orbifold $\FC_A^3/\,\sGamma\times\FC$ with holonomy group $\sSU(3)_A$ is given by
\begin{align} \label{eq:Zpolyhedral}
Z_{[\FC_A^3/\,\sGamma]\times\FC}^{\vec{\mbf r}_A=(1,0,\dots,0)}(\vec\qu\,)=M(-\Qu)^{\#\widehat{\sGamma}} \ \prod_{\stackrel{\scriptstyle \alpha \in \sfR^+}{\scriptstyle \vec c_\sGamma(\alpha)\neq \vec 0}} \, \widetilde{M}\big(\vec\qu^{\,\vec c_\sGamma(\alpha)},-\Qu\big)^{-1/2} \ ,
\end{align}
where
\begin{align}
\vec\qu^{\,\vec c_\sGamma(\alpha)} = \prod_{\sfi=1}^{\#\widehat\sGamma-1} \, \qu_\sfi^{c_\sGamma(\alpha)_\sfi} \qquad \mbox{and} \qquad \Qu = \qu_0\,\qu_1\cdots\qu_{\#\widehat\sGamma-1} \ .
\end{align}
\end{proposition}

\begin{proof}
The reduced A-model closed topological string partition function on \smash{$X_{A}=\Hilb^\sGamma(\FC_A^3)$} is evaluated by Bryan and Gholampour in~\cite{Bryan:2008xra} using a localization formula similar to \eqref{eq:partfnSU3} and calculating the integrals over the connected components $[\frM_\ttF]^{\rm vir}$ by decomposition. The all-genus result is given by
\begin{align}
Z^{\rm top}_{X_{A}}(g_s,\vec v\,) = \prod_{\stackrel{\scriptstyle \alpha \in \sfR^+}{\scriptstyle \vec c_\sGamma(\alpha)\neq \vec 0}} \ \prod_{n=1}^\infty \, \big(1-\vec v^{\,\vec c_\sGamma(\alpha)}\,(-\e^{-g_s})^n\big)^{n/2} \ ,
\end{align}
where $g_s$ is the string coupling constant and \smash{$\vec v=(v_\sfi)_{\sfi=1,\dots,\#\widehat{\sGamma}-1}$} are the exponentiated K\"ahler parameters of $X_A$.

By the Gromov--Witten/Donaldson--Thomas correspondence for Calabi--Yau three-folds~\cite{Maulik:2003rzb}, this is related to the instanton partition function of the $\sU(1)$ cohomological gauge theory on $X_A$ with $\sSU(3)_A$ holonomy through
\begin{align}
M(-\qu)^{-\chi(X_A)} \ Z^{r=1}_{X_A}(\qu,\vec v\,)\,\big|_{\qu=\e^{-g_s}} = Z^{\rm top}_{X_{A}}(g_s,\vec v\,) \ .
\end{align}
Here  the variables \smash{$\vec v$} correspond to the basis  of curve classes in $X_A$ and $\qu$ to the topological Euler characteristic \smash{$\chi(X_A)= 1+ \big(\#\widehat{\sGamma}-1\big) = \#\widehat{\sGamma}$} of $X_A$. The former enumerates fractional instantons or D0--D2--D6 states in the type~IIA setting, while the latter counts regular instantons or pure D0--D6 states. There are no compact four-cycles, and hence no D0--D2--D4--D6 states, because $X_A$  is a semi-small resolution, consistently with our assumption of vanishing first Chern class in the cohomological gauge theory on $X_A$.

Finally, the Donaldson--Thomas crepant resolution conjecture for Calabi--Yau three-orbifolds of~\cite{Young:2008hn,Bryan:2010mx} relates the rank one instanton partition functions of  $X_A$ and $[\FC_A^3/\,\sGamma]$ through the wall-crossing formula
\begin{align}
Z^{\vec r=(1,0,\dots0)}_{[\FC^3/\,\sGamma]}(\vec \qu\,)=M(-\Qu)^{-\chi(X_A)} \ Z_{X_A}^{r=1}(\Qu, \vec v\,)\,Z_{X_A}^{r=1}(\Qu, \vec v^{\,-1}) \ ,
\end{align}
with the changes of variables $v_\sfi=\qu_\sfi$ for \smash{$\sfi=1,\dots,\#\widehat{\sGamma}-1 $} and $\Qu=\qu_0\,\qu_1\cdots\qu_{\#\widehat{\sGamma}-1}$, where we defined \smash{$\vec v^{\,-1}=\big(v^{-1}_\sfi\big)_{\sfi=1,\dots,\#\widehat{\sGamma}-1}$}. Putting everything together we arrive at the formula \eqref{eq:Zpolyhedral}.
\end{proof}

\begin{example}
Let $\sGamma=\FS_3\subset\sSO(3)$ be the group of permutations of three elements (cf. Example~\ref{ex:nofixedpoint}). The $\sD_5$ root system has 20 positive roots (described in~\cite{Bryan:2008xra}) and the $\sU(1)$ tetrahedron instanton partition function \eqref{eq:Zpolyhedral} on the Calabi--Yau four-orbifold $\FC_A^3/\,\FS_3\times\FC$ is given by
\begin{align}\begin{split}
Z_{[\FC_A^3/\,\FS_3]\times\FC}^{\vec{\mbf r}_A=(1,0,\dots,0)}(\vec \qu\,) = \frac{M(-\Qu)^{3}}{\widetilde M(\qu_1,-\Qu)\,\widetilde M(\qu_1\,\qu_2,-\Qu)^2\,\widetilde M(\qu_2,-\Qu)^4\,\widetilde M(\qu_2^2,-\Qu)^{\frac12}\,\widetilde M(\qu_1\,\qu_2^2,-\Qu)} \ ,
 \end{split}
\end{align}
where $\Qu = \qu_0\,\qu_1\,\qu_2$.
\end{example}

\section{Discussion}
\label{sec:discussion}

In this paper we generalized the construction of tetrahedron instantons on flat space $\FC^4$ to backgrounds which are Calabi--Yau orbifolds by a (possibly non-effective) action of a finite group $\sGamma$ on $\FC^4$. Tetrahedron instantons arise as bound states of D$1$-branes probing stacks of intersecting D$7$-branes in the presence of a $B$-field in the low energy limit of type~IIB string theory. They can be regarded as a generalization of noncommutative instantons on $\FC^3$, with which they coincide in the rank one case.

To this end we started in Section~\ref{sec:C3orbifold} by defining and developping a theory on three-folds that we interpreted it as the orbifold Donaldson--Thomas theory twisted by a flat gerbe. This leads to a new class of Donaldson--Thomas invariants for both abelian and non-abelian three-orbifolds, even beyond the standard Calabi--Yau case, i.e. for general holonomy $\sU(3)$. As far as we are aware, these more general invariants have not yet been discussed in the algebraic geometry literature, and it would be interesting to confirm our results through rigorous mathematical constructions, which could shed light on novel geometric structures underpinning virtual cycle constructions in these instances. 

More generally, it would be interesting to rigorously derive our constructions of orbifold tetrahedron instanton partition functions from Section~\ref{sec:Tetra_orb}, by combining the considerations of~\cite{Cao:2023gvn,Fasola:2023ypx}. As a first step, this should be possible for the abelian orbifolds considered in Section~\ref{sub:age1orbifolds}, for which we have obtained closed form expressions for the corresponding partition functions. This would nicely extend the harmonious agreement between instanton computations in physics and algebraic geometry considerations, elaborated previously for orbifolds of the magnificent four model in~\cite{Szabo:2023ixw} and~\cite{Cao:2023gvn} respectively.

In Section~\ref{subsec:crepant} we explicitly calculated rank one partition functions for tetrahedron instantons on local polyhedral singularities $\FC^3/\,\sGamma\times\FC$. This was done by calculating the A-model closed topological string partition on the Calabi--Yau three-fold $\Hilb^{\sGamma}(\FC^3)$, applying the Gromov--Witten/Donaldson--Thomas correspondence, and finally linking the result at large radius to the one on the singularity by a wall-crossing formula. The generalizations of this procedure for general rank, as well as establishing a wall-crossing formula for such configurations, are open questions worthy of future investigation.

From a physics perspective, our results can be used to enlarge the dictionary of the BPS/CFT correspondence~\cite{Nekrasov:2015wsu}, whereby the gauge origami partition function of tetrahedron instantons is reproduced by $qq$-characters associated with D6-branes wrapping $\FC^3\subset\FC^4$~\cite{Kimura:2023bxy}. The considerations of this paper allow for a concrete investigation, for the first time, of the correspondence beyond the case of the flat Calabi--Yau four-fold $\FC^4$ to Calabi--Yau orbifolds of $\FC^4$. This generalizes the gauge origami partition function of spiked instantons~\cite{Nekrasov:2016ydq}, whereby the orbifold version of the theory defines $qq$-character operators with and without defects in quiver gauge theories of affine ADE-type. 

It would be interesting to understand the quantum algebraic structures underlying the orbifold theories we have constructed in this paper. As a first step one could derive the free field representations of the abelian orbifold tetrahedron instanton partition functions, expressing our contour integral formula \eqref{eq:abelian_matrix_model} as a vertex operator correlation function after analytic continuation, and thereby generalizing the representations of~\cite{Pomoni:2021hkn,Kimura:2023bxy} in the case of flat space. Particularly our abelian orbifold results of Section~\ref{sub:age1orbifolds}, wherein we have obtained closed formulas for the partition functions, should be useful for elucidating aspects of this correspondence.

\appendix

\section{Finite Subgroups of ${\sSU(3)}$}\label{app:B}

The classification of the finite subgroups of $\sSU(3)$ began with the work of  Blichfeldt over a century ago~\cite{Miller:1916}.
These groups can be divided into five classes, which we describe in this appendix  following~\cite{Ludl:2011gn}.

\begin{notation}
We write $\# g$ for the multiplicative order of an element $g$ of a finite group $\sGamma$, that is, the smallest positive integer $k$ such that $g^k= 1$. The order of  $\sGamma$ is defined to be its cardinality, denoted $\#\sGamma$. Then $\# g = \#\langle g\rangle$, where $\langle g\rangle\subset\sGamma$ is the cyclic subgroup generated by $g\in\sGamma$.

Let $\xi_n=\e^{\,2\pi\,\ii\,/n}$ be a primitive $n$-th root unity, which generates the cyclic group $\RZ_n$ of order $n$. 

We write $\mathbbm{S}_n$ for the symmetric group of degree $n$ with order $n!\,$, and  $\mathbbm{A}_n\subset\FS_n$ for the alternating group of degree $n$ with order $\frac12\,n!\,$.
\end{notation}

\subsubsection*{Abelian Groups}

The possible structures of the finite abelian subgroups of $\sSU(3)$ are strongly constrained by the simple and powerful

\begin{theorem}\label{thm:SU3ab}
Every finite abelian subgroup $\sGamma_{\textsf{ab}}$ of $\sSU(3)$ is isomorphic to a direct product of cyclic groups, 
\begin{align}
\Ab \ \simeq \ \RZ_m\times\RZ_n \ ,
\end{align}
where
\begin{align}\label{eq:ord}
m=\max_{g\in\sGamma_{\textsf{ab}}} \, \# g
\end{align}
and $n$ is a divisor of $m$. 
\end{theorem}

Similarly to the generator \eqref{eq:ZnSU2generator} of $\RZ_n\subset\sSU(2)$, the generators of $\RZ_m\times\RZ_n\subset\sSU(3)$ are
\begin{align}
g_1 =  {\small \begin{pmatrix} \xi_m & 0 & 0\\ 0 &  \xi_m^{-1} & 0 \\ 0&0 &1 \end{pmatrix} } \normalsize \qquad \mbox{and} \qquad g_2  =  {\small \begin{pmatrix} \xi_n & 0 & 0\\ 0 & 1 & 0 \\0 & 0& \xi_n^{-1} \end{pmatrix} } \normalsize \ .
\end{align}

\subsubsection*{Groups with Two-Dimensional Faithful Representations}

For every finite subgroup of $\sSU(2)$ there is an isomorphic finite subgroup of $\sSU(3)$ given by the embedding $\sSU(2)\,\embd\,\sSU(3)$ defined as
\begin{align}
g\in\sSU(2) \ \longmapsto \ \Big({\small \begin{matrix}
g & 0\\
0 & 1
\end{matrix} } \normalsize \Big) \in \sSU(3) \ .
\end{align}

The finite subgroups of $\sSU(2)$ admit an ADE classification and are preimages of the finite subgroups of $\sSO(3)\subset\sSU(3)$ under the double covering
\begin{align}
\sSU(2) \longrightarrow \sSO(3) \ ,
\end{align}
corresponding to the cyclic groups, the dihedral groups, and the platonic groups. The cyclic groups $\RZ_n$ (which correspond to $\sA_{n-1}$ in the ADE classification) and the Klein four-group $\mathbbm{D}_2 = \RZ_2\times\RZ_2$ (which corresponds to $\sD_4$ in the ADE classification) have already appeared in the first class. Of the non-abelian finite $\sSO(3)$-subgroups, only the dihedral groups $\mathbbm{D}_n = \RZ_n\rtimes\RZ_2$ of a regular $n$-gon (which corresponds to $\sD_{n+2}$ in the ADE classification) possess a two-dimensional faithful representation.

More generally, for every finite subgroup of $\sU(2)$ there corresponds a finite subgroup of $\sSU(3)$ under the faithful embedding $\sU(2)\,\embd\,\sSU(3)$ given by
\begin{align}
g\in\sU(2) \ \longmapsto \ \Big({\small \begin{matrix}
g & 0\\
0 & (\det g)^*
\end{matrix} } \normalsize \Big) \in \sSU(3) \ .
\end{align}
Under the isomorphism
\begin{align}
\sU(2) \ \simeq \ \big(\sSU(2)\times\sU(1)\big) \, \big/ \, \RZ_2 \ ,
\end{align}
the finite subgroups of $\sU(2)$ are given by the $\RZ_2$-invariant finite subgroups of the direct product $\sSU(2)\times\sU(1)$. The complete list can be found in~\cite[Theorem~2.2]{Falbel2004}.

\subsubsection*{Groups of Type C}

The groups $\sC_n(a,b)$ of type C are generated by the matrices
\begin{align} \label{eq:typeCgen}
C:={\small \begin{pmatrix}
0&1&0\\
0&0&1\\
1&0&0
\end{pmatrix} } \normalsize \qquad \mbox{and} \qquad C_{a,b}:={\small \begin{pmatrix}
\xi_n^a&0&0\\
0&\xi_n^b&0\\
0&0&\xi_n^{-a-b}
\end{pmatrix} } \normalsize \ ,
\end{align}
where $a,b\in\{0,1,\dots,n-1\}$. 
If we define
\begin{align}
\check C_{a,b}:=C_{b,-a-b} 
\end{align}
then any element of $\sC_n(a,b)$ can be written uniquely as
\begin{align}
C^i\,C_{a,b}^j\,\check C_{a,b}^k \ ,
\end{align}
for some $i,j,k\in \RZ_{\geq 0}$.

It follows from Theorem~\ref{thm:SU3ab} that
\begin{align}
\sC_n(a,b) \ \simeq \ (\RZ_m\times\RZ_p)\rtimes\RZ_3 \ ,
\end{align}
where the $\RZ_3$-subgroup is generated by the permutation matrix $C$, while
\begin{align}
m = {\rm lcm}\big(\#\xi_n^a,\#\xi_n^b\big) \qquad \mbox{and} \qquad p=\min\big\{k\in\{1,\dots,m\} \ \big| \ \check C_{a,b}^k\in\langle C_{a,b}\rangle\big\} \ .
\end{align}
This class contains the tetrahedral group $\mathbbm{T}$ (which corresponds to $\sE_6$ in the ADE classification) isomorphic to $\mathbbm{A}_4\simeq\sC_2(0,1)$. The dimension of an irreducible representation of a group of type~C is either one or three~\cite{Grimus2011FiniteFG}.

\subsubsection*{Groups of Type D}

The groups $\sD_{n,d}(a,b\,;r,s)$ of type D are generated by the matrices \eqref{eq:typeCgen} together with
\begin{align}
D_{r,s}:= {\small \begin{pmatrix}
\xi_d^r &0&0\\
0&0&\xi_d^s\\
0&-\xi_d^{-r-s}&0
\end{pmatrix} } \normalsize \ ,
\end{align}
where $a,b\in\{0,1,\dots,n-1\}$ and $r,s\in\{0,1,\dots,d-1\}$. A different set of generators consists of three diagonal matrices, the matrix $C$ from \eqref{eq:typeCgen}, and the matrix 
\begin{align}
D = {\small \begin{pmatrix}
-1 & 0 & 0 \\
0 & 0 & -1 \\
0 & -1 & 0 
\end{pmatrix} } \normalsize \ .
\end{align}

Theorem~\ref{thm:SU3ab} in this case implies that the groups of type D have the structure~\cite{Grimus2011FiniteFG}
\begin{align}
\sD_{n,d}(a,b\,;r,s) \ \simeq \ (\RZ_m\times \RZ_p)\rtimes\FS_3 \ ,
\end{align}
where $m$ and $p$ are functions of $(n,d)$ as well as of $(a,b\,;r,s)$, while $\FS_3\subset\sSO(3)$ is generated by $C$ and $D$. This class contains the octahedral group $\mathbbm{O}$ (which corresponds to $\sE_7$ in the ADE classification) isomorphic to $\mathbbm{S}_4\simeq\sD_{2,2}(0,1\,;1,1)$. The dimension of an irreducible representation of a group of type~D is either one, two, three or six.

\subsubsection*{Exceptional Groups}

They are eight exceptional finite subgroups of $\sSU(3)$ which do not fit into any of the four previous classes:
\begin{align}\begin{split}
\sSigma(60)  \quad , \quad \sSigma(60)\times \RZ_3 \quad , & \quad \sSigma(168) \quad , \quad \sSigma (168)\times \RZ_3 \ , \\[4pt]
\sSigma(36{\times} 3) \quad , \quad \sSigma(72{\times} 3) \quad , & \quad \sSigma (216{\times} 3) \quad , \quad \sSigma (360{\times} 3) \ . \end{split}
\end{align}  
The groups $\sSigma(n)$ in the first line have order $n$ and contain the two simple groups: the icosahedral group $\mathbbm{I}$ (which corresponds to $\sE_8$ in the ADE classification) isomorphic to $\mathbbm{A}_5\simeq\sSigma(60)$, and the Klein group $\textsf{PSL}(2,7)\simeq\sSigma(168)$. The groups $\sSigma(n{\times}3)$ in the second line have order $3n$ and contain the centre $\RZ_3$ of $\sSU(3)$ (generated by $\xi_3\,\ident_3$), whereas the factor groups $\sSigma(n)=\sSigma(n{\times}3)/\RZ_3$ for $n\in\{36,72,216,360\}$ are \emph{not} subgroups of $\sSU(3)$.

To write the groups in terms of generators, we introduce the matrices
\begin{align}
\begin{split}
E_1:={\small \begin{pmatrix}
1&0&0\\
0&\xi_3&0\\
0&0&\xi_3^2
\end{pmatrix} } \normalsize \quad ,  \quad E_2 & :={\small \begin{pmatrix}
\xi_9^2&0&0\\
0&\xi_9^2&0\\
0&0&\xi_9^2 \, \xi_3
\end{pmatrix} } \normalsize \quad , \quad E_3:={\small \begin{pmatrix}
-1&0&0\\
0&0&-\xi_3\\
0&-\xi_3^2&0
\end{pmatrix} } \normalsize \ , \\[4pt]
 E_4:=\frac{1}{\sqrt 3  \, \ii}\,{\small \begin{pmatrix}
1&1&1\\
1&\xi_3&\xi_3^2\\[0.5ex]
1&\xi_3^2&\xi_3
\end{pmatrix} } \normalsize \quad , \quad E_5 & :=\frac{1}{2}\,{\small \begin{pmatrix}
-1&\mu_-&\mu_+\\
\mu_-&\mu_+&-1\\
\mu_+&-1&\mu_-
\end{pmatrix} } \normalsize \quad , \quad E_6:=\frac{1}{\sqrt 3 \, \ii}\,{\small \begin{pmatrix}
1&1&\xi_3^2\\
1&\xi_3&\xi_3\\
\xi_3&1&\xi_3
\end{pmatrix} } \normalsize \ , \\[4pt]
E_7:={\small \begin{pmatrix}
\xi_7 &0&0\\
0&\xi_7^2&0\\
0&0&\xi_7^4
\end{pmatrix} } \normalsize \quad , \quad \check E_7 & :=\frac{\ii}{\sqrt 7}\,{\small \begin{pmatrix}
\xi_7^4-\xi_7^3 &\xi_7^2-\xi_7^5 & \xi_7-\xi_7^6 \\[0.5ex]
\xi_7^2-\xi_7^5 &\xi_7-\xi_7^6 & \xi_7^4-\xi_7^3 \\[0.5ex]
\xi_7-\xi_7^6 &\xi_7^4-\xi_7^3 & \xi_7^2-\xi_7^5
\end{pmatrix} } \normalsize \ ,
\end{split}
\end{align}
where $\mu_\pm$ are the roots of the quadratic equation $\mu^2+\mu+1 = 0$.
Using the generators \eqref{eq:typeCgen} of the group $\sC_2(0,1)\simeq\mathbbm{A}_4$, they are then generated as
\begin{align}\begin{split}
\sSigma(60)=\langle C_{0,1},C,E_5\rangle & \ , \ \sSigma(168)=\langle E_7,C,\check E_7\rangle \ , \ \sSigma(36{\times} 3)=\langle E_1,C,E_4\rangle \ , \\[4pt]
\sSigma(72{\times} 3)=\langle E_1,C,E_4,E_6\rangle \ & , \ \sSigma(216{\times} 3)=\langle E_1,C,E_4,E_2\rangle \ , \
\sSigma(360{\times} 3)=\langle C_{0,1},C,E_5,E_3\rangle \ .
\end{split}
\end{align}
We recommend~\cite{Bovier:1981wz,Grimus:2010ak} for exhaustive discussions of their properties, group structures and representations.

\section{Moduli Spaces of Torus-Invariant Instantons are Compact}\label{app:compact}

In this appendix we adapt the proof given in \cite[Section 8]{Nekrasov:2016qym}  to show that, for the non-maximal torus actions appearing in this paper, the $\sT$-fixed components of the moduli spaces \smash{$\frM_{\mbf r,k}^\sT$} are compact with respect to the complex analytic topology inherited from the Frobenius norm on the affine space of ADHM data $(B_a,I_A)_{a\in\ulfour\,,\,A\in\ulfour^\perp}$ for tetrahedron instantons. We use the real description of the ADHM parametrization for this purpose.

\subsubsection*{Instantons on $\mbf{\FC_A^3}$}

Consider tetrahedron instantons of type $\mbf r_A=(r_A,0,0,0)$, for some fixed $A\in\ulfour^\perp$. Let
\begin{align}
\sT_A=\sU(1)^{r_{A}}\times\sU(1)
\end{align}
be the torus group whose action on the ADHM data for instantons on $\FC^3_A$ is given by
\begin{align}
 (B_a,I_{A})_{a\in\ulfour} \, \longmapsto \,  \big(t^{-1} \, B_a \,,\, t^3\, B_{\bar A} \,,\, I_{A} \, \exp\ii\,\underline{\tta}_A \big)_{a\in A} \ ,
\end{align}
where $\underline{\tta}_A=\diag(\tta_{A\,1},\dots, \tta_{A\,r_{A}})$  is the generator of $\sU(1)^{r_{A}}\subset\sU(r_A)$, and $t=\e^{\,\ii\,\epsilon}$ where $\epsilon$ is the generator of $\sU(1)\subset\sSU(4)$.
The infinitesimal equivariant $\sT_A$-fixed point equations are 
\begin{align}\begin{split} \label{eq:TAfixedpoints}
[B_a,\phi]=\epsilon\,B_a \ , \quad
[B_{\bar A},\phi]=-3\,\epsilon\,B_{\bar A} \qquad \mbox{and} \qquad 
\phi\, I_{A} = I_{A}\,\underline{\tta}_{A} \ ,
\end{split}
\end{align}
for $a\in A$, where $\phi$  generates a $\sU(k)$ gauge transformation. 

Similarly to Section~\ref{sub:fixedptsquiver}, under this torus action the vector space \smash{$V=V_{A}$} decomposes into weight spaces 
\begin{align}\label{eq:decom_V_tori_1}
V_{A}=\bigoplus_{n\in\RZ} \, V_{A}^{n}=\bigoplus_{n\in\RZ} \ \bigoplus_{l=1}^{r_{A}} \, V_{A\,l}^{n} 
\end{align}
for the action of $\phi\in\sU(V_A)$, whose eigenvalues are given by
\begin{align}
 \phi\big|_{V_{A\,l}^{n}}=(n\,\epsilon +   \tta_{A\,l}) \ \ident_{V_{A\,l}^n} \ .
\end{align}
By \eqref{eq:TAfixedpoints} the operators $B_a$ raise or lower the $\sU(1)$ charge $n\in\RZ$ according to
\begin{align}
B_a : V_A^{n}\longrightarrow V_A^{n-1} \qquad \mbox{and} \qquad B_{\bar A} : V_A^{n} \longrightarrow V_A^{n+3} \ ,
\end{align}
for $a\in A$.

We write
\begin{align}
k=\dim V_{A}= \sum_{n\in\RZ} \, k_{n}:=\sum_{n\in\RZ} \, \dim V_{A}^{n} \ .
\end{align}
Since $V_{A}$ has finite dimension $k$, there exists $N\in\NN$ such that $k_{n}=0$ for all $|n|>N$. 

\begin{proposition}\label{prop:compact1}
The closure of the moduli space \smash{$ \frM_{r_A,k}^{\sT_A}$} of $\sT_A$-invariant noncommutative instantons of rank $r_A$ and charge $k$ on $\FC_A^3$ is compact.
\end{proposition}

\proof
We prove that the ADHM data $(B_a,I_A)_{a\in\ulfour}$, obeying the ADHM equations \eqref{eq:ADHM_tetra} and the $\sT_A$-fixed point equations \eqref{eq:TAfixedpoints}, are bounded in the Frobenius norm. 
From the D-term equation in \eqref{eq:ADHM_tetra}, it is easy to see that the norm of the operator $I_A$ is fixed to
\begin{align}
\| I_A \|_{\textrm{\tiny F}}^2 = \Tr_{V_{A}} \big(I_{A}\,I_{A}^\dagger\big) = \zeta\,k \ .
\end{align}

Let us move on to bound $\sum_{a\in\ulfour}\,\|B_a\|_{\textrm{\tiny F}}^2$. 
By using the decomposition \eqref{eq:decom_V_tori_1}, together with \eqref{eq:BAVA=0} and  cyclicity of the trace, we find
\begin{align}
\sum_{a\,\in\,\ulfour}\, \Tr_{V^n_{A}}\big(B_a^\dagger\, B_a\big) = \sum_{a\in A}\, \Tr_{V^n_{A}}\big(B_a^\dagger\, B_a\big) \ ,
\end{align}
and
\begin{align}
\begin{split}
\sum_{a\,\in\,\ulfour}\, \Tr_{V^n_{A}}\big(B_a\,B_a^\dagger\big) &= \sum_{a\in A}\, \Tr_{V^{n+1}_{A}}\big(B_a^\dagger\, B_a\big) + \Tr_{V^{n-3}_{A}}\big(B_{\bar A}^\dagger\, B_{\bar A}\big)\\[4pt]
&= \sum_{a\in A}\, \Tr_{V^{n+1}_{A}}\big(B_a^\dagger\, B_a\big)= \ \sum_{a\in A}\, \Tr_{V^n_{A}}\big(B_a\, B_a^\dagger\big) \ . \end{split} 
\end{align}
Then the D-term equation in \eqref{eq:ADHM_tetra} and the decomposition \eqref{eq:decom_V_tori_1} imply
\begin{align}\begin{split}\label{eq:TrVAn}
\sum_{a\in A}\, \Tr_{V^n_{A}}\big(B_a\,B_a^\dagger\big)+ \Tr_{V^n_{A}}\big(I_{A}\,I_{A}^\dagger\big) &= \zeta\, k_n + \sum_{a\in A}\, \Tr_{V^n_{A}}\big(B_a^\dagger\, B_a\big) \\[4pt]
&= \zeta\, k_n + \sum_{a\in A}\, \Tr_{V^{n-1}_{A}}\big(B_a\,B_a^\dagger\big) \ ,
\end{split}
\end{align}
where in the last equality we used cyclicity of the trace. 

We now introduce
\begin{align}
\varDelta_n &:= \frac{1}{\zeta}\,\Tr_{V_{A}^n}\Big(\sum_{a\in A}\, B_a\,B_a^\dagger +I_{A}\,I_{A}^\dagger \Big)
\end{align}
and
\begin{align}
\varDelta &:= \frac{1}{\zeta}\,\Tr_{V_{A}}\Big(\sum_{a\in A}\,B_a\,B_a^\dagger +I_{A}\,I_{A}^\dagger \Big)=\sum_{n\in \RZ}\, \varDelta_n \ .
\end{align}
Using \eqref{eq:TrVAn} we can write
\begin{align}
\varDelta_n=k_n+ \frac{1}{\zeta}\,\sum_{a\in A}\,\Tr_{V_{A}^{n-1}}\big(B_a\, B_a^\dagger\big) \ \leq \ k_n + \varDelta_{n-1} \ .
\end{align}
By iteration we get
\begin{align}
\varDelta_n\leq k_n+k_{n-1}+\cdots+k_{-N} \, \leq k^2 \ ,
\end{align}
where we used $N\leq\frac{k-1}2$ and $k_n\leq k$ for any $n\in \RZ$.

Therefore, since $\varDelta$ is a sum of at most $2N+1\leq k$ terms, we arrive at the bound
\begin{align}
\sum_{a\in\ulfour}\,\| B_a \|_{\textrm{\tiny F}}^2 = \sum_{a\in\ulfour}\,\Tr_{V_{A}}\big(B_a\,B_a^\dagger\big)\, = \sum_{a\in A}\, \Tr_{V_{A}} \big(B_a\,B_a^\dagger\big) \, \leq \, \zeta\,\varDelta \, \leq \, \zeta \,k^3 \ ,
\end{align}
hence $B_a$ for $a\in\ulfour$ are also bounded in the Frobenius norm. 
\endproof

\begin{remark}
To prove that \smash{$ \frM_{r_A,k}^{\sT_A}$} is closed, and hence is itself compact, one would need to find sharper bounds than those given in the proof of Proposition~\ref{prop:compact1} which are saturated by the ADHM variables. While we believe this is possible to do, we do not pursue it in the present paper.
\end{remark}

\subsubsection*{Generalized Folded Instantons}

We now turn our attention to tetrahedron instantons of type $\mbf r_{A_1,A_2}=(r_{A_1},r_{A_2},0,0)$, for fixed distinct $A_1,A_2\in\ulfour^\perp$. We write $A_1\cap A_2=(a_1\,a_2)$, with $a_1,a_2\in\ulfour\,$. With notation as above, consider the action of the torus group
\begin{align}
\sT_{A_1,A_2}=\sU(1)^{r_{A_1}}\times\sU(1)^{r_{A_2}}\times\sU(1)^{\times 2}
\end{align}
on the ADHM data $(B_a,I_{A_1},I_{A_2})_{a\in\ulfour}$ given by
\begin{align}
\begin{split}
 (B_a,I_{A_1},I_{A_2})_{a\in\ulfour} \, \longmapsto \,  \big(t_1^{-1}\, B_{a_1} \,,\, t_1^{-1}\,	B_{a_2} \,,\,
 t_2^{-1}\,  B_{\bar A_2} \,,\ & t_1^2\,t_2\,  B_{\bar A_1} \,,\, \\
 & I_{A_1}\,\exp\ii\,\underline{\tta}_{A_1} \,,\, I_{A_2}\,\exp\ii\,\underline{\tta}_{A_2} \big) \ ,
\end{split}
\end{align}
where $(\underline{\tta}_{A_1},\underline{\tta}_{A_2})$  are the generators of $\sU(1)^{r_{A_1}}\times\sU(1)^{r_{A_2}}\subset\sU(\mbf r_{A_1,A_2})$,  and $(t_1,t_2) = (\e^{\,\ii\,\epsilon_1},\e^{\,\ii\,\epsilon_2})$ where $(\epsilon_1,\epsilon_2)$ are the generators of $\sU(1)^{\times 2}\subset\sSU(4)$.
The infinitesimal equivariant $\sT_{A_1,A_2}$-fixed point equations are 
\begin{align}\begin{split}\label{eq:fixedpoin_Gamma2}
[B_{a_1},\phi]=\epsilon_1\,B_{a_1} \quad , \quad[B_{a_2},\phi]=\epsilon_1\,B_{a_2} \quad & , \quad [B_{\bar A_2},\phi]=\epsilon_2\,B_{\bar A_2} \ , \\[4pt] [B_{\bar A_1},\phi]=-(2\,\epsilon_1+\epsilon_2)\,B_{\bar A_1} \quad , \quad
\phi\, I_{A_1} = I_{A_1}\,\underline{\tta}_{A_1} \quad & , \quad \phi\, I_{A_2} = I_{A_2}\,\underline{\tta}_{A_2} \ ,
\end{split}
\end{align}
where $\phi$ generates a $\sU(k)$ gauge transformation. 

With $V=V_{A_1} + V_{A_2}$, from \eqref{eq:fixedpoin_Gamma2} it follows that $\phi(V_{A_1})\subseteq V_{A_1}$ and $\phi(V_{A_2})\subseteq V_{A_2}$. Hence the vector spaces $V_{A_1}$ and $V_{A_2}$ decompose under this torus action into weight spaces as
\begin{align}\label{eq:decom_V_tori}
V_A=\bigoplus_{i,j\in\RZ}\, V_A^{i,j}=\bigoplus_{i,j\in\RZ} \ \bigoplus_{l=1}^{r_A}\, V_{A\,l}^{i,j} \qquad \mbox{for} \quad A\in\{A_1,A_2\} \ , 
\end{align}
with respective eigenvalues of $\phi$ given by 
\begin{align}
 \phi\big|_{V_{A\,l}^{i,j}}= (i\,\epsilon_1 + j\,\epsilon_2  +  \tta_{A\,l}) \ \ident_{V_{A\,l}^{i,j}}\ .
\end{align}
By \eqref{eq:fixedpoin_Gamma2} the operators $B_a$ raise/lower the $\sU(1)$ charges $i$ and $j$ according to
\begin{align}
B_{a_1},B_{a_2} : V_A^{i,j}\longrightarrow V_A^{i-1,j} \quad, \quad B_{\bar A_2} : V_A^{i,j}\longrightarrow V_A^{i,j-1} \quad , \quad B_{\bar A_1} : V_A^{i,j} \longrightarrow V_A^{i+2,j+1} \ .
\end{align}

For generic values of the equivariant parameters $\epsilon_1$, $\epsilon_2$ and $\tta_{A\,l}$, the sets of eigenvalues of $\phi$ on $V_{A_1}$ and $V_{A_2}$ are disjoint, so the spaces $V_{A_1}$ and $V_{A_2}$ have trivial intersection and $V = V_{A_1}\oplus V_{A_2}$ at the fixed points. We write
\begin{align}\label{eq:decomp_ijk}
k_A=\dim V_A=\sum_{i,j\in\RZ}\, k_{A\,i,j}:=\sum_{i,j\in\RZ}\,\dim V_A^{i,j} \ .
\end{align}
Then
\begin{align}
k = \dim V = k_{A_1} + k_{A_2} \ .
\end{align}
As before, since $V_A$ is finite-dimensional, there exists $N_A\in\NN$ such that $k_{A\,i,j}=0$ for all $i,j\in\RZ$ satisfying $|i|+|j|>N_A$.

\begin{proposition}\label{prop:compact2}
The closure of the moduli space \smash{$ \frM_{\mbf r_{A_1,A_2},k}^{\sT_{A_1,A_2}}$} of $\sT_{A_1,A_2}$-invariant tetrahedron instantons of type \smash{$\mbf r_{A_1,A_2}$} and charge $k$ is compact. 
\end{proposition}

\proof
The proof is similar to the proof of Proposition~\ref{prop:compact1}, so we will be relatively brief. From the D-term equation in \eqref{eq:ADHM_tetra} it follows that
\begin{align}
\| I_{A_1} \|_{\textrm{\tiny F}}^2 + \| I_{A_2} \|_{\textrm{\tiny F}}^2 = \Tr_V\big(I_{A_1}\,I_{A_1}^\dagger\big)+ \Tr_V\big(I_{A_2}\,I_{A_2}^\dagger\big)= \zeta\,k  \ ,
\end{align}
hence $I_A$ for $A\in\{A_1,A_2\}$ are bounded.

From \eqref{eq:BAVA=0} we obtain
\begin{align}\begin{split}
\sum_{a\,\in\,\ulfour}\, \Tr_{V^{i,j}_{A}}\big(B_a\,B_a^\dagger\big) \ &= \ \Tr_{V^{i+1,j}_{A}}\big(B_{a_1}^\dagger\, B_{a_1}\big) + \Tr_{V^{i+1,j}_{A}}\big(B_{a_2}^\dagger\, B_{a_2}\big) \\
& \qquad \ + \Tr_{V^{i,j+1}_{A}}\big(B_{\bar A_2}^\dagger\, B_{\bar A_2}\big) + \Tr_{V^{i-2,j-1}_{A}}\big(B_{\bar A_1}^\dagger\, B_{\bar A_1}\big) \\[4pt]
 &= \ \sum_{a\in A}\, \Tr_{V^{i,j}_{A}}\big(B_a\, B_a^\dagger\big)\ , \end{split}
 \end{align}
and
\begin{align}
\sum_{a\,\in\,\ulfour}\, \Tr_{V^{i,j}_{A}}\big(B_a^\dagger\, B_a\big) \ &=\ \sum_{a\in A}\, \Tr_{V^{i,j}_{A}}\big(B_a^\dagger\, B_a\big) \ , 
\end{align}
for $A\in\{A_1,A_2\}$. We can then use the D-term equation of \eqref{eq:ADHM_tetra} to write
\begin{align} \label{eq:TrVAij}
\sum_{a\in A}\,\Tr_{V_A^{i,j}}\big(B_a\,B_a^\dagger\big) + \Tr_{V_A^{i,j}}\big(I_{A_1}\,I_{A_1}^\dagger\big) + \Tr_{V_A^{i,j}}\big(I_{A_2}\,I_{A_2}^\dagger\big) = \zeta \, k_{A\,i,j}\,+\,\sum_{a\in A}\,\Tr_{V_A^{i,j}}\big(B_a^\dagger\, B_a\big) \ .
\end{align}

We now introduce
\begin{align}
\begin{split}
\varDelta_{A_1,n} &:= \frac{1}{\zeta}\,\sum_{i-3j=n}\,\Tr_{V_{A_1}^{i,j}}\Big(\sum_{a\in A_1}\,B_a\,B_a^\dagger +I_{A_1}\,I_{A_1}^\dagger+I_{A_2}\,I_{A_2}^\dagger \Big) \ , \\[4pt]
\varDelta_{A_2,n} &:= \frac{1}{\zeta}\,\sum_{i+j=n}\,\Tr_{V_{A_2}^{i,j}}\Big(\sum_{a\in A_2}\,B_a\,B_a^\dagger +I_{A_1}\,I_{A_1}^\dagger+I_{A_2}\,I_{A_2}^\dagger \Big) \ ,
\end{split}
\end{align}
along with
\begin{align}
\varDelta_A &:= \frac{1}{\zeta}\,\Tr_{V_A}\Big(\sum_{a\in A}\,B_a\,B_a^\dagger +I_{A_1}\,I_{A_1}^\dagger+I_{A_2}\,I_{A_2}^\dagger \Big)=\sum_{n\in \RZ} \, \varDelta_{A\,n}
\end{align}
for $A\in\{A_1,A_2\}$. 

Using \eqref{eq:TrVAij} and cyclicity of the trace we generate the inequalities
\begin{align}\begin{split}
\varDelta_{A_1,n}&= \sum_{i-3j=n}\, k_{A_1,i,j}+\frac{1}{\zeta}\, \sum_{i-3j=n} \ \sum_{a\in A_1}\,\Tr_{V_{A_1}^{i,j}}\big(B^\dagger_a\, B_a\big)  \\[4pt]
&\hspace{1cm} =\sum_{i-3j=n}\, k_{A_1,i,j}+ \frac{1}{\zeta}\,\sum_{i-3j=n-1} \ \sum_{a\in A_1}\,\Tr_{V_{A_1}^{i,j}}\big(B_a\, B_a^\dagger\big)\\[4pt]
&\hspace{2cm} \leq  \sum_{i-3j=n}\,k_{A_1,i,j}+\varDelta_{{A_1},n-1} \ ,
\end{split}
\end{align}
and
\begin{align}\begin{split}
\varDelta_{A_2,n}&= \sum_{i+j=n}\,k_{A_2,i,j}+\frac{1}{\zeta}\, \sum_{i+j=n} \ \sum_{a\in A_2}\,\Tr_{V_{A_2}^{i,j}}\big(B^\dagger_a\, B_a\big)  \\[4pt]
&\hspace{1cm} =\sum_{i+j=n}\,k_{A_2,i,j}+\frac{1}{\zeta}\, \sum_{i+j=n-1} \ \sum_{a\in A_2}\,\Tr_{V_{A_2}^{i,j}}\big(B_a\, B_a^\dagger\big)\\[4pt]
&\hspace{2cm} \leq \sum_{i+j=n}\,k_{A_2,i,j}+\varDelta_{{A_2},n-1} \ .
\end{split}
\end{align}

By iterating these inequalities and using the bounds
\begin{align}
\sum_{i-3j=n}\, k_{A_1,i,j} \ \leq \ k_{A_1} \qquad \mbox{and} \qquad \sum_{i+j=n}\, k_{A_2,i,j} \ \leq \ k_{A_2} \ ,
\end{align}
we get 
\begin{align}
\varDelta_{A\,n} \leq k_A^2
\end{align} for $A\in\{A_1,A_2\}$. Therefore
\begin{align}
\begin{split}
\sum_{a\,\in\,\ulfour}\,\| B_a \|_{\textrm{\tiny F}}^2 &= \sum_{a\,\in\,\ulfour}\,\Tr_V\big(B_a\,B_a^\dagger\big)\\[4pt]
& = \sum_{A\in\{A_1,A_2\}} \ \sum_{a\in A}\,\Tr_{V_A} \big(B_a\, B_a^\dagger\big) \, \leq \, \zeta\,\big(\varDelta_{A_1}+\varDelta_{A_2}\big) \,\leq \, \zeta\, \big(k_{A_1}^3 +k_{A_2}^3\big) \ ,
\end{split}
\end{align}
which establishes the required boundedness of the ADHM variables.
\endproof

\subsection*{Conflict of Interest Statement}

All authors have no conflicts of interest.

\subsection*{Data Availability Statement}

No additional research data beyond the data presented and cited in this work are needed to validate the research findings in this work.

\bibliographystyle{ourstyle}
\bibliography{Tetrahedron-bibliography}
\end{document}